\documentclass[11pt]{article}
\usepackage{fullpage}

\usepackage{comment,amsfonts,amsmath,amsthm,graphicx,
            algorithmic,mathtools,tablefootnote}

\usepackage[colorlinks,linkcolor=blue,filecolor=blue,citecolor=blue,urlcolor
=blue,pdfstartview=FitH,linktocpage=true]{hyperref}

\newtheorem*{theorem*}{Theorem}
\newtheorem{theorem}{Theorem}
\newtheorem*{definition*}{Definition}
\newtheorem{definition}{Definition}
\newtheorem*{lemma*}{Lemma}
\newtheorem{lemma}{Lemma}
\newtheorem*{claim*}{Claim}
\newtheorem{claim}{Claim}
\newtheorem*{corollary*}{Corollary}
\newtheorem{corollary}{Corollary}
\newtheorem{remark}{Remark}
\newtheorem{fact}{Fact}
\newtheorem*{fact*}{Fact}
\newtheorem{invariant}{Invariant}

\begin{document}
\title{Path-Reporting Distance Oracles with Logarithmic Stretch and Size $O(n\log\log n)$}
\author{Michael Elkin\footnote{Ben-Gurion University of the Negev, Beer-Sheva, Israel. elkinm@bgu.ac.il}, Idan Shabat\footnote{Ben-Gurion University of the Negev, Beer-Sheva, Israel. shabati@post.bgu.ac.il}{\let\thefootnote\relax\footnote{{Supported by Lynn and William Frankel Center for Computer Sciences and ISF grant 2344/19.}}}
}
\date{}
\maketitle

\pagenumbering{gobble}








\begin{abstract}
Given an $n$-vertex undirected graph $G=(V,E,w)$, and a parameter $k\geq1$, a \textbf{path-reporting distance oracle} (or \textbf{PRDO}) is a data structure of size $S(n,k)$, that given a \textit{query} $(u,v)\in V^2$, returns an $f(k)$-approximate shortest $u-v$ path $P$ in $G$ within time $q(k)+O(|P|)$. Here $S(n,k)$,  $f(k)$ and $q(k)$ are arbitrary (hopefully slowly-growing) functions. A distance oracle that only returns an approximate estimate $\hat{d}(u,v)$ of the distance $d_G(u,v)$ between the queried vertices is called a \textbf{non-path-reporting distance oracle}.

A landmark PRDO due to Thorup and Zwick \cite{TZ01} has $S(n,k)=O(k\cdot n^{1+\frac{1}{k}})$, $f(k)=2k-1$ and $q(k)=O(k)$. Wulff-Nilsen \cite{WN13} devised an improved query algorithm for this oracle with $q(k)=O(\log k)$. The size of this oracle is $\Omega(n\log n)$ for all $k$. Elkin and Pettie \cite{EP15} devised a PRDO with $S(n,k)=O(\log k\cdot n^{1+\frac{1}{k}})$, $f(k)=O(k^{\log_{4/3}7})$ and $q(k)=O(\log k)$. Neiman and Shabat \cite{NS23} recently devised an improved PRDO with $S(n,k)=O(n^{1+\frac{1}{k}})$, $f(k)=O(k^{\log_{4/3}4})$ and $q(k)=O(\log k)$. These oracles (of \cite{EP15,NS23}) can be much sparser than $O(n\log n)$ (the oracle of \cite{NS23} can have \textit{linear} size), but their stretch is \textit{polynomially larger} than the optimal bound of $2k-1$. On the other hand, a long line of \textit{non}-path-reporting distance oracles culminated in a celebrated result by Chechik \cite{C15}, in which $S(n,k)=O(n^{1+\frac{1}{k}})$, $f(k)=2k-1$ and $q(k)=O(1)$.

In this paper we make a dramatic progress in bridging the gap between path-reporting and non-path-reporting distance oracles. In particular, we devise a PRDO with size $S(n,k)=O\left(\left\lceil\frac{k\cdot\log\log n}{\log n}\right\rceil\cdot n^{1+\frac{1}{k}}\right)$, stretch $f(k)=O(k)$ and query time $q(k)=O\left(\log\left\lceil\frac{k\cdot\log\log n}{\log n}\right\rceil\right)$. As $\left\lceil\frac{k\cdot\log\log n}{\log n}\right\rceil=O(\log k)$ for $k\leq\log n$, its size is always at most $O(\log k\cdot n^{1+\frac{1}{k}})$, and its query time is $O(\log\log k)$. Moreover, for $k=O\left(\frac{\log n}{\log\log n}\right)$, we have $\left\lceil\frac{k\cdot\log\log n}{\log n}\right\rceil=O(1)$, i.e., $S(n,k)=O(n^{1+\frac{1}{k}})$, $f(k)=O(k)$, and $q(k)=O(1)$. For $k=\Theta(\log n)$, our oracle has size $O(n\log\log n)$, stretch $O(\log n)$ and query time $O(\log^{(3)}n)$. We can also have linear size $O(n)$, stretch $O(\log n\cdot\log\log n)$ and query time $O(\log^{(3)}n)$.

These trade-offs exhibit polynomial improvement in stretch over the PRDOs of \cite{EP15,NS23}. For $k=\Omega\left(\frac{\log n}{\log\log n}\right)$, our trade-offs also strictly improve the long-standing bounds of \cite{TZ01,WN13}.

Our results on PRDOs are based on novel constructions of \textit{approximate distance preservers}, that we devise in this paper. Specifically, we show that for any $\epsilon>0$, any $k=1,2,...$, and any graph $G=(V,E,w)$ and a collection $\mathcal{P}$ of $p$ vertex pairs, there exists a $(1+\epsilon)$-approximate preserver for $G,\mathcal{P}$ with $O(\gamma(\epsilon,k)\cdot p+n\log k+n^{1+\frac{1}{k}})$ edges, where $\gamma(\epsilon,k)=\left(\frac{\log k}{\epsilon}\right)^{O(\log k)}$. These new preservers are significantly sparser than the previous state-of-the-art approximate preservers due to Kogan and Parter \cite{KP22}.
\end{abstract}

\newpage
\tableofcontents
\newpage

\pagenumbering{arabic}
\setcounter{page}{1}
\section{Introduction}

\subsection{Distance Oracles}

\subsubsection{Background}

Computing shortest paths and distances, exact and approximate ones, is a fundamental and extremely well-studied algorithmic problem \cite{Sei92,AGM97,ACIM99,Zwi99,C00,DHZ00,E01,Zwi01,CZ01,Zwi02,Zwi06,BK10,RS11,EN16}. A central topic in the literature on computing shortest paths and distances is the study of \textit{distance oracles} \cite{TZ01,SVY09,WN13,PR14,C14,ENWN16,C15,EP15,RT23}.

Given an undirected weighted graph $G=(V,E)$ with weights $w(e)\geq0$ on the edges, and a pair of vertices $(u,v)\in V^2$, let $d_G(u,v)$ denote the \textbf{distance} between them in $G$, i.e., the weight $w(P_{u,v})$ of the shortest $u-v$ path $P_{u,v}$. When the graph $G$ is clear from the context, we write $d(u,v)=d_G(u,v)$.

A \textbf{path-reporting distance oracle} (shortly, \textbf{PRDO}) is a data structure that given a \textit{query} vertex pair $(u,v)\in V^2$ returns an \textit{approximately} shortest $u-v$ path $P'_{u,v}$ and its length $\hat{d}(u,v)=w(P'_{u,v})$. The latter value is also called the \textbf{distance estimate}. An oracle that returns only distance estimates (as opposed to paths and estimates) is called a \textbf{non-path-reporting} distance oracle.

Three most important parameters of distance oracles are their size, stretch and query time\footnote{Although many works also study the \textit{construction time} of distance oracles, in this work we focus on the query time-stretch-size trade-offs of our distance oracles, and do not try to optimize their construction time. However, all our distance oracles can be constructed in $\tilde{O}(mn)$ time.}. The \textbf{size} of a distance oracle is the number of computer words required for storing the data structure. The \textbf{stretch} of the distance oracle is the minimum value $\alpha$ such that for any query $(u,v)\in V^2$ to the oracle, it is guaranteed that the path $P'_{u,v}$ returned by it satisfies $d(u,v)\leq w(P'_{u,v})\leq\alpha\cdot d(u,v)$. The time required for the PRDO to handle a query $(u,v)$ can typically be expressed as $O(q+|P'_{u,v}|)$. The (worst-case) overhead $q$ is referred to as the \textbf{query time} of the PRDO.

First (implicit) constructions of distance oracles were given in \cite{AP90a,C93}. The authors devised hierarchies of neighborhood covers, which can be viewed as PRDOs with stretch $(4+\epsilon)k$ (for any parameter $k=1,2,...$), size $O_{\epsilon}(k\cdot n^{1+\frac{1}{k}}\cdot\log\Lambda)$ (where $\Lambda=\frac{\max_{u,v}d_G(u,v)}{\min_{u\neq v}d_G(u,v)}$ is the \textit{aspect ratio} of the input graph), and query time $O(\log\log\Lambda)$.\footnote{The query time is not explicated in \cite{AP90a,C93}. In \nameref{sec:AppendixF} we argue that their construction of neighborhood covers gives rise to PRDOs with these parameters.} Another variant of Cohen's construction \cite{C93} provides a PRDO with stretch $(2+\epsilon)k$. However, the query time of this distance oracle is $\tilde{O}_\epsilon(k\cdot n^{\frac{1}{k}}\cdot\log\Lambda)$. The aforementioned distance oracles (of \cite{AP90a,C93}) are path-reporting.

Matousek \cite{Mat96} came up with an $\ell_\infty$-embedding, that can be viewed as a non-path-reporting distance oracle with stretch $2k-1$, size $\tilde{O}(k\cdot n^{1+\frac{1}{k}})$ and query time $\tilde{O}(k\cdot n^{\frac{1}{k}})$.

A landmark path-reporting distance oracle was devised by Thorup and Zwick \cite{TZ01}. It provides stretch $2k-1$, size $O(k\cdot n^{1+\frac{1}{k}})$, and query time $O(k)$. Wulff-Nilsen \cite{WN13} devised an improved query algorithm for this PRDO with query time $O(\log k)$.

Assuming Erd\H{o}s girth conjecture \cite{E64} holds, it is easy to see that any distance oracle with stretch $2k-1$ requires $\Omega(n^{1+\frac{1}{k}})$ bits (see \cite{TZ01}). Moreover, Chechik \cite{C15} showed that for path-reporting distance oracles, a stronger lower bound of $\Omega(n^{1+\frac{1}{k}})$ \textit{words} applies.

Mendel and Naor \cite{MN06} devised a non-path-reporting distance oracle with stretch $128k$, size $O(n^{1+\frac{1}{k}})$ and query time $O(1)$. Naor and Tao \cite{NT12} improved the stretch of this oracle to $33k$, and argued that the approach of \cite{MN06,NT12} can hardly lead to a stretch smaller than $16k$. A path-reporting variant of Mendel-Naor's oracle was devised by Abraham et al. \cite{ACEFN20}. Their PRDO has size $O(kn^{1+\frac{1}{k}})$, stretch $O(k\cdot\log\log n)$, and query time $O(1)$. This PRDO has, however, size $\Omega(n\log n)$.

Wulff-Nilsen \cite{WN13} came up with a non-path-reporting distance oracle with stretch $(2+\epsilon)k$ (for any parameters $\epsilon>0$ and $k=1,2,...$), size $O(k\cdot n^{1+\frac{1}{k}})$, and query time $O_\epsilon(1)$. This result was improved by Chechik \cite{C14,C15}, who devised a non-path-reporting distance oracle with stretch $2k-1$, size $O(n^{1+\frac{1}{k}})$ and query time $O(1)$.

To summarize, for non-path-reporting distance oracles, Chechik's construction \cite{C15} provides near-tight bounds. For path-reporting distance oracles, the construction of \cite{TZ01,WN13} is tight up to the factor $k$ in the size, and the factor $\log k$ in the query time. In particular, the PRDO of \cite{TZ01,WN13} always has size $\Omega(n\log n)$, regardless of the choice of the parameter $k$.

The first PRDO with size $o(n\log n)$ was devised by Elkin et al. \cite{ENWN16}. For a parameter $t\geq1$, the oracle of \cite{ENWN16} provides stretch $O(t\cdot n^{\frac{1}{t}})$, size $O(n\cdot t\cdot\log_n\Lambda)$ and query time $O(\log(t\cdot\log_n\Lambda))$. Note, however, that the stretch of this oracle is prohibitively large.

Elkin and Pettie \cite{EP15} devised two constructions of PRDOs with size $o(n\log n)$. Their first construction provides, for parameters $k=1,2,...$ and $\epsilon>0$, a PRDO with stretch $O\left(\left(\frac{1}{\epsilon}\right)^c\cdot k\right)$, size $O(n^{1+\frac{1}{k}})$ and query time $O(n^\epsilon)$. Here $c=\log_{4/3}7$ is a universal constant. Their second construction provides stretch $O(\log^cn)$, size $O(n\log\log n)$ and query time $O(\log\log n)$. Based on ideas from \cite{EN16b}, one can improve the constant $c$ in these results to $\log_{4/3}5$. Recently, Neiman and Shabat \cite{NS23} further improved this constant to $c'=\log_{4/3}4$. Specifically, their PRDO provides, for a parameter $k=1,2,...$, stretch $O(k^{c'})$, size $O(n^{1+\frac{1}{k}})$ and query time $O(\log k)$. Moreover, for unweighted graphs, the stretch of their PRDO is $O(k^2)$.

\subsubsection{Our Results}

Note that the state-of-the-art PRDOs that can have size $o(n\log n)$ \cite{EP15,NS23} suffer from either prohibitively high query time (the query time of the first PRDO by \cite{EP15} is polynomial in $n$), or have stretch $k^{c'}$ (for a constant $c'=\log_{4/3}4>4.81$, or $c'=2$ in the case of unweighted graphs). This is in a stark contrast to the state-of-the-art non-path-reporting distance oracle of Chechik \cite{C15}, that provides stretch $2k-1$, size $O(n^{1+\frac{1}{k}})$ and query time $O(1)$.

In this paper we devise PRDOs that come much closer to the tight bounds for non-path-reporting distance oracles than it was previously known. Specifically, one of our PRDOs (see Theorem \ref{thm:InteractiveSpanner1}), given parameters $k=1,2,...$ and an arbitrarily small constant $\epsilon>0$, provides 
stretch $(4+\epsilon)k$, size 
\[O\left(n^{1+\frac{1}{k}}\cdot\left\lceil\frac{k\cdot\log\log n\cdot\log^{(3)}n}{\log n}\right\rceil\right)~,\]
and query time\footnote{Note that the total query time is $q(k)+O(|P'|)$, where $P'$ is the returned path. Thus decreasing the query time overhead at the expense of increased stretch may result in increasing the total query time. From this viewpoint our two strongest bounds are stretch $(4+\epsilon)k$, size $O\left(\left\lceil\frac{k\cdot\log\log n\cdot\log^{(3)}n}{\log n}\right\rceil\cdot n^{1+\frac{1}{k}}\right)$, and stretch $(12+\epsilon)k$, size $O\left(\left\lceil\frac{k\cdot\log\log n}{\log n}\right\rceil\cdot n^{1+\frac{1}{k}}\right)$ (for the latter, see Table \ref{table:InteractiveSpannersVariety}). In both of these bounds, the query time is $O(\log k)$, like in \cite{TZ01,WN13}. Nevertheless, we believe that it is of interest to study how small can be the query time overhead, even at the expense of increased stretch, for two reasons. First, it can serve as a toolkit for future constructions with yet smaller (hopefully, optimal) stretch and query time. Second, there might be queries $u,v$ for which the number of edges in the returned path $P'_{u,v}=P'$ is smaller than the query time overhead.} $O(\log k)$. For $k\leq\frac{\log n}{\log\log n\cdot\log^{(3)}n}$, the size bound is $O(n^{1+\frac{1}{k}})$. Note that the stretch here is linear in $k$, and in fact, it exceeds the desired bounds of $2k-1$ by a factor of just $2+\epsilon$. At the same time, the size of this PRDO in its sparsest regime is $O(n\log\log n\cdot\log^{(3)}n)$, i.e., far below $n\log n$. Recall that the only previously-existing PRDOs with comparable stretch-size trade-offs \cite{TZ01,WN13} have size $\Omega(n\log n)$. Also, our query time $O(\log k)$ is the same as that of \cite{WN13,EP15,NS23}.

We can also have (see Theorem \ref{thm:InteractiveSpanner2}) a slightly larger stretch $O(k)$ (i.e., still \textit{linear} in $k$), size \[O\left(n^{1+\frac{1}{k}}\cdot\left\lceil\frac{k\cdot\log\log n}{\log n}\right\rceil\right)~,\]
and query time $O\left(\log\left\lceil\frac{k\cdot\log\log n}{\log n}\right\rceil\right)=O(\log\log k)$. The query time of this oracle is \textit{exponentially} smaller than that of the state-of-the-art PRDOs \cite{TZ01,WN13,EP15,NS23}. Moreover, for $k=O\left(\frac{\log n}{\log\log n}\right)$, the query time is constant. At the same time, its stretch is \textit{polynomially better} than in the oracles of \cite{EP15,NS23}. Its size is slightly worse that that of \cite{NS23}. However, we can trade size for stretch and get stretch $O\left(k\cdot\left\lceil\frac{k\cdot\log\log n}{\log n}\right\rceil\right)=O(k\log k)$, size $O(n^{1+\frac{1}{k}})$ and query time $O(\log\log k)$, consequently outperforming the oracle of \cite{NS23} in all parameters. In particular, for $k=\log n$, we get a PRDO with stretch $O(\log n\cdot\log\log n)$, size $O(n)$ and query time $O(\log^{(3)}n)$.

It is instructive to compare our PRDO to that of Thorup and Zwick \cite{TZ01,WN13}. For stretch $\alpha=t\cdot\frac{\log n}{\log\log n}$, where $1\leq t\leq\log\log n$, the PRDO of \cite{TZ01,WN13} has size $O\left(nt\cdot\frac{\log^{1+\frac{2}{t}}n}{\log\log n}\right)$ and query time $O(\log\log n)$, while our PRDO  (with stretch $C\cdot k$, size $O\left(\left\lceil\frac{k\cdot\log\log n}{\log n}\right\rceil\cdot n^{1+\frac{1}{k}}\right)$ and query time $O\left(\log\left\lceil\frac{k\cdot\log\log n}{\log n}\right\rceil\right)$, for a universal constant $C$) has size $O\left(nt\cdot\log^{\frac{C}{t}}n\right)$. For $t$ greater than a sufficiently large constant, the size of our PRDO is therefore strictly smaller than that of \cite{TZ01,WN13} (for \textbf{the same stretch}), while the query time $O\left(\log\left\lceil\frac{k\cdot\log\log n}{\log n}\right\rceil\right)=O(\log t)=O(\log^{(3)}n)$ is at least \textit{exponentially} smaller then the query time of \cite{TZ01,WN13} (which is $O(\log\log n)$).

In addition to the two new PRDOs that were described above (the one with stretch $(4+\epsilon)k$ and query time $O(\log k)$, and the one with stretch $C\cdot k$, for a much larger constant $C$, and query time $O\left(\log\left\lceil\frac{k\cdot\log\log n}{\log n}\right\rceil\right)$), we also present a variety of additional PRDOs that trade between these two extremes. Specifically, we can have stretch $C'\cdot k$, for several different constant values $C'$, $4+\epsilon<C'<C$, and query time $o(\log k)$ (though larger than $O\left(\log\left\lceil\frac{k\cdot\log\log n}{\log n}\right\rceil\right)$). See Section \ref{sec:SpannersVariety} and Table \ref{table:InteractiveSpannersVariety} for full details.

Tables \ref{table:NonPathRep} and \ref{table:PathRep} provide a concise summary of main previous and new results on distance oracles.

\begin{table}[ht]
\begin{center}
\begin{tabular}{|c|c|c|c|}
\hline
Stretch  & Size  & Query Time  & Paper \\ \hline
$2k-1$  & $k\cdot n^{1+\frac{1}{k}}$  & $k\cdot n^{\frac{1}{k}}$ & \cite{Mat96}\\ \hline
$128k$  & $n^{1+\frac{1}{k}}$ & $1$ & \cite{MN06}\\ \hline
$33k$   & $n^{1+\frac{1}{k}}$ & $1$ & \cite{MN06,NT12}\\ \hline
$(2+\epsilon)k$ & $k\cdot n^{1+\frac{1}{k}}$ & $O_\epsilon(1)$ & \cite{WN13}\\ \hline
$2k-1$  & $k\cdot n^{1+\frac{1}{k}}$ & $1$ & \cite{C14}\\ \hline
$2k-1$  & $n^{1+\frac{1}{k}}$  & $1$ & \cite{C15} \\ \hline
\end{tabular}
\end{center}
\caption{A summary of results on \textbf{non-path-reporting} distance oracles. We mostly omit $O$-notations.} \label{table:NonPathRep}
\end{table}

\begin{table*}[ht]
\begin{center}
\begin{tabular}{|c|c|c|c|}
\hline
Stretch  & Size  & Query Time  & Paper \\ \hline
$(4+\epsilon)k$ & $k\cdot n^{1+\frac{1}{k}}\cdot\log_{1+\epsilon}\Lambda$ & $\log\log\Lambda$ & \cite{AP90a,C93}\\ \hline 
$(2+\epsilon)k$ & $k\cdot n^{1+\frac{1}{k}}\cdot\log_{1+\epsilon}\Lambda$ & $k\cdot n^{\frac{1}{k}}\log\Lambda$ & \cite{C93} (2)\\ \hline 
$2k-1$  & $k\cdot n^{1+\frac{1}{k}}$ & $k$ & \cite{TZ01}\\ \hline 
$2k-1$  & $k\cdot n^{1+\frac{1}{k}}$ & $\log k$ & \cite{TZ01,WN13}\\ \hline 
$k\cdot\log\log n$  & $k\cdot n^{1+\frac{1}{k}}$ & $1$ & \cite{ACEFN20}\\ \hline 
$t\cdot n^{\frac{1}{t}}$ & $t\cdot n\cdot\log_n\Lambda$ & $\log(t\cdot\log_n\Lambda)$ & \cite{ENWN16}\\ \hline 
$k^{\log_{4/3}7}$ & $\log k\cdot n^{1+\frac{1}{k}}$& $\log k$ & \cite{EP15} (1)\\ \hline 
$k\cdot(\frac{1}{\epsilon})^{\log_{4/3}7}$ & $n^{1+\frac{1}{k}}$ & $n^\epsilon$ & \cite{EP15} (2)\\ \hline 
$k^{\log_{4/3}5}$ & $\log k\cdot n^{1+\frac{1}{k}}$& $\log k$ & \cite{EP15,EN19}\\ \hline 
$k^{\log_{4/3}4}$ & $n^{1+\frac{1}{k}}$  & $\log k$  & \cite{NS23}\\ \hline 
$k^2$   & $n^{1+\frac{1}{k}}$  & $\log k$ & \cite{NS23}, unweighted\\ \hline
\boldmath$(4+\epsilon)k$   & \boldmath$n^{1+\frac{1}{k}}\cdot\left\lceil\frac{k\cdot\log\log n\cdot\log^{(3)}n}{\log n}\right\rceil$  & \boldmath$\log k$ & \textbf{This paper}, $\epsilon=\Omega(1)$\\ \hline
\boldmath$(8+\epsilon)k$  & \boldmath$n^{1+\frac{1}{k}}\cdot\left\lceil\frac{k\cdot\log\log n\cdot\log^{(3)}n}{\log n}\right\rceil$ & \boldmath$\log\log k$ & \textbf{This paper}, $\Lambda=poly(n)$, $\epsilon=\Omega(1)$\\ \hline
\boldmath$O(k)$   & \boldmath$n^{1+\frac{1}{k}}\cdot\left\lceil\frac{k\cdot\log\log n}{\log n}\right\rceil$  & \boldmath$\log\log k$ & \textbf{This paper}\\ \hline
\boldmath$k\log k$   & \boldmath$n^{1+\frac{1}{k}}$  & \boldmath$\log\log k$ & \textbf{This paper}\\ \hline
\boldmath$\log n\cdot\log\log n$   & \boldmath$n$  & \boldmath$\log^{(3)}n$ & \textbf{This paper}, $k=\log n$\\ \hline
\end{tabular}
\end{center}
\caption{A summary of results on \textbf{path-reporting} distance oracles. Our result with stretch $(8+\epsilon)k$ relies on having a polynomial $\Lambda$ in $n$, where $\Lambda=\frac{\max_{u,v}d_G(u,v)}{\min_{u\neq v}d_G(u,v)}$ is the \textit{aspect ratio} of the input graph. We mostly omit $O$-notations, except for the stretch parameter in one of our results. In fact, our query time in this result is $O\left(\log\left\lceil\frac{k\cdot\log\log n}{\log n}\right\rceil\right)$, which is always $O(\log\log k)$, but for $k=O\left(\frac{\log n}{\log\log n}\right)$ it is $O(1)$.} \label{table:PathRep}
\end{table*}

\subsubsection{Ultra-Sparse PRDOs}

Neiman and Shabat \cite{NS23} also devised an \textbf{ultra-sparse} (or, an \textbf{ultra-compact}) PRDO: for a parameter $t\geq1$, their PRDO has size $n+O\left(\frac{n}{t}\right)$, stretch $O(t\cdot\log^{c'}n)$, and query time $O(t\cdot\log^{c'}n)$ (recall that $c'=\log_{4/3}4$). Our PRDO can also be made ultra-sparse (see Corollary \ref{cor:UltraSparsePRDO1}). Specifically, for a parameter $t\geq1$, the size of our PRDO is $n+O\left(\frac{n}{t}\right)$ (the same as in \cite{NS23}), the stretch is $O(t\cdot\log n\cdot\log\log n)$, and the query time is $O(t\cdot\log\log n)$.

\subsubsection{Related Work}

The study of very sparse distance oracles is a common thread not only in the context of general graphs (as was discussed above), but also in the context of planar graphs. In particular, improving upon oracles of Thorup \cite{Tho04} and Klein \cite{Kle02}, Kawarabayashi et al. \cite{KKS11,KST13} came up with $(1+\epsilon)$-approximate distance oracles with size $O_{\epsilon}(n\log\log n)$.

Very sparse spanners and emulators were extensively studied in \cite{Pet09,P10,EN16b,EM20}. Ultra-sparse sparsifiers, or shortly, ultra-sparsifiers, are also a subject of intensive research \cite{CKM14,KLOS14,KMST10,She13}.

\subsection{Subsequent Work}

Very recently, after a preliminary version of our paper was published in FOCS'23 \cite{ES23}, Chechik and Zhang \cite{CZ24} came up with a construction of PRDOs with stretch $12k$, size $O(n^{1+\frac{1}{k}})$ and query time $O(\log\log k)$. For $k=\log n$, this implies stretch $O(\log n)$, size $O(n)$ and query time $O(\log^{(3)}n)$. Their construction is based on Chechik's optimal non-path-reporting distance oracle \cite{C15}. Whether there exist PRDOs with stretch $2k-1$, size $O(n^{1+\frac{1}{k}})$ and small query time remains, however, an outstanding open problem. We hope that the techniques that we introduce and develop in this paper will be instrumental for its eventual resolution.

\pagebreak
\subsection{Distance Preservers}

All our PRDOs heavily exploit a novel construction of distance preservers that we devise in this paper. Pairwise distance preservers were introduced by Coppersmith and Elkin in \cite{CE05}. Given an $n$-vertex graph and a collection $\mathcal{P}\subseteq V^2$ of vertex pairs, a sub-graph $G'=(V,H)$, $H\subseteq E$, is called a \textbf{pairwise preserver} with respect to $G,\mathcal{P}$ if for every $(u,v)\in\mathcal{P}$, it holds that
\begin{equation} \label{eq:DistancePreservers}
    d_{G'}(u,v)=d_G(u,v)~.
\end{equation}

The pairs in $\mathcal{P}$ are called the \textbf{demand pairs}, and the set $\mathcal{P}$ itself is called the \textbf{demand set}.

We often relax the requirement in Equation (\ref{eq:DistancePreservers}) to only hold \textit{approximately}. Namely, we say that the sub-graph $G'$ is an \textbf{approximate distance preserver} with stretch $\alpha$, or shortly, an \textbf{$\alpha$-preserver}, if for every $(u,v)\in\mathcal{P}$, it holds that
\[d_{G'}(u,v)\leq\alpha\cdot d_G(u,v)~.\]

It was shown in \cite{CE05} that for every $n$-vertex undirected weighted graph and a set $\mathcal{P}$ of $p$ vertex pairs, there exists a ($1$-)preserver with $O(n+\sqrt{n}\cdot p)$ edges. They also showed an upper bound of $O(\sqrt{p}\cdot n)$, that applies even for weighted directed graphs.

Improved exact preservers (i.e., $1$-preservers) for \textit{unwieghted} undirected graphs were devised by Bodwin and Vassilevska-Williams \cite{BVW21}. Specifically, the upper bound in \cite{BVW21} is $O((n\cdot p)^{2/3}+n\cdot p^{1/3})$.

Lower bounds for exact preservers were proven in \cite{CE05,BVW21,B21}. Preservers for unweighted undirected graphs, that allow a small additive error were studied in \cite{Pet09,Par14,CGK13,K17}. In particular, Kavitha \cite{K17} showed that for any unweighted undirected graph and a set of $p$ pairs, there exists a preserver with $\tilde{O}(n\cdot p^{2/7})$ (respectively, $O(n\cdot p^{1/4})$) edges and with purely\footnote{By ``purely additive stretch" we mean that the multiplicative stretch is $1$.} additive stretch of $4$ (respectively, $6$).

Note, however, that the sizes of all these preservers are super-linear in $n$ when $p$ is large. One of the preservers of \cite{CE05} has linear size for $p=O(\sqrt{n})$, and this was exactly the preserver that \cite{EP15} utilized for their PRDO. To obtain much better PRDOs, one needs preservers with linear size, or at least near-linear size, for much larger values of $p$.

For unweighted graphs, a very strong upper bound on $(1+\epsilon)$-preservers follows directly from constructions of \textit{near-additive spanners}. Given a graph $G=(V,E)$, an \textbf{$(\alpha,\beta)$-spanner} is a sub-graph $G'=(V,E')$, $E'\subseteq E$, such that for every pair of vertices $u,v\in V$,
\[d_{G'}(u,v)\leq\alpha\cdot d_G(u,v)+\beta~.\]
When $\alpha=1+\epsilon$, for some small positive parameter $\epsilon>0$, an $(\alpha,\beta)$-spanner is called a \textit{near-additive} spanner.

Constructions of near-additive spanners \cite{EP04,E01,TZ06,Pet09,EN16b} give rise to constructions of $(1+\epsilon)$-preservers (see, e.g., \cite{ABSHKS21}). Specifically, one builds a $(1+\epsilon,\beta_H)$-spanner $H$ \cite{EP04,Pet09,EN16}, and adds it to the (initially empty) preserver. Then for every pair $(u,v)\in\mathcal{P}$ with $d_G(u,v)\leq\frac{\beta_H}{\epsilon}$, one inserts a shortest $u-v$ path $P_{u,v}$ into the preserver. The preserver employs $O(|H|+\frac{\beta_H}{\epsilon}\cdot p)$ edges. For the stretch bound, observe that for any $(u,v)\in \mathcal{P}$ with $d_G(u,v)>\frac{\beta_H}{\epsilon}$, the spanner $H$ provides it with distance
\[d_H(u,v)\leq(1+\epsilon)d_G(u,v)+\beta_H<(1+2\epsilon)d_G(u,v)~.\]
Plugging here the original construction of $(1+\epsilon,\gamma_2)$-spanners with $O(\gamma_2\cdot n^{1+\frac{1}{k}})$ edges, where $\gamma_2=\gamma_2(k)=O\left(\frac{\log k}{\epsilon}\right)^{\log k}$, from \cite{EP04}, one obtains a $(1+\epsilon)$-preserver with $O(\gamma_2\cdot n^{1+\frac{1}{k}}+\frac{\gamma_2}{\epsilon}\cdot p)$ edges. Plugging instead the construction from \cite{Pet09,EN16b} of $(1+\epsilon,\gamma_{4/3})$-spanners with $O(n^{1+\frac{1}{k}}+n\cdot\log k)$ edges, where $\gamma_{4/3}=\gamma_{4/3}(k)=O\left(\frac{\log k}{\epsilon}\right)^{\log_{4/3}k}$, one obtains a $(1+\epsilon)$-preserver with $O(n^{1+\frac{1}{k}}+n\cdot\log k+\frac{\gamma_{4/3}}{\epsilon}\cdot p)$ edges.

For many years it was open if these results can be extended to \textit{weighted} graphs. The first\footnote{See, however, the footnote on path-reporting hopsets in Section \ref{sec:SupportSize}. In a hindsight, a similar result was implicit in \cite{EN16}.} major progress towards resolving this question was recently achieved by Kogan and Parter \cite{KP22}. Specifically, using hierarchies of hopsets they constructed $(1+\epsilon)$-preservers for undirected weighted $n$-vertex graphs with $\left(O(n\log n)^{1+\frac{1}{k}}\cdot\log n+p\right)\cdot\left(\frac{\log n\cdot\log k}{\epsilon}\right)^{\log k}$ edges, for any parameters $k=1,2,...$ and $\epsilon>0$. Note, however, that this preserver always has size $\Omega(n\cdot(\log n)^{\log\log n})$, and this makes it unsuitable for using it for PRDOs of size $o(n\log n)$ (which is the focus of our paper).

In the current paper we answer the aforementioned question in the affirmative, and devise constructions of $(1+\epsilon)$-preservers for weighted graphs, that are on par with the bounds of \cite{EP04,Pet09,EN16b,ABSHKS21} for unweighted ones. Specifically, we devise two constructions of $(1+\epsilon)$-preservers for weighted graphs. The first one (see Theorem \ref{thm:DistancePreserver1}) has size $O(n^{1+\frac{1}{k}}+n\cdot\log k+\gamma_{4/3}\cdot p)$, and the second one (see Theorem \ref{thm:DistancePreserver2}) has size $O(n^{1+\frac{1}{k}}+n(\log k+\log\log\frac{1}{\epsilon})+\tilde{\gamma}_2\cdot p)$, where
\begin{equation} \label{eq:Gamma2Tilde}
    \tilde{\gamma}_2=\left(\frac{\log k+\log\log\frac{1}{\epsilon}}{\epsilon}\right)^{\log k+O(\log^{(3)}k)+\log\log\frac{1}{\epsilon}}~.
\end{equation}

In fact, our bound on $\tilde{\gamma}_2$ is even stronger than that, and is given in Theorem \ref{thm:DistancePreserver2}. The term $\log\log\frac{1}{\epsilon}$ in the size can also be replaced by $\log\left(1+\frac{\log\frac{1}{\epsilon}}{\log^{(3)}n}\right)$. When $\epsilon\geq(\log\log n)^{-k^{c'}}$, for a specific constant $c'>0$, the second construction is better than the first one.

Note that for $p\geq n^{1+\frac{1}{k}}$, for some constant $k$, our preserver uses $O_{\epsilon,k}(p)$ edges  (i.e., only a constant factor more than a trivial lower bound\footnote{For $\epsilon<2$, there is a lower bound of $|\mathcal{P}|$ on the size of $(1+\epsilon)$-preservers, which is demonstrated on an unweighted complete bipartite graph, with an arbitrary set of pairs $\mathcal{P}$. The same lower bound can be also proved for preservers with larger stretch, using dense graphs with high girth.}), while the overhead in the construction of \cite{KP22} depends polylogarithmically on $n$, with a degree that grows with $k$ (specifically, it is $O\left(\frac{\log n}{\epsilon}\right)^{\log k}$). Also, for $p\leq\frac{n}{\gamma_{4/3}(\log n)}=\frac{n}{\left(\frac{\log\log n}{\epsilon}\right)^{O(\log\log n)}}$, our construction provides a preserver of \textit{near-linear} size (specifically, $O(n\log\log n)$), while the size of the preserver of \cite{KP22} is always $\Omega(n\cdot(\log n)^{\log\log n})$. This property is particularly useful for constructing our PRDOs.


We also present a construction of $(3+\epsilon)$-preservers (for weighted graphs) with $O(n^{1+\frac{1}{k}}+n\cdot\log k+k^{O(\log\frac{1}{\epsilon})}\cdot p)$ edges. By substituting here $k=\log n$, we obtain a $(3+\epsilon)$-preserver, for an arbitrarily small constant $\epsilon>0$, with size $O(n\log\log n+p\cdot\log^{O(1)}n)$. In addition, all our constructions of approximate distance preservers directly give rise to \textit{pairwise PRDOs} with the same stretch and size.

A \textbf{pairwise PRDO} is a scheme that given a graph $G$ and a set $\mathcal{P}$ of $p$ vertex pairs, produces a data structure, which given a query $(u,v)\in\mathcal{P}$ returns an approximate shortest $u-v$ path. The size and the stretch of pairwise PRDOs are defined in the same way as for ordinary distance oracles. Pairwise PRDOs with stretch $1$ were studied in \cite{BHT22}, where they were called \textit{Shortest Path Oracles} (however, in \cite{BHT22} they were studied in the context of \textit{directed} graphs). Neiman and Shabat \cite{NS23} have recently showed that the $(1+\epsilon)$-preservers of \cite{KP22} with $\left(O(n\log n)^{1+\frac{1}{k}}\cdot\log n+p\right)\cdot\left(\frac{\log n\cdot\log k}{\epsilon}\right)^{\log k}$ edges can be converted into pairwise PRDOs with the same stretch and size.

We improve upon this by presenting $(1+\epsilon)$-stretch pairwise PRDOs of size \\ $O(n^{1+\frac{1}{k}}+n\cdot\log k+\gamma_{4/3}\cdot p)$ and of size  $O(n^{1+\frac{1}{k}}+n\cdot\log k+\tilde{\gamma}_2\cdot p)$, and $(3+\epsilon)$-stretch pairwise PRDOs of size $O(n^{1+\frac{1}{k}}+n\cdot\log k+k^{O(\log\frac{1}{\epsilon})}\cdot p)$. See Theorem \ref{thm:DistancePreserver1}, Theorem \ref{thm:DistancePreserver2} and Remark \ref{remark:PreserverWithWorseStretch}.

Interestingly, to the best of our knowledge, the unweighted $(1+\epsilon)$-preservers of \cite{EP04,Pet09,EN16,ABSHKS21} do not give rise to partial PRDOs with similar properties.

See Table \ref{table:Preservers} for a concise summary of existing and new constructions of $(1+\epsilon)$-preservers and pairwise PRDOs.

\begin{table*}[ht]
\begin{center}
\begin{tabular}{|c|c|c|c|c|}
\hline
Weighted/  & Stretch  & Size  & Paper  & Pairwise  \\ 
Unweighted &          &       &        & PRDO? \\ \hline
Weighted   & $1$          & $n+\sqrt{n}\cdot p$ & \cite{CE05} & \cite{EP15}*\\ \hline 
Weighted   & $1$          & $n\cdot\sqrt{p}$ & \cite{CE05} & NO\\ \hline 
Unweighted & $1$          & $(n\cdot p)^{2/3}+n\cdot p^{1/3}$ & \cite{BVW21} & NO\\ \hline 
Unweighted & $1+\epsilon$ & $\gamma_2\cdot n^{1+\frac{1}{k}}+\frac{\gamma_2}{\epsilon}\cdot p$ & \cite{EP04,ABSHKS21} & NO\\ \hline 
Unweighted & $1+\epsilon$ & $n^{1+\frac{1}{k}}+n\cdot\log k+\frac{\gamma_{4/3}}{\epsilon}\cdot p$ & \cite{Pet09,EN16,ABSHKS21} & NO\\ \hline 
Weighted & $1+\epsilon$ & $\left((n\log n)^{1+\frac{1}{k}}\cdot\log n+p\right)\cdot\left(\frac{\log n\cdot\log k}{\epsilon}\right)^{\log k}$ & \cite{KP22} & \cite{NS23}\\ \hline 
\textbf{Weighted} & \boldmath$1+\epsilon$ & \boldmath$n^{1+\frac{1}{k}}+n\cdot\log k+\gamma_{4/3}\cdot p$ & \textbf{This paper} & \textbf{YES}\\ \hline 
\textbf{Weighted} & \boldmath$1+\epsilon$ & \boldmath$n^{1+\frac{1}{k}}+n\cdot\log k+\tilde{\gamma}_2\cdot p$ & \textbf{This paper} & \textbf{YES}\\ \hline 
\textbf{Weighted} & \boldmath$3+\epsilon$ & \boldmath$n^{1+\frac{1}{k}}+n\cdot\log k+k^{O(\log\frac{1}{\epsilon})}\cdot p$ & \textbf{This paper} & \textbf{YES}\\ \hline 
\end{tabular}
\end{center}
\caption{A summary of most relevant existing and new results about pairwise preservers and PRDOs. The first column indicates if a result applies to unweighted or weighted graphs. The \textit{Size} omits $O$-notations (except for an exponent in the last row). The last column indicates whether there is a known PRDO-analogue of this particular preserver. In case there is, a reference is provided. The PRDO-analogue of the exact preserver by \cite{CE05}, devised in \cite{EP15} (marked in the table by $*$), provides somewhat weaker bounds on the size than those of the original preserver. See \cite{EP15} for details. Recall that $\gamma_2=O\left(\frac{\log k}{\epsilon}\right)^{\log k}$, $\gamma_{4/3}=O\left(\frac{\log k}{\epsilon}\right)^{\log_{4/3}k}$, and $\tilde{\gamma}_2$ is given by (\ref{eq:Gamma2Tilde}). In the last six rows, one can select $k$ that optimizes the size bound.} \label{table:Preservers}
\end{table*}

\pagebreak
\section{A Technical Overview} \label{sec:TechnicalOverview}

\subsection{Preservers and Hopsets with Small Support Size} \label{sec:SupportSize}

The construction of $(1+\epsilon)$-preservers in \cite{KP22} is obtained via an ingenious black-box reduction from hierarchies of hopsets. Our constructions of significantly sparser preservers is obtained via a direct hopset-based approach, which we next outline.

Given a graph $G=(V,E)$, and a pair of positive parameters $\alpha,\beta$, a set $H\subseteq\binom{V}{2}$ is called an $(\alpha,\beta)$-\textbf{hopset} of $G$ if for every pair of vertices $u,v\in V$, we have
\[d^{(\beta)}_{G\cup H}(u,v)\leq\alpha\cdot d_G(u,v)~.\]
Here $G\cup H$ denotes the weighted graph that is obtained by adding the edges $H$ to $G$, while assigning every edge $(x,y)\in H$ the weight $d_G(x,y)$. Also, $d^{(\beta)}_{G\cup H}(u,v)$ stands for $\beta$\textit{-bounded} $u-v$ distance in $G\cup H$, i.e., the weight of the shortest $u-v$ path in $G\cup H$ with at most $\beta$ edges.

Hopsets were introduced in a seminal work by Cohen \cite{C00}. Elkin and Neiman \cite{EN16} showed that for any $\epsilon>0$ and $k=1,2,...$, and any $n$-vertex undirected weighted graph $G=(V,E)$, there exists a $(1+\epsilon,\gamma_2)$-hopset, where $\gamma_2(k)=O\left(\frac{\log k}{\epsilon}\right)^{\log k}$, with $\tilde{O}(n^{1+\frac{1}{k}})$ edges. Elkin and Neiman \cite{EN19} and Huang and Pettie \cite{HP17} showed that the Thorup-Zwick emulators \cite{TZ06} give rise to yet sparser $(1+\epsilon,\gamma_2)$-hopsets. Specifically, these hopsets have size $O(n^{1+\frac{1}{k}})$.

In this paper we identify an additional (to the stretch $\alpha$, the hopbound $\beta$, and the size $|H|$) parameter of hopsets, that turns out to be crucially important in the context of PRDOs. Given a hopset $H$ for a graph $G=(V,E)$, we say that a subset $E_H$ of $E$ is a \textbf{supporting edge-set} of the hopset $H$, if for every edge $(u,v)\in H$, the subset $E_H$ contains a shortest $u-v$ path in $G$. If these paths are rather $t$-approximate shortest paths, instead of exact shortest paths, for a parameter $t>1$, then we say that $E_H$ is a \textbf{$t$-approximate supporting edge-set} of the hopset $H$. The minimum size $|E_H|$ of a ($t$-approximate) supporting edge-set for the hopset $H$ is called the ($t$-approximate) \textbf{support size} of the hopset.

We demonstrate that for any $\epsilon>0,k=1,2,...$, there exists a $\gamma_{4/3}=\gamma_{4/3}(\epsilon,k)$ (defined above) such that any $n$-vertex graph admits a $(1+\epsilon,\gamma_{4/3})$-hopset with size $O(n^{1+\frac{1}{k}})$ and support size $O(n^{1+\frac{1}{k}}+n\cdot\log k)$. Moreover, these hopsets are \textit{universal}, i.e., the same construction provides, in fact, a $(1+\epsilon,\gamma_{4/3}(\epsilon,k))$-hopsets for all $\epsilon>0$ \textit{simultaneously} (this is a property of Thorup-Zwick emulators \cite{TZ06}, and of hopsets obtained via their construction of emulators \cite{EN19,HP17}). Furthermore, the same hopset serves also as a $(3+\epsilon,k^{O(\log\frac{1}{\epsilon})})$-hopset (with the same size and support size, and also for all $\epsilon>0$ simultaneously).

Next we explain how such hopsets directly give rise to very sparse approximate distance preservers. Below we will also sketch how hopsets with small support size are constructed. Given a hopset $H$ as above, and a set $\mathcal{P}\subseteq V^2$ of $p$ vertex pairs, we construct the preserver $E_{\mathcal{P}}$ in two steps. First, the supporting edge-set $E_H$ of the hopset $H$, with size $O(n^{1+\frac{1}{k}}+n\cdot\log k)$, is added to the (initially empty) set $E_{\mathcal{P}}$. Second, for each pair $(u,v)\in\mathcal{P}$, consider the $u-v$ path $P_{u,v}$ in $G\cup H$, with weight $d^{(\beta)}_{G\cup H}(u,v)\leq(1+\epsilon)d_G(u,v)$, and at most $\beta=\gamma_{4/3}$ edges. For every original edge $e=(x,y)\in P_{u,v}$ (i.e., $e\in E$, as opposed to edges of the hopset $H$), we add $e$ to $E_{\mathcal{P}}$. This completes the construction.

The size bound $|E_{\mathcal{P}}|\leq|E_H|+\gamma_{4/3}\cdot p=O(n^{1+\frac{1}{k}}+n\cdot\log k+\gamma_{4/3}\cdot p)$ is immediate. For the stretch bound, consider a vertex pair $(u,v)\in\mathcal{P}$. For every edge $e\in P_{u,v}\cap E$, the edge belongs to $E_{\mathcal{P}}$. For every hopset edge $e'=(x',y')\in P_{u,v}\cap H$, the shortest $x'-y'$ path in $G$ belongs to the supporting edge-set $E_H$ of the hopset, and thus to $E_{\mathcal{P}}$ as well. Hence $d_{E_{\mathcal{P}}}(x',y')=d_G(x',y')$. Therefore, $d_{E_{\mathcal{P}}}(u,v)\leq\tilde{w}(P_{u,v})$, where $\tilde{w}$ is the weight function in $G\cup H$. It follows that
\[d_{E_{\mathcal{P}}}(u,v)\leq\tilde{w}(P_{u,v})=d^{(\beta)}_{G\cup H}(u,v)\leq(1+\epsilon)d_G(u,v)~.\]

More generally, the above construction shows how an $(\alpha,\beta)$-hopset with supporting edge-set $E_H$ gives rise to an $\alpha$-preserver with size $|E_H|+\beta\cdot p$. The insertion of the paths $P_{u,v}\subseteq G\cup H$, for every pair $(u,v)\in\mathcal{P}$, is in line with the reduction of \cite{KP22} from missing spanners to near-exact preservers. For an undirected weighted graph $G=(V,E)$, and a pair of positive parameters $r$ (referred to as the \textit{missing bound}) and $t$, an \textbf{$r$-missing $t$-spanner}, introduced in \cite{KP22}, is a sub-graph $G'=(V,E')$ that contains all but at most $r$ edges of some $t$-approximate shortest $u-v$ path, for every pair of vertices $u,v\in V$. Given this definition, note that the supporting edge-set $E_H$ of the hopset $H$ serves as a $\beta$-missing $\alpha$-spanner. This immediately improves the $\hat{\beta}$-missing $(1+\epsilon)$-spanners of \cite{KP22}, that have $\hat{\beta}=\hat{\beta}(k,\epsilon,n)=O\left(\frac{k\log k\cdot(\log\log n-\log k)}{\epsilon}\right)^{\log k}$ and size $\tilde{O}\left(n^{1+\frac{1}{k}}\cdot(\hat{\beta}(k,\epsilon,n))^2\right)$. In particular, we get a $\gamma_{4/3}$-missing $(1+\epsilon)$-spanner with size $O(n^{1+\frac{1}{k}}+n\log k)$. Furthermore, recall that the hopset $H$ also serves as a $(3+\epsilon,k^{O(\log\frac{1}{\epsilon})})$-hopset, for all $\epsilon>0$ simultaneously (see \cite{EGN22}). Since its size and support size are the same as above, we also obtain a $k^{O(\log\frac{1}{\epsilon})}$-missing $(3+\epsilon)$-spanner with size $O(n^{1+\frac{1}{k}}+n\log k)$. Table \ref{table:MissingSpanners} summarizes this discussion.

\begin{table}[ht]
\begin{center}
\begin{tabular}{|c|c|c|c|}
\hline
Stretch ($t$)  & Missing Bound ($r$)  & Size  & Paper \\ \hline
$1+\epsilon\cdot k(\log\log n-\log k)$  & $r=O\left(\frac{\log k}{\epsilon}\right)^{\log_{2}k}$  & $\tilde{O}(r^2n^{1+\frac{1}{k}})$ & \cite{KP22}\\ \hline
\boldmath$1+\epsilon$  & \boldmath$O\left(\frac{\log k}{\epsilon}\right)^{\log_{4/3}k}$  & \boldmath$O(n^{1+\frac{1}{k}}+n\log k)$ & \textbf{This paper}\\ \hline
\boldmath$3+\epsilon$  & \boldmath$k^{O\left(\frac{1}{\epsilon}\right)}$  & \boldmath$O(n^{1+\frac{1}{k}}+n\log k)$ & \textbf{This paper}\\ \hline
\end{tabular}
\end{center}
\caption{A summary of results on missing spanners.} \label{table:MissingSpanners}
\end{table}

To convert the preserver $E_{\mathcal{P}}$ into a pairwise PRDO, we store a $\beta=\gamma_{4/3}$-length path $P_{u,v}$ in $G\cup H$ for every pair $(u,v)\in\mathcal{P}$, with weight $d^{(\beta)}_{G\cup H}(u,v)\leq(1+\epsilon)d_G(u,v)$. In addition, we build a dedicated pairwise PRDO $D_H$, that given a hopset-edge $(x,y)\in H$, produces a (possibly approximate) shortest path $P_{x,y}$ in $G$. The oracle $D_H$ implicitly stores the union of these paths $P_{x,y}$. Its size is at most the support size of the hopset, i.e., it is small. We note that the naive construction of (unweighted) $(1+\epsilon)$-preservers, which is based on near-additive spanners \cite{EP04,Pet09,EN16b,ABSHKS21}, does not give rise to a pairwise PRDO. The problem is that the paths in the underlying spanner may be arbitrarily long.

Another benefit of our hopsets $H$ with small support size, is that they give rise to \textit{near-additive spanners}. Namely, these hopsets serve also as \textit{near-additive emulators}. That is, given two vertices $u,v$ in the graph $G=(V,E)$, not only that there is a (low-hop) $u-v$ path in $G\cup H$ with low stretch, but there is also a low-stretch $u-v$ path that is entirely contained in $H$ (rather than in $G\cup H$) - without the guarantee on the number of hops in this path. Since every edge $(x,y)\in H$ can be replaced by a path $P_{x,y}\subseteq E_H\subseteq E$ (where $E_H$ is the supporting edge-set of $H$), with the same weight, we conclude that $E_H$ serves as a near-additive spanner for $G$. Specifically, we provide a $(1+\epsilon,\Tilde{\gamma}_2)$-hopset $H$ of size $O(n^{1+\frac{1}{k}})$ and support size $O(n(\log k+\log\log\frac{1}{\epsilon})+n^{1+\frac{1}{k}})$, for parameters $\epsilon,k$ and $\Tilde{\gamma}_2\approx O\left(\frac{\log k}{\epsilon}\right)^{\log_2k}$. This hopset is also a $(1+\epsilon,\Tilde{\gamma}_2\cdot W)$-emulator (see Definition \ref{def:MixedStretchEmulator} below). Since the support size bound still holds for $H$, we conclude that $E_H$ is a $(1+\epsilon,\Tilde{\gamma}_2\cdot W)$-spanner of $G$, with size $O(n(\log k+\log\log\frac{1}{\epsilon})+n^{1+\frac{1}{k}})$. This strictly improves the additive stretch of the near-additive spanner by \cite{EGN22} (which was $\gamma_{4/3}=O\left(\frac{\log k}{\epsilon}\right)^{\log_{4/3}k}$).\footnote{In fact, the size of our spanner (and the supporting edge set of the hopset) is even better than that. It is $O\left(n\left(\log k+\log\left(1+\frac{\log\frac{1}{\epsilon}}{\log^{(3)}n}\right)\right)+n^{1+\frac{1}{k}}\right)$. In particular, for $\epsilon\geq\frac{1}{poly(\log\log n)}$, its size is $O(n\log k+n^{1+\frac{1}{k}})$, i.e., independent of $\epsilon$. The additive error is given by
\[\Tilde{\gamma}_2(\epsilon,k,n)=O\left(\frac{\log k}{\epsilon}\right)^{\log_2k+O\left(\log\log k+\log\left(1+\frac{\log\frac{1}{\epsilon}}{\log^{(3)}n}\right)\right)}=O\left(\frac{\log k}{\epsilon}\right)^{\log_2k+O(\log\log k+\log\log\frac{1}{\epsilon})}~.\]
It is instructive to compare this result with the $(1+\epsilon,\gamma_2)$-spanners of \cite{ABP17} with size $O\left(\frac{(\log k)^{7/4}}{\epsilon^{3/4}}\cdot n^{1+\frac{1}{k}}\right)$. The size of our spanners has a much better dependence on $\epsilon$ and $k$, while the additive error $\gamma_2=O\left(\frac{\log k}{\epsilon}\right)^{\log_2k}$ is better in \cite{ABP17}.}

Next, we shortly outline the construction of hopsets with small support size\footnote{It is easy to see that the path-reporting hopsets from \cite{EN16,EM20} also have relatively small support size. Specifically, their support size is $O\left(\left(\frac{k\cdot\log n}{\epsilon}\right)^{\log k+1}\cdot n^{1+\frac{1}{k}}\right)$. As a result, they give rise to $(1+\epsilon)$-preservers of roughly the same size as those of \cite{KP22}. In this paper we devise hopsets with support size $O(n^{1+\frac{1}{k}}+n\log k)$, and use them to provide dramatically sparser preservers.}. Since they are closely related to the construction of Thorup-Zwick universal emulators \cite{TZ06} and universal hopsets \cite{EN19,HP17}, we start by overviewing the latter constructions.

\begin{definition} \label{def:MixedStretchEmulator}
An $(\alpha,\beta\cdot W)$-\textbf{emulator} for a (possibly weighted) graph $G=(V,E)$ is a weighted graph $G'=(V,E',w')$, where $E'$ is \textit{not necessarily a subset of $E$} (that is, $G'$ may consist of entirely different edge-set than that of $G$), such that for every two vertices $u,v\in V$,
\[d_G(u,v)\leq d_{G'}(u,v)\leq\alpha\cdot d_G(u,v)+\beta\cdot W(u,v)~,\]
where $W(u,v)$ is the weight of the heaviest edge on any $u-v$ shortest path in $G$.

If $E'\subseteq E$, we say that $G'$ is an $(\alpha,\beta\cdot W)$-spanner.
\end{definition}

Consider a hierarchy of subsets $V=A_0\supseteq A_1\supseteq A_2\supseteq\cdots\supseteq A_{l-1}\supseteq A_l=\emptyset$, where $l$ is some integer parameter. 
For every index $i\in[0,l-1]$ and every vertex $v\in A_i$, we define $B_i(v)$ - the $i$-th \textbf{bunch} of $v$ - to be the set of $A_i$-vertices which are closer to $v$ than the closest vertex $u\in A_{i+1}$ to $v$. The latter vertex $u$ is called the $(i+1)$-th \textbf{pivot} of $v$, and is denoted by $p_{i+1}(v)$. The emulator $H$ is now defined by adding, for every $i\in[0,l-1]$, an edge from every $v\in A_i$ to all of its pivots $p_j(v)$, for $j>i$, and all of its bunch members $u\in B_i(v)$. The weight of an emulator edge $(u,v)$ is $d_G(u,v)$.

Thorup and Zwick \cite{TZ06} showed that $H$ is (in any \textit{unweighted} graph) a $(1+\epsilon,\gamma_2=\gamma_2(\epsilon,k))$-emulator, for all $\epsilon>0$ simultaneously, when choosing $l\approx\log_2k$. Elkin and Neiman \cite{EN19} and Huang and Pettie \cite{HP17} showed that the very same construction provides $(1+\epsilon,\gamma_2)$-hopsets in any \textit{weighted} graph (again, for all $\epsilon>0$ simultaneously). In addition, \cite{TZ06,EN19,HP17} showed that if each subset $A_{i+1}$ is sampled independently at random from $A_i$ with an appropriate sampling probability, then the size of $H$ is $O(n^{1+\frac{1}{k}})$.

Pettie \cite{Pet09} made a crucial observation that one can replace the bunches in this construction by ``one-third-bunches", and obtain sparse near-additive \textit{spanners} (as opposed to emulators). For a parameter $0<\rho\leq1$, a \textbf{\boldmath$\rho$-bunch} $B^{\rho}_i(v)$ of a vertex $v\in A_i$ (for some $i\in[0,l-1]$) contains all vertices $u\in A_i$ that satisfy
\[d_G(v,u)<\rho\cdot d_G(v,p_{i+1}(v))~.\]
Note that under this definition, a bunch is a $1$-bunch. \textit{One-third-bunch} (respectively, \textit{half-bunch}) is $\rho$-bunch with $\rho=\frac{1}{3}$ (respectively, $\rho=\frac{1}{2}$).

On one hand, the introduction of factor $\frac{1}{3}$ makes the additive term $\beta$ somewhat worse. On the other hand, one can now show that the set of all $v-u$ shortest paths, for $v\in V$, $i\in[0,l-1]$, $u\in B^{1/3}_i(v)$, is sparse. As a result, one obtains a spanner rather than an emulator. To establish this, Pettie \cite{Pet09} analysed the number of branching events\footnote{A \textit{branching event} \cite{CE05} between two paths $P_1,P_2$ is a triple $(P_1,P_2,x)$ such that $x$ is a vertex on $P_1,P_2$, and the adjacent edges to $x$ in $P_1$ are not exactly the same as in $P_2$. See Section \ref{sec:GeneralPreliminaries}.} induced by these paths. This approach is based on the analysis of exact preservers from \cite{CE05}.

The same approach (of using one-third-bunches as opposed to bunches) was then used in the construction of PRDOs by \cite{EP15}. Intuitively, using one-third-bunches rather than bunches is necessary for the path-reporting property of these oracles. This comes, however, at the price of increasing the stretch (recall that the stretch of the PRDO of \cite{EP15} is $O(\log^{\log_{4/3}7}n)$, instead of the desired bound of $O(\log n)$).

Elkin and Neiman \cite{EN19} argued that half-bunches can be used instead of one-third-bunches in the constructions of \cite{Pet09,EP15}. This led to improved constructions of near-additive spanners and PRDOs (in the latter context, the stretch becomes $O(\log^{\log_{4/3}5}n)$).

In this paper we argue that for any weighted graph, the set $H^{1/2}$ (achieved using the same construction as $H$, but with half-bunches instead of bunches) is a $(1+\epsilon,\gamma_{4/3})$-hopset with size $O(n^{1+\frac{1}{k}})$ and \textit{support size} $O(n^{1+\frac{1}{k}}+n\log k)$. The analysis of the support size is based on the analysis of near-additive spanners (for unweighted graphs) of \cite{Pet09,EN16b}. The proof that $H^{1/2}$ is a sparse hopset can be viewed as a generalization of related arguments from \cite{EN19,HP17}. There it was argued that $H=H^1$ is a sparse hopset, while here we show that this is the case for $H^\rho$, for any constant\footnote{In fact, we only argue this for $\rho\geq\frac{1}{2}$, but the same argument extends to any constant $\rho>0$.} $0<\rho<1$. Naturally, its properties deteriorate as $\rho$ decreases, but for any constant $\rho>0$, the set $H^\rho$ is quite a good hopset, while for $\rho\leq\frac{1}{2}$, we show that its support size is small.

\subsection{PRDOs} \label{sec:TechnicalPRDOs}

Our constructions of PRDOs consist of three steps. In the first step, we construct $h=O(\log k)$ layers of the Thorup-Zwick PRDO \cite{TZ01}, while using the query algorithm of \cite{WN13}. The result is an oracle with stretch $2h-1=O(\log k)$, size\footnote{The actual choice of the parameter $h$ is more delicate than $h=O(\log k)$, and as a result, our oracle ultimate size is typically much smaller than $O(\log k\cdot n^{1+\frac{1}{k}})$. However, for the sake of clarity, in this sketch we chose to set parameters in this slightly sub-optimal way.} $O(h\cdot n^{1+\frac{1}{k}})=O(\log k\cdot n^{1+\frac{1}{k}})$ and query time $O(\log h)=O(\log\log k)$. However, this oracle is a \textit{partial} PRDO, rather than an actual PRDO. It provides approximate shortest paths only for queries of a certain specific form (see below). For the rest of the queries $(u,v)\in V^2$, it instead provides two paths $P_{u,u'},P_{v,v'}$, from $u$ to some vertex $u'$ and from $v$ to some vertex $v'$, such that $u',v'$ are members in a relatively small set $S$ of size $O(n^{1-\frac{h}{k}})$. The set $S$ is actually exactly the set $A_h$, from the hierarchy of sets that defines Thorup-Zwick PRDO (the vertices $u',v'$ are the corresponding \textit{pivots} of $u,v$, respectively). The two paths $P_{u,u'},P_{v,v'}$ have the property that
\[w(P_{u,u'}),w(P_{v,v'})\leq h\cdot d_G(u,v)~.\]

Intuitively, given a query $(u,v)$, the partial PRDO provides a $u-v$ path $P_{u,v}$ (rather than paths $P_{u,u'},P_{v,v'}$) if and only if the Thorup-Zwick oracle does so while employing only the first $h$ levels of the oracle. We note that using the first $h$ levels of the TZ oracle indeed guarantees almost all the aforementioned properties. The only problem is that the query time of such partial oracle is $O(h)$, rather than $O(\log h)$. To ensure that this oracle has query time $O(\log h)$, we generalize the argument of Wulff-Nilsen \cite{WN13} from full oracles to partial ones.

In the second step, to find a low-stretch path between $u'$ and $v'$, we construct an \textit{interactive emulator} for the set $S$. An \textbf{interactive emulator} is an oracle that similarly to PRDOs, provides approximate shortest paths between two queried vertices. However, as opposed to a PRDO, the provided path does not use actual edges in the original graph, but rather employs \textit{virtual} edges that belong to some sparse low-stretch \textit{emulator} of the graph.

Our construction of interactive emulator is based on the distance oracle of Mendel and Naor \cite{MN06}. Recall that the Mendel-Naor oracle is not path-reporting. Built for $N$-vertex graph and a parameter $k_1$, this oracle provides stretch $O(k_1)$, has size $O(N^{1+\frac{1}{k_1}})$, and has query time $O(1)$. We utilize the following useful property of the Mendel-Naor oracle: there is an $O(k_1)$-emulator $H'$ of the original graph, of size $O(N^{1+\frac{1}{k_1}})$, such that upon a query $(u',v')$, the oracle can return not only a distance estimate $\hat{d}(u',v')$, but also a $u'-v'$ path - \textit{in $H'$} - of length $\hat{d}(u',v')$. We explicate this useful property in this paper. In fact, any other construction of interactive emulator could be plugged instead of the Mendel-Naor oracle in our construction. Indeed, one of the versions of our PRDO uses the Thorup-Zwick PRDO \cite{TZ01} for this step as well, instead of the Mendel-Naor oracle. This version has the advantage of achieving a better stretch, at the expense of a larger query time.

Yet another version of our PRDO employs an emulator based upon a hierarchy of neighborhood covers due to Awerbuch and Peleg \cite{AP90a} and to Cohen \cite{C93}. As a result, this PRDO has intermediate stretch and query time: its stretch is larger than that of the PRDO that employs interactive emulators based on the construction of Thorup and Zwick \cite{TZ01} (henceforth, \textit{TZ-PRDO}), but smaller than that of the PRDO that employs Mendel-Naor's interactive emulator (henceforth \textit{MN-PRDO}). On the other hand, its query time is smaller than that of the TZ-PRDO, but larger than that of the MN-PRDO.

Our sparse interactive emulator $H'$ is built for the complete graph $K=(S,\binom{S}{2})$, where the weight of an edge $(x,y)$ in $K$ is defined as $d_G(x,y)$. Note that since we build this interactive emulator on $K$ rather than on $G$, it is now \textit{sparser} - it has size $O(|S|^{1+\frac{1}{k_1}})$ as opposed to $O(n^{1+\frac{1}{k_1}})$ - while the edges it uses are still virtual (not necessarily belong to $G$).

After the first two steps, we have a path between the original queried vertices $u,v$, that passes through the vertices $u',v'\in S$. Note, however, that this path is not a valid path in the original graph, as its edges between $u'$ and $v'$ are virtual edges in the emulator $H'$ (of size $O(|S|^{1+\frac{1}{k_1}})$). To solve this problem, in the third step, we employ our new sparse \textit{pairwise PRDO} on the graph $G$ and the set of pairs $\mathcal{P}=H'\subseteq V^2$.

Notice that since $|S|=O(n^{1-\frac{h}{k}})$, we can choose a slightly smaller $k_1$ for the Mendel-Naor oracle, such that $|H'|=O(|S|^{1+\frac{1}{k_1}})$ is still relatively small. We choose $k_1$ slightly larger than $\frac{k}{h}$, so that
\[|H'|=O((n^{1-\frac{h}{k}})^{1+\frac{1}{k_1}})=O(n^{1+\frac{1}{k_1}-\frac{h}{k}})\]
will be smaller than $\frac{n^{1+\frac{1}{k}}}{\gamma_{4/3}}$. Then, we use our novel $(1+\epsilon)$-pairwise PRDO (which is also a $(1+\epsilon)$-preserver) that was described above, that has size $O(\gamma_{4/3}\cdot|\mathcal{P}|+n\log k+n^{1+\frac{1}{k}})$ and query time $O(1)$, on $\mathcal{P}=H'$. The size of this pairwise PRDO is therefore $O(n\log k+n^{1+\frac{1}{k}})$. Recall that the size of the partial TZ oracle is $O(\log k\cdot n^{1+\frac{1}{k}})$. Thus the total size of our construction is $O(\log k\cdot n^{1+\frac{1}{k}})$.

The stretch of this new PRDO is the product between the stretches of the partial Thorup-Zwick oracle, the interactive emulator, and the pairwise preserver, which is $O(h\cdot k_1)=O(k)$. The query time also consists of the query times of these three oracles, which is $O(\log h)=O(\log\log k)$. In case we use the Thorup-Zwick PRDO instead of Mendel-Naor oracle, we get a small constant coefficient on the stretch, while the query time increases to $O(\log k)$ (which is the running time of the query algorithm by \cite{WN13}, used in the Thorup-Zwick PRDO).

One final ingredient of our PRDO is the technique that enables one to trade overheads in the size by overheads in the stretch. Specifically, rather than obtaining stretch $O(k)$ and size $O(\log k\cdot n^{1+\frac{1}{k}})$, we can have stretch $O(k\log k)$ and size $O(n^{1+\frac{1}{k}})$ (the query time stays basically the same). 

To accomplish this, we build upon a recent result of Bezdrighin et al. \cite{BEGGHIV22} about \textit{stretch-friendly partitions}. Given an $n$-vertex weighted graph $G=(V,E)$, a \textbf{stretch-friendly partition} $\mathcal{C}$ of $G$ is a partition of $V$ into \textit{clusters}, where each cluster $C\in\mathcal{C}$ is equipped with a rooted spanning tree $T_C$ for $C$. Moreover, these clusters and their spanning trees must satisfy the following properties:
\begin{enumerate}
    \item For every edge $e=(u,v)\in E$ with both endpoints belonging to the same cluster $C\in\mathcal{C}$, and every edge $e'$ in the unique $u-v$ path in $T_C$, $w(e')\leq w(e)$.
    \item For every edge $e=(u,v)\in E$ such that $u\in C$ and $v\notin C$ for some cluster $C\in\mathcal{C}$, and every edge $e'$ in the unique path from $u$ to the root of $T_C$, $w(e')\leq w(e)$.
\end{enumerate}

In \cite{BEGGHIV22}, the authors prove that for every $n$-vertex graph, and every positive parameter $t$, there is a stretch-friendly partition with at most $\frac{n}{t}$ clusters, where the spanning tree of each cluster is of depth $O(t)$. We use stretch-friendly partitions with parameter $t=\log k$, and then apply our PRDO on the resulting cluster-graph. Naturally, the stretch grows by a factor of $t$, but the number of vertices in the cluster-graph is at most $\frac{n}{t}$. Thus, the size of this PRDO becomes $O(\log k\cdot(\frac{n}{\log k})^{1+\frac{1}{k}})=O(n^{1+\frac{1}{k}})$. There are some details in adapting stretch-friendly partitions to the setting of PRDOs that we suppress here. See Section \ref{sec:LinearInteractiveSpanner} for the full details. We note that originally, stretch-friendly partitions were devised in \cite{BEGGHIV22} in the context of ultra-sparse spanners. A similar approach in the context of PRDOs for unweighted graphs was recently applied by Neiman and Shabat \cite{NS23}.

For weighted graphs, the construction of \cite{NS23} is based on a clustering with weaker properties than that of stretch-friendly partitions. Once the clustering is constructed, the construction invokes a black-box PRDO (of Thorup and Zwick \cite{TZ01,WN13}) on the graph induced by this clustering. As a result, the stretch in \cite{NS23} is at least the product of the maximum radius of the clustering by the stretch of Thorup-Zwick's oracle. The product of these quantities is at least $\Omega(\log^2n)$. Indeed, in its sparsest regime, the oracle of \cite{TZ01} has size $\Omega(n\log n)$ and stretch $\Theta(\log n)$. To have overall size of $O(n)$, one needs to invoke the oracle on a cluster-graph with $N=O\left(\frac{n}{\log n}\right)$ nodes, implying that the maximum radius of the clustering is $\Omega(\log n)$. As a result, the approach of \cite{NS23} is doomed to produce PRDOs with stretch $\Omega(\log^2n)$.

Neiman and Shabat \cite{NS23} show how this composition of clustering with a black-box PRDO can produce ultra-sparse PRDOs. To produce our ultra-sparse PRDOs we follow their approach. However, we use a stretch-friendly partition instead of their clustering, and we use our novel PRDO of size $O(n\log\log n)$ instead of the black-box PRDO of \cite{TZ01,WN13}. As a result, we obtain ultra-sparse PRDOs with stretch $\tilde{O}(\log n)$ (as opposed to $O(\log^{c'}n)$ in \cite{NS23}), and query time $O(\log\log n)$ (as opposed to at least $O(\log n)$ in \cite{NS23}).

\subsection{New Combinatorial and Algorithmic Notions}

In this paper we introduce and initiate the study of a number of new combinatorial and algorithmic notions. We have already mentioned them in previous sections, but we believe that they are important enough for highlighting them again here.

The first among these notions is the \textit{support size} of hopsets (see Section \ref{sec:SupportSize} for definition). We present a number of constructions of sparse hopsets with small hopbound and small support size. We have also showed that they directly lead to significantly improved preservers. We believe that these new hopsets could also be useful in many other applications of hopsets, in which the ultimate objective is to compute short paths, rather than just distance estimates.

Another class of new notions that we introduce are \textit{interactive} structures. These include interactive \textit{distance preservers}, interactive \textit{spanners}, and interactive \textit{emulators}\footnote{A related notion of \textit{path-reporting low-hop spanners} for metric spaces was recently introduced and studied in \cite{KLMS22,Fil22}. In the context of graphs, these objects are stronger than interactive emulators (because of the low-hop guarantee), but weaker than interactive spanners (as they do not provide paths in the input graph). Also, the constructions of \cite{KLMS22,Fil22} for general metrics have size $\Omega(n\log n)$ for all choices of parameters.}. Intuitively, interactive distance preserver is a pairwise PRDO with the additional property that all of its output paths belong to a sparse approximate distance preserver. We show that our new sparse preservers give rise directly to interactive distance preservers with asymptotically the same size.

Interactive spanner is a stronger notion than that of PRDO, because it adds an additional requirement that the union of all output paths needs to form a sparse spanner. Nevertheless, all known PRDOs for general graphs (including those that we develop in this paper) are, in fact, interactive spanners with asymptotically the same size.

Finally, interactive emulator is a stronger notion than that of non-path-reporting distance oracle. It generalizes the distance oracle of Mendel-Naor \cite{MN06}. Note that even though interactive emulators are not PRDOs by themselves, we show that in conjunction with interactive distance preservers, they are extremely helpful in building PRDOs.

\subsection{Organization}

After some preliminaries in Section \ref{sec:Preliminaries}, in Section \ref{sec:NewDistancePreserver} we devise our two new constructions of distance preservers. One of them is fully described in Section \ref{sec:NewDistancePreserver}, while for the other we provide a short sketch (full details of it appear in \nameref{sec:AppendixC}). Section \ref{sec:NewEmulator} is dedicated to the construction of our interactive emulator that is based on the Mendel-Naor distance oracle \cite{MN06}. Finally, in Section \ref{sec:NewPRDO} we present our new constructions of PRDOs.

In \nameref{sec:AppendixA} we show that  $H^{1/2}$ is a $(1+\epsilon,\left(\frac{\log k}{\epsilon}\right)^{O(\log k)})$-hopset, and also a $(3+\epsilon,k^{O(\log\frac{1}{\epsilon})})$-hopset. \nameref{sec:AppendixB} is devoted to a few technical lemmas. In \nameref{sec:AppendixD} we argue that the first $h$ levels (for a parameter $h\leq k$) of the TZ oracle can be used as a partial PRDO (in the sense described in Section \ref{sec:TechnicalPRDOs}), and extend Wulff-Nilsen's query algorithm \cite{WN13} to apply to this PRDO (with query time $O(\log h)$). \nameref{sec:AppendixE} is devoted to trading stretch for size in a general way, via stretch-friendly partitions (based on \cite{BEGGHIV22}). In \nameref{sec:AppendixF} we devise one more construction of an interactive emulator, in addition to those based on the TZ PRDO \cite{TZ01,WN13} and the Mendel-Naor oracle \cite{MN06}. The new construction is based upon neighborhood covers due to Awerbuch and Peleg \cite{AP90a} and Cohen \cite{C93}. \nameref{sec:AppendixG} provides a more detailed table that summarizes our new results.

\section{Preliminaries} \label{sec:Preliminaries}

\subsection{Ultrametrics, HSTs and Branching Events} \label{sec:GeneralPreliminaries}

Given an undirected weighted graph $G=(V,E)$, we denote by $d_G(u,v)$ the distance between some two vertices $u,v\in V$. When the graph $G$ is clear from the context, we sometimes omit the subscript $G$ and write $d(u,v)$.

Given some positive parameter $\beta$, we denote by $d^{(\beta)}_G(u,v)$ the weight of the shortest path between $u$ and $v$, among the paths that have at most $\beta$ edges.

When the graph $G$ is clear from the context, $n$ denotes the number of vertices in $G$ and $m$ denotes the number of edges.

\textbf{A Consistent Choice of Shortest Paths.} Given two vertices $u,v$ of a graph $G$, there can be more than one shortest path in $G$ between $u,v$. We often want to choose shortest paths \textbf{in a consistent manner}, that is, we assume that we have some fixed consistent rule of how to choose a shortest path between two vertices. For example, when viewing a path as a sequence of edges, choosing shortest paths in a consistent manner can be done by choosing for every $u,v\in V$ the shortest path between $u,v$ that is the smallest \textit{lexicographically} (assuming a total order over the edge set $E$).

For Section \ref{sec:NewEmulator}, we will need the following definitions, regarding metric spaces.
\begin{definition} \label{def:MetricSpace}
A \textbf{metric space} is a pair $(X,d)$, where $X$ is a set and $d:X\rightarrow\mathbb{R}_{\geq0}$ satisfies
\begin{enumerate}
    \item For every $x,y\in X$, $d(x,y)=0$ if and only if $x=y$.
    \item For every $x,y\in X$, $d(x,y)=d(y,x)$.
    \item (Triangle Inequality) For every $x,y,z\in X$, $d(x,z)\leq d(x,y)+d(y,z)$.
\end{enumerate}

If $d$ also satisfies $d(x,z)\leq\max\{d(x,y),d(y,z)\}$ for every $x,y,z\in X$, then $d$ is called an \textbf{ultrametric} on $X$, and $(X,d)$ is called an \textbf{ultrametric space}.
\end{definition}

\begin{definition} \label{def:HST}
A \textbf{hierarchically (well) separated tree} or \textbf{HST}\footnote{We actually give here the definition of a $1$-HST. The definition of a $k$-HST, for a general $k$, is given in \cite{BLMN03}.} is a rooted tree $T=(V,E)$ with labels\footnote{The original definition by Bartal in \cite{B96} uses weights on the \textit{edges} instead of labels on the vertices. We use a different equivalent notion, that was given by Bartal et al. in another paper \cite{BLMN03}.} $\ell:V\rightarrow\mathbb{R}_{\geq0}$, such that if $v\in V$ is a child of $u\in V$, then $\ell(v)\leq\ell(u)$, and for every leaf $v\in V$, $\ell(v)=0$. 
\end{definition}

Let $x,y$ be two vertices in a rooted tree (e.g., in an HST). Let $lca(x,y)$ denote the lowest common ancestor of $x,y$, which is a vertex $z$ such that $x,y$ are in its sub-tree, but $z$ does not have a child with the same property. Given an HST, let $L$ be its set of leaves, and define the function $d(x,y)=\ell(lca(x,y))$. It is not hard to see that $(L,d)$ is an ultrametric space\footnote{Strictly speaking, to satisfy property 1 of Definition \ref{def:MetricSpace}, one needs to require that all internal labels are strictly positive.}. In Section \ref{sec:NewEmulator}, we will use the fact, that was proved by Bartal et al. in \cite{BLMN03}, that \textit{every} finite ultrametric can be represented by an HST with this distance function.

Another useful notion is that of \textit{Branching Events} \cite{CE05}. Given a graph $G=(V,E)$ and a set of pairs $\mathcal{P}\subseteq V^2$, for every pair $a=(x,y)\in\mathcal{P}$, let $P_a$ be a shortest path between $x,y$ in the graph $G$ (for every $a\in\mathcal{P}$ we choose a single $P_a$ in some consistent manner). For $a,b\in\mathcal{P},x\in V$, we say that $(a,b,x)$ is a \textbf{Branching Event} if $x\in V(P_a)\cap V(P_b)$, and the adjacent edges to $x$ in $P_a$ are not the same as the adjacent edges to $x$ in $P_b$. Define $Branch(\mathcal{P})$ as the set of all branching events of $\mathcal{P}$.

This notion of \textit{Branching Events} was introduced by Coppersmith and Elkin in \cite{CE05} in the context of distance preservers, and was used in many subsequent works on this subject \cite{Pet09,B21,BVW21,AB24}. It also plays a key role in the construction by Elkin and Pettie in \cite{EP15}. In this paper, the authors built 
a pairwise PRDO with stretch $1$. The following theorem is an immediate corollary from Theorem 3.2 in \cite{EP15}.
\begin{theorem}[\cite{EP15}] \label{thm:DPPRO}
Given an undirected weighted graph $G$ and a set $\mathcal{P}$ of pairs of vertices, there is an interactive $1$-distance preserver with query time $O(1)$ and size 
\[O(n+|\mathcal{P}|+|Branch(\mathcal{P})|)~.\]
\end{theorem}

\subsection{Interactive Structures} \label{sec:Definitions}


\begin{definition}[Interactive Spanner]
Given an undirected weighted graph $G=(V,E,w)$, an \textbf{interactive $\alpha$-spanner} is a pair $(S,D)$ where $S$ is an (edges set of an) $\alpha$-spanner of $G$, and $D$ is an oracle that answers every query $(u,v)\in V^2$ with a path $P_{u,v}\subseteq S$ such that
\[d_G(u,v)\leq w(P_{u,v})\leq\alpha\cdot d_G(u,v)~.\]
\end{definition}
The size of $(S,D)$ is defined as $\max\{|S|,|D|\}$, where $|D|$ is the storage size of the oracle $D$, measured in \textit{words}. The query time of $(S,D)$ is the query time of $D$. Note that in particular, $D$ is a PRDO for $G$.

\begin{definition}[Interactive Emulator] \label{def:InteractiveEmulator}
Given an undirected weighted graph $G=(V,E,w)$, an \textbf{interactive $\alpha$-emulator} is a pair $(H,D)$, where $H=(V,E',w')$ is an $\alpha$-emulator (i.e., an $(\alpha,0)$-emulator) of $G$, and $D$ is an oracle that answers every query $(u,v)\in V^2$ with a path $P_{u,v}\subseteq E'$ such that
\[d_G(u,v)\leq w'(P_{u,v})\leq\alpha\cdot d_G(u,v)~.\]
\end{definition}

\begin{definition}[Interactive Distance Preserver]
Given an undirected weighted graph $G=(V,E,w)$ and a set $\mathcal{P}$ of pairs of vertices, an \textbf{interactive $\alpha$-distance preserver} is a pair $(S,D)$, where $S$ is an $\alpha$-distance preserver of $G$ and $\mathcal{P}$, and $D$ is an oracle that answers every query $(u,v)\in\mathcal{P}$ with a path $P_{u,v}\subseteq S$ such that
\[d_G(u,v)\leq w(P_{u,v})\leq\alpha\cdot d_G(u,v)~.\]
\end{definition}
Note that $D$ is a pairwise PRDO for $G$ and $\mathcal{P}$.

\subsection{Hierarchy of Sets} \label{sec:HierarchyOfSets}

Let $G=(V,E)$ be an undirected weighted graph. In most of our constructions, we use the notion of a \textit{hierarchy of sets} $V=A_0\supseteq A_1\supseteq A_2\supseteq\cdots\supseteq A_l=\emptyset$, for an integer parameter $l$. Given an index $i\in[0,l-1]$ and a vertex $v\in V$, we define
\begin{enumerate}
    \item $p_i(v)=$ the closest vertex to $v$ from $A_i$. The vertex $p_i(v)$ is called the $i$-th \textbf{pivot} of $v$.
    \item $B_i(v)=\{u\in A_i\;|\;d(v,u)<d(v,p_{i+1}(v)\}$. The set $B_i(v)$ is called the $i$-th \textbf{bunch} of $v$.
    \item $\bar{B}_i(v)=B_i(v)\cup\{p_i(v)\}$. The set $\bar{B}_i(v)$ is called the $i$-th \textbf{extended bunch} of $v$.
    \item $B^{1/2}_i(v)=\{u\in A_i\;|\;d(v,u)<\frac{1}{2}d(v,p_{i+1}(v)\}$. The set $B^{1/2}_i(v)$ is called the $i$-th \textbf{half-bunch} of $v$.
    \item $\bar{B}^{1/2}_i(v)=B^{1/2}_i(v)\cup\{p_i(v)\}$. The set $\bar{B}^{1/2}_i(v)$ is called the $i$-th \textbf{extended half-bunch} of $v$.
\end{enumerate}

Note that at this point, the sets $B_{l-1}(v),\bar{B}_{l-1}(v),B^{1/2}_{l-1}(v),\bar{B}^{1/2}_{l-1}(v)$ are not well-defined, since $p_l(v)$ is not defined. We define them as the set $A_{l-1}$.

We also define the following sets, each of them consists of pairs of vertices from $V$.
\begin{enumerate}
    \item For every $0\leq i<l$, $H_i=\bigcup_{v\in A_i}\bigcup_{u\in B_i(v)}\{(v,u)\}$.
    \item For every $0\leq i<l$, $\bar{H}_i=\bigcup_{v\in A_i}\bigcup_{u\in\bar{B}_i(v)}\{(v,u)\}$.
    \item $H=\bigcup_{i=0}^{l-1}H_i$, and $\bar{H}=\bigcup_{i=0}^{l-1}\bar{H}_i$.
    \item For every $0\leq i<l$, $H^{1/2}_i=\bigcup_{v\in A_i}\bigcup_{u\in B^{1/2}_i(v)}\{(v,u)\}$.
    \item For every $0\leq i<l$, $\bar{H}^{1/2}_i=\bigcup_{v\in A_i}\bigcup_{u\in\bar{B}^{1/2}_i(v)}\{(v,u)\}$.
    \item $H^{1/2}=\bigcup_{i=0}^{l-1}H^{1/2}_i$, and $\bar{H}^{1/2}=\bigcup_{i=0}^{l-1}\bar{H}^{1/2}_i$.
\end{enumerate}

In \nameref{sec:AppendixA} we prove the following claim, that generalizes a similar result by Elkin and Neiman \cite{EN19} and Huang and Pettie \cite{HP17}.
\begin{claim} \label{claim:ENHopset}
For every choice of a set hierarchy, the set $\bar{H}^{1/2}$ is a $(1+\epsilon,\beta_l)$-hopset, simultaneously for all positive $\epsilon\leq O(l)$, where $\beta_l=O(\frac{l}{\epsilon})^{l-1}$.
\end{claim}

Given a constant parameter $t>1$, we sometimes use the notation $\gamma_t(\epsilon,k)=\beta_{\lceil\log_tk\rceil+1}$, i.e., 
\begin{equation} \label{eq:GammaDef}
\gamma_t(\epsilon,k)=O\left(\frac{\log k}{\epsilon}\right)^{\lceil\log_tk\rceil}
\end{equation}

The following lemma will be useful for constructing interactive structures.
\begin{lemma} \label{lemma:PathsToPivots}
For every $i=0,1,...,l-1$, there is an interactive $1$-distance preserver for the set of pairs 
\[\mathcal{P}_i=\{(v,p_i(v))\;|\;v\in V\}~,\]
with query time $O(1)$ and size $O(n)$.
\end{lemma}

\begin{proof}
Define an oracle $D$ that for every $v\in V$ stores a pointer $q(v)$ to the next vertex on the shortest path from $v$ to $p_i(v)$. Given a query $(v,p_i(v))$, return the path created by following these pointers from $v$ to $p_i(v)$.

Note that if the vertex $x$ is on the shortest path from $v$ to $p_i(v)$, then it is easy to verify that $p_i(x)=p_i(v)$. Therefore the pointer $q(x)$ keeps us on the shortest path from $v$ to $p_i(v)$. Hence, the stretch of $D$ is $1$. The query time of $D$ is linear in the size of the returned path. Also, the size of $D$ is $O(n)$.

Lastly, note that every edge on an output path of $D$ is of the form $(v,q(v))$, where by $q(v)$ we refer to the vertex that $q(v)$ points to. Thus, the number of such edges is at most $n$. That is, if we denote by $S$ the set of edges that participate in an output path of $D$, then $|S|\leq n$, and $(S,D)$ is an interactive distance preserver with stretch $1$, query time $O(1)$ and size $O(n)$.


\end{proof}

We now state a useful property of the above notions. A similar claim was proved in \cite{Pet09,EP15} for \textit{one-third}-bunches\footnote{The formal definition of a one-third-bunch is analogous to the definition of a half-bunch, but with $\frac{1}{3}$ instead of $\frac{1}{2}$: $B^{1/3}_i(v)=\{u\in A_i\;|\;d_G(v,u)<\frac{1}{3}d_G(v,p_{i+1}(v))\}$.} (instead of half-bunches), and later it was extended to half-bunches in \cite{EGN22}. 
\begin{lemma} \label{lemma:BranchingEvents}
For every $i=0,1,...,l-1$,
\[|Branch(H^{1/2}_i)|\leq4\sum_{u\in A_i}|B_i(u)|^3~.\]
\end{lemma}

The proof of this lemma can be found in \nameref{sec:AppendixB}.


In all our constructions, we employ the following general process of producing the hierarchy of sets. The same process was employed also in \cite{TZ01,EP15}, and in many other works on this subject. We define $A_0$ to be $V$, and then for every $i=0,1,...,l-2$ sequentially, we construct $A_{i+1}$ by sampling each vertex from $A_i$ independently with some probability $q_i$. 

In \nameref{sec:AppendixB} we prove the following lemma.

\begin{lemma} \label{lemma:BasicSizes}
The following inequalities hold:
\begin{enumerate}
    \item For every $i<l-1$, $\mathbb{E}[|\bar{H}_i|]\leq\frac{n}{q_i}\prod_{j=0}^{i-1}q_j$.\newline
    \item $\mathbb{E}[|\bar{H}_{l-1}|]\leq2\cdot\max\{1,(n\prod_{j=0}^{l-2}q_j)^2\}$.\newline
    \item For every $i<l-1$, $\mathbb{E}[|Branch(H^{1/2}_i)|]\leq24\cdot\frac{n}{q^3_i}\prod_{j=0}^{i-1}q_j$.\newline
    \item $\mathbb{E}[|Branch(H^{1/2}_{l-1})|]\leq60\cdot\max\{1,(n\prod_{j=0}^{l-2}q_j)^4\}$.
\end{enumerate}
\end{lemma}

\section{New Interactive Distance Preservers} \label{sec:NewDistancePreserver}

We start by introducing two constructions of interactive $(1+\epsilon)$-distance preservers. For most choices of the parameters (as long as $\epsilon>(\log\log n)^{-k^c}$, for a constant $c$ which is approximately $0.58$), the second construction has an asymptotically smaller size than the first one. See Remark \ref{remark:PreserversComparison} for comparison between the two constructions.

\begin{theorem} \label{thm:DistancePreserver1}
Given an undirected weighted graph $G=(V,E)$, an integer $k\geq3$, a positive parameter $\epsilon\leq O(\log k)$, and a set of pairs $\mathcal{P}\subseteq V^2$, there exists an interactive $(1+\epsilon)$-distance preserver with query time $O(1)$ and size
\[O(|\mathcal{P}|\cdot\gamma_{4/3}(\epsilon,k)+n\log k+n^{1+\frac{1}{k}})~,\]
where $\gamma_{4/3}(\epsilon,k)=O(\frac{\log k}{\epsilon})^{\lceil\log_{4/3}k\rceil}$.

In addition, there exists a universal\footnote{By \textit{universal} $(\alpha(\epsilon),\beta(\epsilon))$-\textit{hopset}, we mean that this hopset is an $(\alpha(\epsilon),\beta(\epsilon))$-hopset for every $\epsilon$ simultaneously.} $(1+\epsilon,\gamma_{4/3}(\epsilon,k))$-hopset with size $O(n^{1+\frac{1}{k}})$ and support size $O(n^{1+\frac{1}{k}}+n\log k)$.
\end{theorem}

\begin{proof}

We first construct a hierarchy of sets as in Section \ref{sec:HierarchyOfSets}, with 
\[q_i=\frac{1}{2}n^{-\frac{(4/3)^i}{3k}}\text{ and }l=\lceil\log_{4/3}k\rceil+1~,\]
for $i=0,1,...,l-2$. These choices are similar to the values of sampling probabilities in \cite{Pet09,EP15,EGN22}, which produce near-additive spanners (\cite{Pet09,EGN22}) and PRDOs (\cite{EP15}).

Recall that by Claim \ref{claim:ENHopset}, $\bar{H}^{1/2}$ is a $(1+\epsilon,\beta_l)$-hopset, for all $\epsilon>0$ simultaneously, where $\beta_l=O(\frac{l}{\epsilon})^{l-1}$ (see \nameref{sec:AppendixA} for a proof). For each $i=0,1,...,l-1$, let $(S^1_i,D^1_i)$ be the interactive $1$-distance preserver from Lemma \ref{lemma:PathsToPivots} for the set $\{(v,p_i(v))\;|\;v\in V\}$, and let $(S^2_i,D^2_i)$ be the interactive $1$-distance preserver from Theorem \ref{thm:DPPRO} for the set of pairs $H^{1/2}_i$.
Let $\mathcal{P}$ be a set of pairs of vertices in $G$. We define an oracle $D$ for our new interactive distance preserver as follows.
\begin{enumerate}
    \item For every pair $(x,y)\in\mathcal{P}$, the oracle $D$ stores a path $P_{x,y}\subseteq G\cup\bar{H}^{1/2}$ between $x,y$ with length at most $\beta_l$ and weight at most $(1+\epsilon)d(x,y)$.
    \item For every $i=0,1,...,l-1$, the oracle $D$ stores the oracles $D^1_i$ and $D^2_i$.
    \item For every edge $e\in\bar{H}^{1/2}$, the oracle $D$ stores a \textit{flag} variable $f_e$. If $e\in\{(v,p_i(v))\;|\;v\in V\}$ for some $i$, then $f_e=1$. If $e\in H^{1/2}_i$ for some $i$, then $f_e=2$. In both cases, $D$ also stores the relevant index $i$.
\end{enumerate}

Next, we describe the algorithm for answering queries. Given a query $(x,y)\in\mathcal{P}$, find the stored path $P_{x,y}$. Now replace each hop-edge $e=(a,b)\in\bar{H}^{1/2}$ that appears in $P_{x,y}$; if $f_e=1$, replace $e$ with the path returned from $D^1_i(a,b)$; if $f_e=2$, replace $e$ with the path returned from $D^2_i(a,b)$. Return the resulting path as an output.

Since the query times of $D^1_i,D^2_i$ are linear in the length of the returned path, for every $i$, the total query time of $D$ is also linear in the size of the output. Thus, the query time of our PRDO is $O(1)$.

Also, note that the paths that are returned by the oracles $D^1_i,D^2_i$ have the same weight as the hop-edges that they replace. Therefore the resulting path has the same weight as $P_{x,y}$, which is at most $(1+\epsilon)d(x,y)$. Hence, $D$ has a stretch of $1+\epsilon$.

The size of $D$ is the sum of its components' sizes. For storing the paths $P_{x,y}$ (item $1$ in the description of $D$) we need $O(|\mathcal{P}|\beta_l)$ space. For storing the flags and the indices (item $3$) we need $O(|\bar{H}^{1/2}|)$ space. For each $i$, the size of $D^1_i$ is $O(n)$ and, by Theorem \ref{thm:DPPRO}, the size of $D^2_i$ is $O(n+|H^{1/2}_i|+|Branch(H^{1/2}_i)|)$. For estimating these quantities, we use Lemma \ref{lemma:BasicSizes}. By inequalities (1) and (3) in Lemma \ref{lemma:BasicSizes}, for every $i<l-1$, in expectation,
\begin{eqnarray*}
    |\bar{H}^{1/2}_i|&\leq&|\bar{H}_i|\leq\frac{n}{q_i}\prod_{j=0}^{i-1}q_j~,\\
    |Branch(H^{1/2}_i)|&\leq&24\cdot\frac{n}{q^3_i}\prod_{j=0}^{i-1}q_j~.
\end{eqnarray*}
Summing these two quantities, we get
\begin{eqnarray*}
|\bar{H}^{1/2}_i|+|Branch(H^{1/2}_i)|&\leq&\frac{n}{q_i}\prod_{j=0}^{i-1}q_j+24\cdot\frac{n}{q^3_i}\prod_{j=0}^{i-1}q_j
\leq25\cdot\frac{n}{q^3_i}\prod_{j=0}^{i-1}q_j\\
&=&25n\cdot2^3n^{\frac{(4/3)^i}{k}}\cdot\frac{1}{2^i}n^{-\frac{(4/3)^i-1}{k}}
=\frac{25}{2^{i-3}}n^{1+\frac{1}{k}}~.
\end{eqnarray*}

For $i=l-1$, first denote $q=\prod_{j=0}^{l-2}q_j$. By inequalities (2) and (4) in Lemma \ref{lemma:BasicSizes}, if $nq\geq1$, then
\begin{eqnarray*}
|\bar{H}^{1/2}_{l-1}|+|Branch(H^{1/2}_{l-1})|&\leq&2(nq)^2+60(nq)^4
\leq62(n\prod_{j=0}^{l-2}q_j)^4\\
&=&62(\frac{1}{2^{l-1}}n^{1-\frac{(4/3)^{l-1}-1}{k}})^4
\leq62(\frac{1}{2^{l-1}}n^{1-\frac{k-1}{k}})^4\\
&=&\frac{62}{2^{4l-4}}n^{\frac{4}{k}}
=O(n^{1+\frac{1}{k}})~,
\end{eqnarray*}
where the last step is true since $k\geq3$. Otherwise, if $nq<1$, then the expected value of $|\bar{H}^{1/2}_{l-1}|+|Branch(H^{1/2}_{l-1})|$ is bounded by a constant, which is also $O(n^{1+\frac{1}{k}})$.

Hence, we get that the total expected size of $D$ is
\begin{eqnarray*}
&&O(|\mathcal{P}|\beta_l+|\bar{H}^{1/2}|)+\sum_{i=0}^{l-1}O(n+|H^{1/2}_i|+|Branch(H^{1/2}_i)|)\\
&=&O(|\mathcal{P}|\beta_l)+\sum_{i=0}^{l-1}O(n+|\bar{H}^{1/2}_i|+|Branch(H^{1/2}_i)|)\\
&=&O(|\mathcal{P}|\beta_l)+O(n\cdot l)+\sum_{i=0}^{l-2}O(\frac{1}{2^{i-3}}n^{1+\frac{1}{k}})+O(n^{1+\frac{1}{k}})\\
&=&O(|\mathcal{P}|\beta_l+n\cdot l+n^{1+\frac{1}{k}})=O(|\mathcal{P}|\beta_l+n\cdot\log k+n^{1+\frac{1}{k}})~.
\end{eqnarray*}

Define
\[S=\tilde{H}\cup\bigcup_{(x,y)\in\mathcal{P}}(P_{x,y}\cap G)~,\]
where $\tilde{H}=\bigcup_{i=0}^{l-1}S^1_i\cup\bigcup_{i=0}^{l-1}S^2_i$. Then, the output paths of $D$ are always contained in $S$, and $S$ has the same size bound as $D$. Hence, $(S,D)$ is an interactive $(1+\epsilon)$-distance preserver with query time $O(1)$ and size
\[O(|\mathcal{P}|\beta_l+n\cdot\log k+n^{1+\frac{1}{k}})~.\]

In addition, note that $\tilde{H}$ is the supporting edge-set of the hopset $\bar{H}^{1/2}$. The analysis above implies that $|\tilde{H}|=O(n^{1+\frac{1}{k}}+n\cdot l)=O(n^{1+\frac{1}{k}}+n\log k)$. Since this is true for every $\epsilon>0$, we conclude that $\bar{H}^{1/2}$ is a universal $(1+\epsilon,\gamma_{4/3}(\epsilon,k))$-hopset with size $O(n^{1+\frac{1}{k}})$ and support size $O(n^{1+\frac{1}{k}}+n\log k)$.

\end{proof}

\begin{remark} \label{remark:PreserverWithWorseStretch}

In Theorem 8 in \cite{EGN22}, the authors prove that $\bar{H}$ (see Section \ref{sec:HierarchyOfSets} for definition) is also a $(3+\epsilon,2(3+\frac{12}{\epsilon})^{l-1})$-hopset, for $\epsilon\in(0,12]$. By a similar argument, we also prove that $\bar{H}^{1/2}$ is a $(3+\epsilon,(12+\frac{40}{\epsilon})^{l-1})$-hopset, for $\epsilon=O(l)$. For completeness, this proof appears in \nameref{sec:AppendixA}. Using this property of $\bar{H}^{1/2}$, we get an interactive $(3+\epsilon)$-distance preserver with a better size of
\[O(|\mathcal{P}|k^{\log_{4/3}(12+\frac{40}{\epsilon})}+n\cdot\log k+n^{1+\frac{1}{k}})~.\]
Furthermore, using a similar technique as in the proof of Theorem \ref{thm:DistancePreserver2} below, the coefficient $k^{\log_{4/3}(12+\frac{40}{\epsilon})}$ of $|\mathcal{P}|$ in the size can be improved to $k^{\log_{2}(12+\frac{40}{\epsilon})}\cdot\left(\left\lceil\frac{k\log\log n\cdot\log^{(3)}n}{\log n}\right\rceil\cdot\log\frac{1}{\epsilon}\right)^{\log_{4/3}(12+\frac{40}{\epsilon})}\leq k^{\log_{2}(12+\frac{40}{\epsilon})}\cdot\left(\log k\cdot\log\log k\cdot\log\frac{1}{\epsilon}\right)^{\log_{4/3}(12+\frac{40}{\epsilon})}$.

\end{remark}

Our second construction of an interactive $(1+\epsilon)$-distance preserver improves the factor of $\gamma_{4/3}(\epsilon,k)$ in the size of the preserver from Theorem \ref{thm:DistancePreserver1}, to almost $\gamma_2$ (see Equation (\ref{eq:GammaDef}) for the definition of $\gamma_t(\epsilon,k)$). Since large parts of the proof are very similar to the proof of Theorem \ref{thm:DistancePreserver1}, we give here only a sketch of the proof, while the full proof appears in \nameref{sec:AppendixC}.

\begin{theorem} \label{thm:DistancePreserver2}
Given an $n$-vertex undirected weighted graph $G=(V,E)$, an integer $k\in[3,\log_2n]$, a positive parameter $\epsilon\leq1$, and a set of pairs $\mathcal{P}\subseteq V^2$, $G$ has an interactive $(1+\epsilon)$-distance preserver with query time $O(1)$ and size 
\[O\left(|\mathcal{P}|\cdot\Tilde{\gamma}_2(\epsilon,k,n)+n\cdot\left(\log k+\log\left(1+\frac{\log\frac{1}{\epsilon}}{\log^{(3)}n}\right)\right)+n^{1+\frac{1}{k}}\right)~,\]
where $\Tilde{\gamma}_2(\epsilon,k,n)=O\left(\frac{\log k+\rho}{\epsilon}\right)^{\left\lceil\log_2k\right\rceil+\rho}$, and
\begin{equation*}
    \rho=\rho(\epsilon,k,n)=\log_{4/3}\left(\frac{k\log\log n\cdot\log^{(3)}n}{\log n}\right)+\log_{4/3}\left(1+\frac{\log\frac{1}{\epsilon}}{\log^{(3)}n}\right)+O(1)
\end{equation*}
(note that $\log_{4/3}\left(\frac{k\log\log n\cdot\log^{(3)}n}{\log n}\right)\leq\log_{4/3}\log k+\log_{4/3}\log\log k$, and if $\epsilon=(\log\log n)^{-k^{o(1)}}$, then \newline 
$\log_{4/3}\left(1+\frac{\log\frac{1}{\epsilon}}{\log^{(3)}n}\right)=o(\log k)$. Thus, in this case, $\rho=o(\log k)$. Also, generally, $\rho\leq\log_{4/3}\log k+\log_{4/3}\log\log k+\log_{4/3}\log\frac{1}{\epsilon}$).

In addition, there exists a $(1+\epsilon,\tilde{\gamma}_2(\epsilon,k,n))$-hopset with size $O(n^{1+\frac{1}{k}})$ and $(1+\epsilon)$-approximate support size 
\[O\left(n^{1+\frac{1}{k}}+n\left(\log k+\log\left(1+\frac{\log\frac{1}{\epsilon}}{\log^{(3)}n}\right)\right)\right)=O\left(n^{1+\frac{1}{k}}+n(\log k+\log\log\frac{1}{\epsilon})\right)~.\]
\end{theorem}

\begin{proof}[Sketch of the Proof]

We use the same hierarchy of sets as in Theorem \ref{thm:DistancePreserver1}, up to the first level $h$ such that $|A_h|\leq\frac{n}{\gamma_{4/3}}$ (we denote here $\gamma_{4/3}=\gamma_{4/3}(\epsilon,k)$). For this purpose, it is enough to set $h=\rho$.
For these $h$ levels, we store the same information in our oracle as before. Specifically, for $i<h$, we store the interactive distance preservers from Theorem \ref{thm:DPPRO} for the demand sets $H^{1/2}_i$, and the ones from Lemma \ref{lemma:PathsToPivots} for the demand sets $\{(v,p_i(v))\;|\;v\in V\}$. 

For each level $i\geq h$ in the hierarchy, we change the probability $q_i$ to $\frac{1}{2}(\frac{n}{\gamma_{4/3}})^{-\frac{2^{i-h}}{k}}$. We now apply our interactive approximate distance preserver from Theorem \ref{thm:DistancePreserver1} for the demand set $\bigcup_{i=h}^{l-1}\bar{H}^{1/2}_i$. Denote the resulting preserver by $(S^3,D^3)$.

We store these interactive distance preservers in our new oracle, together with the same information as before: a path $P_{x,y}\subseteq G\cup\bar{H}^{1/2}$ with low stretch and length, for each $(x,y)\in\mathcal{P}$, and flags for the hop-edges.

Since the stored interactive distance preservers cover all of the pairs in $\bar{H}^{1/2}$, we can use the same stretch analysis as before, and now get a stretch of at most $(1+\epsilon)^2$ (the additional factor of $1+\epsilon$ is because we are now using an interactive $(1+\epsilon)$-distance preserver for $\bigcup_{i=h}^{l-1}\bar{H}^{1/2}_i$, instead of an exact one). Replacing $\epsilon$ by $\frac{\epsilon}{3}$, we still get a stretch of $1+\epsilon$.

The query time analysis is the same as before. As for the size analysis, the first difference is that now we need to include the size of $(S^3,D^3)$, the interactive distance preserver from Theorem \ref{thm:DistancePreserver1} for the demand set $\bigcup_{i=h}^{l-1}\bar{H}^{1/2}_i$. The size of this set is $O\left((\frac{n}{\gamma_{4/3}})^{1+\frac{1}{k}}\right)$, as can be proved using a similar analysis as in the proof of Theorem \ref{thm:DistancePreserver1}. Thus, the size of $(S^3,D^3)$ is
\[O\left((\frac{n}{\gamma_{4/3}})^{1+\frac{1}{k}}\cdot\gamma_{4/3}+n\cdot l+n^{1+\frac{1}{k}}\right)=O(n\cdot l+n^{1+\frac{1}{k}})~.\]
Note that due to the change in the choice of the hierarchy, the value of $l$ increases to $\lceil\log_2k\rceil+h+1$.

The second and final difference in the size analysis is that now storing the paths $P_{x,y}$ requires $|\mathcal{P}|\cdot\beta_l$ space, \textit{for the new value of $l$}. The value of $\beta_l$, that is defined as the hopbound of the hopset in Claim \ref{claim:ENHopset}, is
\[\beta_l=O\left(\frac{l}{\epsilon}\right)^{l-1}=O\left(\frac{\log k+\rho}{\epsilon}\right)^{\lceil\log_2k\rceil+\rho}=\Tilde{\gamma}_2(\epsilon,k,n)~.\]

In conclusion, while the stretch and the query time stay the same, the size decreases to
\begin{eqnarray*}
O(|\mathcal{P}|\Tilde{\gamma}_2(\epsilon,k,n)+n\cdot l+n^{1+\frac{1}{k}})
=O\left(|\mathcal{P}|\Tilde{\gamma}_2(\epsilon,k,n)+n\left(\log k+\log_{4/3}\left(1+\frac{\log\frac{1}{\epsilon}}{\log^{(3)}n}\right)\right)+n^{1+\frac{1}{k}}\right)~.    
\end{eqnarray*}

Accordingly, the support size of the hopset $\bar{H}^{1/2}$ is 
\[O\left(n\left(\log k+\log_{4/3}\left(1+\frac{\log\frac{1}{\epsilon}}{\log^{(3)}n}\right)\right)+n^{1+\frac{1}{k}}\right)~.\]

\end{proof}

\begin{remark} \label{remark:PreserversComparison}

For a constant $c<1-\log_2\frac{4}{3}$, if
\begin{equation} \label{eq:UselessCondition}
    \epsilon\geq(\log\log n)^{-k^c}~,
\end{equation}
we have $\log\frac{1}{\epsilon}\leq k^c\log^{(3)}n$, and therefore
\[\rho(\epsilon,k,n)=\log_{4/3}\left(\frac{k\log\log n\cdot\log^{(3)}n}{\log n}\right)+\log_{4/3}\left(1+\frac{\log\frac{1}{\epsilon}}{\log^{(3)}n}\right)+O(1)=O(\log\log k)+c\cdot\log_{4/3}k~.\]
Thus, under condition (\ref{eq:UselessCondition}), the second construction is better than the first one, as \[\log_2k+c\cdot\log_{4/3}k+O(\log\log k)<\log_{4/3}k~.\]

\end{remark}

\section{A New Interactive Emulator} \label{sec:NewEmulator}

Recall Definition \ref{def:InteractiveEmulator} of an interactive emulator. For every positive integer parameter $k$, the distance oracle of Thorup and Zwick in \cite{TZ01} can be viewed as an interactive $(2k-1)$-spanner, and thus, also as an $(2k-1)$-emulator, with size $O(kn^{1+\frac{1}{k}})$ and query time $O(k)$. Wulff-Nilsen \cite{WN13} later improved this query time to $O(\log k)$.

Next, we devise another interactive emulator, based on the (non-path-reporting) distance oracle by Mendel and Naor in \cite{MN06}, with smaller query time and size, albeit with a slightly larger stretch. We will use it later in Section \ref{sec:NewPRDO}, for building our new PRDO.

\begin{theorem} \label{thm:InteractiveEmulator}
Given an undirected weighted $n$-vertex graph $G=(V,E)$ and an integer $k\geq1$, there is an interactive $O(k)$-emulator with size $O(n^{1+\frac{1}{k}})$ and query time $O(1)$.
\end{theorem}

\begin{proof}

The graph $G=(V,E)$ induces an $n$-point metric space, where the metric $d(x,y)$, for $x,y\in V$, is the shortest path metric in $G$. We apply the following lemma from \cite{MN06} on this metric space. For definitions of metric spaces, ultrametrics and HSTs, see Section \ref{sec:Preliminaries}.

\begin{lemma}[Mendel and Naor \cite{MN06}, Lemma 4.2] \label{lemma:MNDistanceOracle}
Given a metric space $(V,d)$ with $|V|=n$, there is a randomized algorithm that produces a hierarchy of sets $V=X_0\supseteq X_1\supseteq X_2\supseteq\cdots\supseteq X_s=\emptyset$, and a sequence of ultrametrics $\rho_0,\rho_1,...,\rho_{s-1}$ on $V$, such that the following properties hold.
\begin{enumerate}
    \item $\mathbb{E}[s]\leq kn^{\frac{1}{k}}$.
    \item $\mathbb{E}[\sum_{j=0}^{s-1}|X_j|]\leq n^{1+\frac{1}{k}}$.
    \item For every $j=0,1,...,s-1$, and for all $x,y\in V$, it holds that $d(x,y)\leq\rho_j(x,y)$. Also, for all $x\in V,y\in X_j\setminus X_{j+1}$, it holds that $\rho_j(x,y)\leq O(k)\cdot d(x,y)$.
\end{enumerate}
\end{lemma}

For the purpose of our proof, we restrict the ultrametric $\rho_j$ to the subset $X_j\subseteq V$, i.e., we view $(X_j,\rho_j)$ as an ultrametric space, for every $j$. 

It is a well-known fact (see Lemma 3.5 in \cite{BLMN03} by Bartal et al.) that every ultrametric space $(X,\rho)$ can be represented by an HST, such that its leaves form the set $X$. For every $j=0,1,...,s-1$, let $T_j=(N_j,E_j)$ be the HST for the ultrametric space $(X_j,\rho_j)$. The leaves of $T_j$ are $X_j$ and there are labels $\ell_j:N_j\rightarrow\mathbb{R}_{\geq0}$ such that for every $x,y\in X_j$, $\ell_j(lca(x,y))=\rho_j(x,y)$, where $lca(x,y)$ is the lowest common ancestor of $x,y$ in $T_j$. For any edge $(u,v)\in E_j$, such that $u$ is the parent of $v$, we assign a weight $w_j(u,v)=\frac{\ell_j(u)-\ell_j(v)}{2}$. It is easy to check that the weight of the unique path in $T_j$ between two leaves $x,y$ is exactly $\ell_j(lca(x,y))=\rho_j(x,y)$. Note that $\ell_j(x)=\ell_j(y)=\rho_j(x,x)=\rho_j(y,y)=0$.

We now use the following result by Gupta \cite{Gupta01}.
\begin{lemma}[Gupta \cite{Gupta01}] \label{lemma:Gupta}
Given a weighted tree $T$ and a subset $R$ of vertices in $T$, there exists a weighted tree $T'$ on $R$ such that for every $x,y\in R$,
\[d_T(x,y)\leq d_{T'}(x,y)\leq8d_T(x,y)~.\]
\end{lemma}

For every $j=0,1,...,s-1$, let $T'_j$ be the resulting tree from Lemma \ref{lemma:Gupta}, when applied on the tree $T_j$, where the set $R$ is the set of leaves of $T_j$, which is exactly $X_j$. Suppose that $x,y\in X_j$ such that $y\notin X_{j+1}$. Then, $x,y$ are vertices of $T'_j$, and
\begin{eqnarray*}
d_{T'_j}(x,y)\in[d_{T_j}(x,y),8d_{T_j}(x,y)]
=[\rho_j(x,y),8\rho_j(x,y)]\subseteq[d(x,y),8\cdot O(k)\cdot d(x,y)]~,
\end{eqnarray*}
i.e., $d(x,y)\leq d_{T'_j}(x,y)\leq O(k)\cdot d(x,y)$.

We define an oracle $D$ that stores the following:
\begin{enumerate}
    \item For every $v\in V$, the oracle $D$ stores $j(v)$ - the largest index $j\in[0,s-1]$ such that $v\in X_j$.
    \item For every $v\in V$, the oracle $D$ stores $parent(v),depth(v)$, which are vectors of length $j(v)+1$, such that $[parent(v)]_j$ is a pointer to the parent of $v$ in $T'_j$, and $[depth(v)]_j$ is the number of edges in $T'_j$ between $v$ and the root of $T'_j$. 
\end{enumerate}

Given a query $(u,v)\in V^2$, let $j=\min\{j(u),j(v)\}$. If, w.l.o.g., $j=j(v)$, then $u,v\in X_j$, but $v\notin X_{j+1}$, and therefore $d(u,v)\leq d_{T'_j}(u,v)\leq O(k)\cdot d(u,v)$. Using the pointers $[parent(\cdot)]_j,[depth(\cdot)]_j$, the oracle $D$ finds the unique path in $T'_j$ between $u,v$, and returns it as an output.

We can now finally define our interactive emulator $(H,D)$. The sub-graph $H$ is the union graph of the trees $T'_j$, for every $j=0,1,...,s-1$, and $D$ is the oracle described above. We already saw that for every query $(u,v)\in V^2$, the oracle $D$ returns a path $P_{u,v}\subseteq H$ between them, such that
\[d(u,v)\leq w(P_{u,v})=d_{T'_j}(u,v)\leq O(k)\cdot d(u,v)~.\]
Hence, the stretch of our emulator is $O(k)$. 

The query time is linear in the length of the returned path. Hence, the query time is $O(1)$. For the sizes of $H,D$, we have in expectation
\[|H|\leq\sum_{j=0}^{s-1}|T'_j|\leq\sum_{j=0}^{s-1}|X_j|\leq n^{1+\frac{1}{k}}~,\]
\begin{eqnarray*}
|D|\leq O(\sum_{v\in V}(1+(j(v)+1)+(j(v)+1))
=O(n+\sum_{v\in V}(j(v)+1))~.
\end{eqnarray*}
For the last sum, note that $j(v)+1$ is the number of $j$'s such that $v\in X_j$. Hence, if we define an indicator $\delta_{v,j}\in\{0,1\}$ to be $1$ iff $v\in X_j$, we get (in expectation)
\begin{eqnarray*}
\sum_{v\in V}(j(v)+1)=\sum_{v\in V}\sum_{j=0}^{s-1}\delta_{v,j}=\sum_{j=0}^{s-1}\sum_{v\in V}\delta_{v,j}
=\sum_{j=0}^{s-1}|X_j|\leq n^{1+\frac{1}{k}}~.
\end{eqnarray*}
Therefore we still have in expectation that 
\[|D|\leq O(n+n^{1+\frac{1}{k}})=O(n^{1+\frac{1}{k}})~.\]

\end{proof}

\begin{remark} \label{remark:MNCoefficient}
Following the proof of Theorem \ref{thm:InteractiveEmulator}, one can see that the stretch of the resulting interactive emulator is $8c_{MN}$, where $c_{MN}$ is the constant coefficient\footnote{Naor and Tao \cite{NT12} showed that $c_{MN}$ can be $33$, and argued that it can hardly be smaller than $16$.} of $k$ in stretch of the ultrametrics $\rho_j$ in Lemma \ref{lemma:MNDistanceOracle} (see item $3$ in the lemma). We believe that the coefficient $8c_{MN}$ can be improved. However, in this paper we made no effort to optimize it.
\end{remark}

\section{A New Path-Reporting Distance Oracle} \label{sec:NewPRDO}

Our new interactive spanner is constructed by combining three main elements: a \textit{partial TZ oracle} (that we will define below), an interactive emulator, and an interactive distance preserver. Note that we have several possible choices for the interactive emulator and for the interactive distance preserver. For example, we may use either our new interactive emulator from Section \ref{sec:NewEmulator}, or other known constructions of emulators. In addition, the choice of the interactive distance preserver may be from Theorem \ref{thm:DistancePreserver1}, from Theorem \ref{thm:DistancePreserver2} or from Remark \ref{remark:PreserverWithWorseStretch}. By applying different choices of emulator and distance preserver, we get a variety of results, that exhibit different trade-offs between stretch, query time and size.

\subsection{Main Components of the Construction}

The first ingredient of our new interactive spanner is the \textit{partial TZ oracle} - an oracle that given a query $(u,v)\in V^2$, either returns a low-stretch path between $u$ and $v$, or returns two paths, from $u,v$ to vertices $u',v'$, respectively, where $u',v'$ are not far from $u,v$, and both are in some smaller subset of $V$. The formal details are in the following lemma. This oracle is actually the same construction as the distance oracle of Thorup and Zwick in \cite{TZ01}, but only for a certain restricted number of levels. The query algorithm is based on that of Wulff-Nilsen \cite{WN13}.

\begin{lemma} \label{lemma:PartialPRDO}
Let $G=(V,E)$ be an undirected weighted graph with $n$ vertices, and let $1\leq h<k\leq\log_2n$ be two integer parameters. There is a set $S\subseteq V$ of size $O(n^{1-\frac{h}{k}})$ and an oracle $D$ with size $O(h\cdot n^{1+\frac{1}{k}})$, that acts as follows. Given a query $(u,v)\in V^2$, $D$ either returns a path between $u,v$ with weight at most $(2h+1)d_G(u,v)$, or returns two paths $P_{u,u'},P_{v,v'}$, from $u$ to some $u'\in S$ and from $v$ to some $v'\in S$, such that
\[w(P_{u,u'}),w(P_{v,v'})\leq h\cdot d_G(u,v)~.\]

The query time of the oracle $D$ is $O(\log h)$. In addition, there is a set $E'\subseteq E$ such that the output paths of $D$ are always contained in $E'$, and
\[|E'|=O(h\cdot n^{1+\frac{1}{k}})~.\]
\end{lemma}

The proof of Lemma \ref{lemma:PartialPRDO} appears in \nameref{sec:AppendixD}. We will call the oracle $D$ a \textbf{partial TZ oracle}.

The second and third ingredients of our new interactive spanner are interactive emulator and interactive distance preserver. Table \ref{table:EmulatorsAndPreservers} summarizes the properties of interactive emulators and distance preservers that will be used.

\begin{table*}[ht]
\begin{center}
\begin{tabular}{|c|c|c|c|c|}
\hline
Type & Stretch & Size & Query Time & Source \\ \hline
Emulator & $8c_{MN}\cdot k$ & $O(n^{1+\frac{1}{k}})$  & $O(1)$ & Theorem \ref{thm:InteractiveEmulator}\\ \hline
Emulator & $2k-1$  & $O(k\cdot n^{1+\frac{1}{k}})$ & $O(\log k)$ & \cite{TZ01,WN13}\\ \hline
Emulator & $(4+\epsilon)k$   & $O(k\cdot\log_{1+\epsilon}\Lambda\cdot n^{1+\frac{1}{k}})$ & $O(\log\log k)$ & Theorem \ref{thm:APEmulator}\\ \hline
Distance &&&&\\ 
Preserver & $1+\epsilon$ & $O(|\mathcal{P}|\cdot\gamma_{4/3}(\epsilon,k)+n\log k+n^{1+\frac{1}{k}})$ & $O(1)$ & Theorem \ref{thm:DistancePreserver1}\\ \hline
Distance &&&&\\ 
Preserver & $3+\epsilon$ & $O(|\mathcal{P}|\cdot k^{\log_{4/3}(12+\frac{40}{\epsilon})}+n\log k+n^{1+\frac{1}{k}})$ & $O(1)$ & Remark \ref{remark:PreserverWithWorseStretch}\\ \hline
\end{tabular}
\end{center}
\caption{A summary of results on interactive emulators and distance preservers. The first column indicates whether the construction is of an interactive emulator or interactive distance preserver. The last column specifies the source which the construction is based on. Here, $c_{MN}$ denotes the constant coefficient of $k$ in the stretch of the distance oracle of Mendel and Naor \cite{MN06}. Note that the construction of \cite{TZ01,WN13} is actually of an interactive \textit{spanner}. However, any interactive spanner is also an interactive emulator.} \label{table:EmulatorsAndPreservers}
\end{table*}

\subsection{A Variety of Interactive Spanners} \label{sec:SpannersVariety}

We now show how to combine the partial TZ oracle, an interactive emulator and an interactive distance preserver, to construct an interactive spanner. We phrase the following lemma without specifying which of the interactive emulators and distance preservers we use. This way, we can later substitute the emulators and distance preservers from Table \ref{table:EmulatorsAndPreservers}, and get a variety of results for interactive spanners.

\pagebreak

\begin{lemma} \label{lemma:InteractiveSpannerVariety}
Suppose that for every $n$-vertex graph, every integer $k\in[3,\log_2n]$, and every $\epsilon\in(0,1]$, there is an interactive emulator with stretch $\alpha_E\cdot k$, query time $q(k)$ and size $O(k^{\delta}n^{1+\frac{1}{k}})$. In addition, suppose that for every such $k,\epsilon$, for every $n$-vertex graph and every set of pairs $\mathcal{P}$, there is an interactive distance preserver with stretch $\alpha_P$, query time $O(1)$, and size $O(k^{\tau}\cdot|\mathcal{P}|+n\log k+n^{1+\frac{1}{k}})$. 

Then, for any graph $G=(V,E)$ with $n$ vertices, any integer $k\in[3,\log_2n]$, and any $0\leq\epsilon\leq\frac{1}{2}$, there is an integer $h=O\left(\left\lceil\frac{(\delta+\tau)k\cdot\log\log n}{\epsilon\log n}\right\rceil\right)$ and an interactive spanner with stretch $2\alpha_P\cdot\alpha_E((1+\epsilon)k+O(h))$, size $O(n\log k+h\cdot n^{1+\frac{1}{k}})$
and query time $q(k)+O(\log h)$.
\end{lemma}

\begin{proof}

Denote $\sigma=\delta+\tau$. Given the graph $G=(V,E)$, denote by $D_0$ the partial TZ oracle from Lemma \ref{lemma:PartialPRDO} with parameter
\begin{equation} \label{eq:HDef}
    h=\left\lceil\frac{\sigma\cdot k\cdot\log\log n}{\log n}\right\rceil\cdot\left\lceil\frac{1}{\epsilon}\right\rceil~.
\end{equation}
Denote by $S\subseteq V$ the corresponding set from Lemma \ref{lemma:PartialPRDO} (associated with the oracle $D_0$), and recall that the size of $S$ is $|S|=O(n^{1-\frac{h}{k}})$. Also, denote by $E_0\subseteq E$ the ``underlying" edge set of $D_0$, which is the set of size $O(h\cdot n^{1+\frac{1}{k}})$ that contains every output path of $D_0$.

Define the graph $K=(S,\binom{S}{2})$ to be the complete graph over the vertices of $S$, where the weight of an edge $(x,y)$ is defined to be $d_G(x,y)$. Let $(H_E,D_E)$ be an interactive emulator for $K$, with parameter
\[k_1=\left\lceil\frac{k(1+2\epsilon)}{h}\right\rceil~,\]
instead of $k$. That is, $(H_E,D_E)$ has stretch $\alpha_E\cdot k_1$, query time $q(k_1)$ and size $O(k_1^{\delta}|S|^{1+\frac{1}{k_1}})$. The set $H_E$ can be viewed as a set of pairs of vertices in $G$. The size of this interactive emulator is
\begin{eqnarray*}
O(k_1^{\delta}|S|^{1+\frac{1}{k_1}})&=&O(k^{\delta}(n^{1-\frac{h}{k}})^{1+\frac{1}{k_1}})
=O(k^{\delta}\cdot n^{1+\frac{1}{k_1}-\frac{h}{k}})
=O(k^{\delta}\cdot n^{1+\frac{h}{(1+2\epsilon)k}-\frac{h}{k}})\\
&=&O(k^{\delta}\cdot n^{1-\frac{2\epsilon\cdot h}{(1+2\epsilon)k}})
=O(k^{\delta}\cdot n^{1-\frac{\epsilon\cdot h}{k}})\\
&=&O(k^{\delta}\cdot n^{1-\frac{\sigma\cdot\log\log n}{\log n}})
=O\left(k^{\delta}\cdot\frac{n}{\log^{\sigma}n}\right)~,
\end{eqnarray*}
and therefore the size of $H_E$ is
\begin{equation} \label{eq:EmulatorSize}
    |H_E|=O\left(k^{\delta}\cdot\frac{n}{\log^{\sigma}n}\right)~.
\end{equation}

Now, construct an interactive distance preserver $(H_P,D_P)$ on the graph $G$ and the set of pairs $\mathcal{P}=H_E$, with stretch $\alpha_P$, query time $O(1)$ and size 
\begin{eqnarray*}
O(k^{\tau}\cdot|H_E|+n\log k+n^{1+\frac{1}{k}})
&=&O(k^{\tau}\cdot k^{\delta}\cdot\frac{n}{\log^{\sigma}n}+n\log k+n^{1+\frac{1}{k}})\\
&=&O(k^{\sigma}\cdot\frac{n}{\log^{\sigma}n}+n\log k+n^{1+\frac{1}{k}})~.
\end{eqnarray*}
Since $k\leq\log n$, the first term above is bounded by $n$, and therefore the total size is
\begin{equation} \label{eq:PreserverSize}
    O(n\log k+n^{1+\frac{1}{k}})~.
\end{equation}

Next, define a new oracle $D$ that stores the oracles $D_0,D_E$ and $D_P$. Given a query $(u,v)\in V^2$, the oracle $D$ first applies the oracle $D_0$ on this query. The outcome is one of the following. In the first case, $D_0$ outputs a $u-v$ path in $E_0$ with stretch at most $2h+1$. In this case, the oracle $D$ simply returns the same path as an output. In the second case, $D_0$ returns two paths, $P_{u,u'}\subseteq E_0$ from $u$ to $u'\in S$ and $P_{v,v'}\subseteq E_0$ from $v$ to $v'\in S$, both with weight at most $h\cdot d_G(u,v)$. In this case, $D$ queries the oracle $D_E$ on $(u',v')$ to find a path $P_{u',v'}\subseteq H_E$ between $u',v'$. Then it replaces each edge $(x,y)\in P_{u',v'}$ by the path $D_P(x,y)\subseteq H_P$ returned by the oracle $D_P$. The concatenation of $P_{u,u'}$, these paths, and $P_{v',v}$, is a path in $E_0\cup H_P\subseteq E$ between $u,v$, and we return it as an output. Here, $P_{v',v}$ is the same path as $P_{v,v'}$, but with the order of the edges reversed. See Figure \ref{fig:OraclesComposing} for an illustration.

\begin{figure}[!ht]
\centerline{\includegraphics[width=8.7cm, height=6cm]{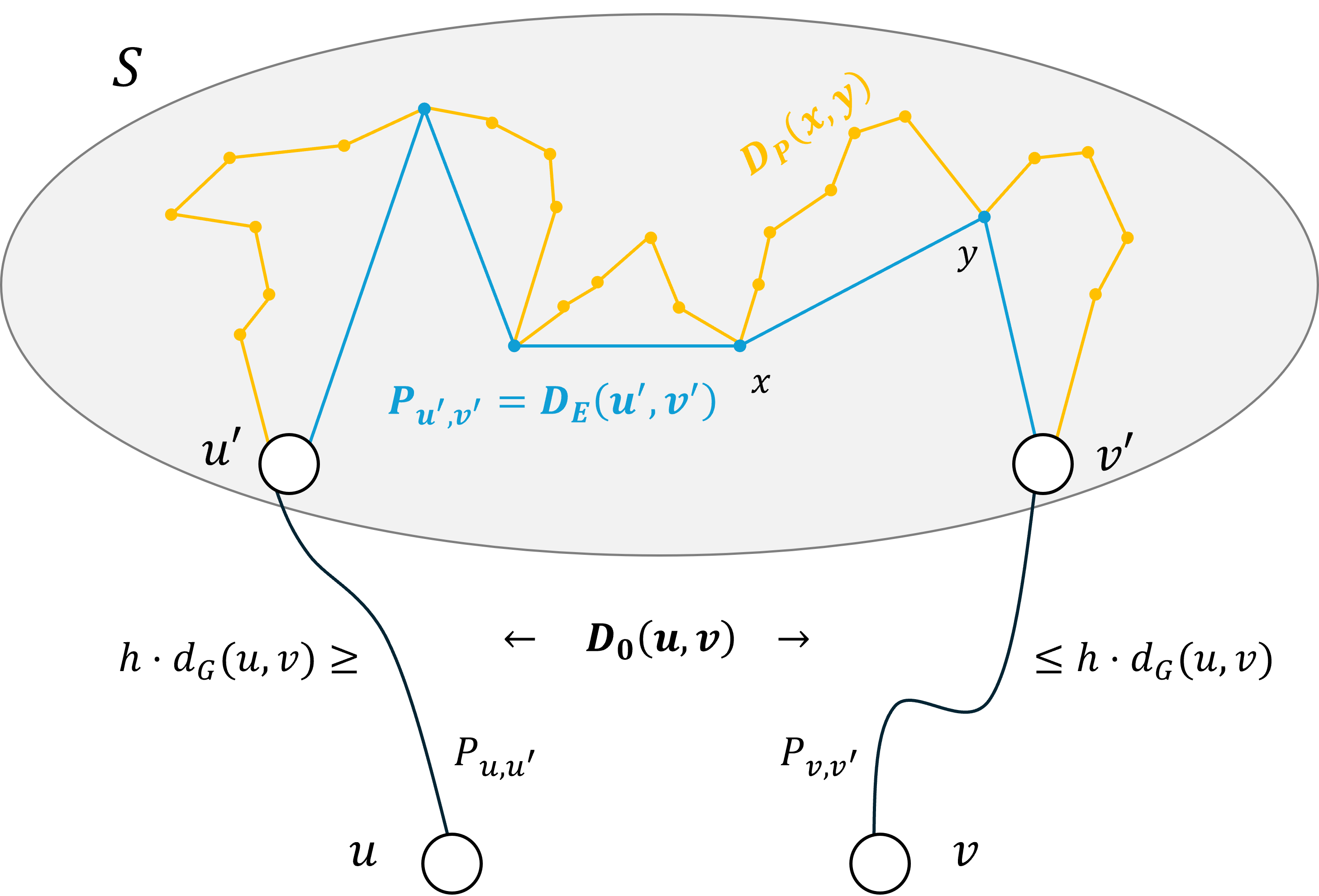}}
\caption{An illustration for the way that the oracle $D$ acts on a query $(u,v)$. First, it invokes the partial TZ oracle $D_0$ on the query $(u,v)$. In case that $D_0$ did not output a $u-v$ path, it outputs two vertices $u',v'\in S$ and two paths - $P_{u,u'}$ from $u$ to $u'$, and $P_{v,v'}$ from $v$ to $v'$. Each of these paths (depicted by black lines) has weight at most $h\cdot d_G(u,v)$, and thus, by the triangle inequality, $d_G(u',v')\leq(2h+1)d_G(u,v)$. Then, the oracle $D$ invokes the interactive emulator $D_E$ on the query $(u',v')$ to produce a $u'-v'$ path $P_{u',v'}$ (depicted by blue lines), that is not necessarily contained in the graph $G$, but consists only of vertices of $S$, and has weight at most $\alpha_Ek_1\cdot(2h+1)d_G(u,v)$. Finally, on each edge $(x,y)$ of $P_{u',v'}$, the oracle $D$ invokes the interactive distance preserver $D_P$, to obtain an $x-y$ path in $G$ (depicted by orange lines), with weight at most $\alpha_Pd_G(x,y)$. Concatenating all these paths, for every edge $(x,y)$ in $P_{u',v'}$, results in a path in $G$, with weight at most $\alpha_P\cdot\alpha_Ek_1\cdot(2h+1)d_G(u,v)$. Together with $P_{u,u'}$ and $P_{v,v'}$, we obtain a $u-v$ path in $G$, with weight at most $\left(2h+\alpha_P\cdot\alpha_Ek_1\cdot(2h+1)\right)d_G(u,v)$.}
\label{fig:OraclesComposing}
\end{figure}

\pagebreak

Our interactive spanner is now defined as $(E_0\cup H_P,D)$. The resulting path described above is of weight at most 
\begin{eqnarray*}
&&h\cdot d_G(u,v)+\sum_{e\in P_{u',v'}}\alpha_P\cdot w(e)+h\cdot d_G(u,v)\\
&=&2h\cdot d_G(u,v)+\alpha_P\cdot w(P_{u',v'})\\
&\leq&2h\cdot d_G(u,v)+\alpha_P\cdot\alpha_E\cdot k_1\cdot d_G(u',v')\\
&\leq&2h\cdot d_G(u,v)+\alpha_P\cdot\alpha_E\cdot k_1\cdot(d_G(u',u)+d_G(u,v)+d_G(v,v'))\\
&\stackrel{\text{Lemma \ref{lemma:PartialPRDO}}}{\leq}&2h\cdot d_G(u,v)+\alpha_P\cdot\alpha_E\cdot k_1\cdot(2h+1)d_G(u,v)\\
&=&(2h+\alpha_P\cdot\alpha_E\cdot k_1\cdot(2h+1))d_G(u,v)\\
&=&\left(2h+\alpha_P\cdot\alpha_E\cdot\left\lceil\frac{k(1+2\epsilon)}{h}\right\rceil\cdot(2h+1)\right)d_G(u,v)\\
&\leq&\alpha_P\cdot\alpha_E\cdot\left(2h+\left\lceil\frac{k(1+2\epsilon)}{h}\right\rceil\cdot(2h+1)\right)d_G(u,v)\\
&\leq&\alpha_P\cdot\alpha_E\cdot\left(4h+1+\frac{k(1+2\epsilon)(2h+1)}{h}\right)d_G(u,v)\\
&\leq&\alpha_P\cdot\alpha_E\cdot\left(5h+\frac{k(1+2\epsilon)(2h+1)}{h}\right)d_G(u,v)~.
\end{eqnarray*}

Note that since $h\geq\frac{1}{\epsilon}$ (see (\ref{eq:HDef})), and since $1+2\epsilon\leq2$ (because $\epsilon\leq\frac{1}{2}$), we have
\begin{eqnarray*}
\frac{(1+2\epsilon)(2h+1)}{h}=2(1+2\epsilon)+\frac{1+2\epsilon}{h}
\leq2(1+2\epsilon)+\frac{2}{1/\epsilon}=2(1+3\epsilon)~.    
\end{eqnarray*}


Thus, the weight of the resulting path is at most
\begin{eqnarray*}
\alpha_P\cdot\alpha_E\cdot\left(5h+\frac{k(1+2\epsilon)(2h+1)}{h}\right)d_G(u,v)
&\leq&\alpha_P\cdot\alpha_E\cdot\left(5h+2(1+3\epsilon)k\right)d_G(u,v)\\
&=&2\alpha_P\cdot\alpha_E\cdot\left(\frac{5}{2}h+(1+3\epsilon)k\right)d_G(u,v)~.
\end{eqnarray*}

Using the same proof, but with $\frac{\epsilon}{3}$ instead of $\epsilon$, we get that our interactive spanner has stretch
\begin{eqnarray*}
\max\{2h+1,2\alpha_P\cdot\alpha_E\cdot(O(h)+(1+\epsilon)k)\}
=2\alpha_P\cdot\alpha_E\cdot(O(h)+(1+\epsilon)k)~,
\end{eqnarray*}
as desired.


The size of $D$ is the sum of the sizes of $D_0,D_E$ and $D_P$. The size of $D_0$, by Lemma \ref{lemma:PartialPRDO}, is $O(h\cdot n^{1+\frac{1}{k}})$. The size of $D_E$ has the same bound as of the size of $H_E$, which is $O\left(k^{\delta}\cdot\frac{n}{\log^{\sigma}n}\right)=O(n)$ (by Equation (\ref{eq:EmulatorSize})). Lastly, the size of $D_P$ is, by (\ref{eq:PreserverSize}), $O(n\log k+n^{1+\frac{1}{k}})$. We get that the total size of the oracle $D$ is
\[O(n\log k+h\cdot n^{1+\frac{1}{k}})~.\]
In addition, note that the output paths of $D$ are always contained in the set $E_0\cup H_P$, which has the same size bound as for the oracle $D$. We conclude that our interactive spanner $(E_0\cup H_P,D)$ is of size $O(n\log k+h\cdot n^{1+\frac{1}{k}})$.

Regarding the query time, notice that finding $P_{u,u'}$ and $P_{v,v'}$ requires $O(\log h)$ time (by Lemma \ref{lemma:PartialPRDO}).
Finding the path $P_{u',v'}$ requires $q(k_1)$ time, and applying $D_P$ on every edge requires linear time in the size of the returned path. Hence, the query time of the resulting oracle is 
\[q(k_1)+O(\log h)\leq q(k)+O(\log h)~.\]

\end{proof}

The following theorems follow from Lemma \ref{lemma:InteractiveSpannerVariety}, when using the emulators and distance preservers from Table \ref{table:EmulatorsAndPreservers}.

\begin{theorem} \label{thm:InteractiveSpanner1}
For every $n$-vertex graph and parameters $k\in[3,\log_2n]$ and $\epsilon$ such that $\log^{-O(1)}\log n\leq\epsilon\leq\frac{1}{2}$, there is an interactive spanner with stretch $(4+\epsilon)k+O(\frac{1}{\epsilon}\cdot\log k\cdot\log\log k)$, query time $O(\log k+\log\frac{1}{\epsilon})$, and size
\[O\left(\left\lceil\frac{k\cdot\log\log n\cdot\log^{(3)}n}{\epsilon\log n}\right\rceil\cdot n^{1+\frac{1}{k}}\right)~.\]

If also $\epsilon\geq\sqrt{\frac{\log k\cdot\log\log k}{k}}$, then we also have an interactive spanner with stretch $(4+\epsilon)k$, query time $O(\log k)$, and the same size as above.
\end{theorem}

\begin{proof}

We use the interactive emulator that is based on \cite{TZ01,WN13} (see Table \ref{table:EmulatorsAndPreservers}), the interactive distance preserver from Theorem \ref{thm:DistancePreserver1}, and the notations from Lemma \ref{lemma:InteractiveSpannerVariety}, all with $\frac{\epsilon}{12}$ instead of $\epsilon$. Then, we have $\alpha_E=2$, $\delta=1$, and $q(k)=O(\log k)$, $\alpha_P=1+\frac{\epsilon}{12}$, $\tau=\log_k(\gamma_{4/3}(\frac{\epsilon}{12},k))=O(\log\log k+\log\frac{1}{\epsilon})$.

Thus, $\sigma=\delta+\tau=O(\log\log k+\log\frac{1}{\epsilon})=O(\log^{(3)}n)$ (since $k\leq\log n$ and $\epsilon\geq\log^{-O(1)}\log n$). Therefore, using the notation $h$ from Lemma \ref{lemma:InteractiveSpannerVariety} (see Equation (\ref{eq:HDef})),
\begin{eqnarray*}
h=O\left(\left\lceil\frac{k\cdot\log\log n\cdot\log^{(3)}n}{\epsilon\log n}\right\rceil\right)
=O\left(\frac{1}{\epsilon}\cdot\log k\cdot\log\log k\right)~,
\end{eqnarray*}
where the last step is true since $\frac{\log x\cdot\log\log x}{x}$ is a decreasing function for all $x\geq8$.

By Lemma \ref{lemma:InteractiveSpannerVariety}, there is an interactive spanner with stretch 
\begin{eqnarray*}
2\left(1+\frac{\epsilon}{12}\right)\cdot2\left(\left(1+\frac{\epsilon}{12}\right)k+O(h)\right)
=4\left(1+\frac{\epsilon}{12}\right)^2k+O(h)\leq(4+\epsilon)k+O(h)~,
\end{eqnarray*}
query time $q(k)+O(\log h)=O(\log k+\log\frac{1}{\epsilon})$, and size
\begin{eqnarray*}
O(n\log k+h\cdot n^{1+\frac{1}{k}})
=O\left(\left\lceil\frac{k\cdot\log\log n\cdot\log^{(3)}n}{\epsilon\log n}\right\rceil\cdot n^{1+\frac{1}{k}}\right)~,
\end{eqnarray*}
as desired (for the last equation, note that for $k\leq\frac{\log n}{\log^{(3)}n}$ we have $n\log k\leq n^{1+\frac{1}{k}}$, and for $\frac{\log n}{\log^{(3)}n}<k\leq\log n$ we have $\log k\leq\left\lceil\frac{k\cdot\log\log n\cdot\log^{(3)}n}{\epsilon\log n}\right\rceil$).

If $\epsilon\geq\sqrt{\frac{\log k\cdot\log\log k}{k}}$, then we also have $\log k\cdot\log\log k\leq\epsilon^2k$, thus
$h=O\left(\frac{\log k\cdot\log\log k}{\epsilon}\right)=O(\epsilon k)~.$
Hence, the stretch is $(4+\epsilon)k+O(h)=(4+\epsilon+O(\epsilon))k=(4+O(\epsilon))k$. By replacing $\epsilon$ with $\frac{\epsilon}{C}$, for a sufficiently large constant $C$, we get a stretch of $(4+\epsilon)k$. Since in this case we also have $\frac{1}{\epsilon}\leq k$, the query time is $O(\log k)$.

\end{proof}

In the following theorem, the precise constant coefficient of the stretch is not specified. Instead it appears in Table \ref{table:InteractiveSpannersVariety}.

\begin{theorem} \label{thm:InteractiveSpanner2}
For every $n$-vertex graph and an integer parameter $k\in[3,\log_2n]$, there is an interactive spanner with stretch 
$O(k)$, query time $O\left(\log\left\lceil\frac{k\cdot\log\log n}{\log n}\right\rceil\right)=O(\log\log k)$, and size
\[O\left(\left\lceil\frac{k\cdot\log\log n}{\log n}\right\rceil\cdot n^{1+\frac{1}{k}}\right)~.\]
\end{theorem}

\begin{proof}

We use the interactive emulator from Theorem \ref{thm:InteractiveEmulator}, the interactive distance preserver from Remark \ref{remark:PreserverWithWorseStretch}, and the notations from Lemma \ref{lemma:InteractiveSpannerVariety}, all with some constant $\epsilon$. Then, we have $\alpha_E=O(1)$, $\delta=0$, $q(k)=O(1)$, $\alpha_P=3+\epsilon=O(1)$, and $\tau=\log_{4/3}(12+\frac{40}{\epsilon})=O(1)$.

Thus, $\sigma=O(1)$ and
\[h=O\left(\left\lceil\frac{k\cdot\log\log n}{\log n}\right\rceil\right)=O\left(\log k\right)~,\]
where the last step is true since $\frac{\log x}{x}$ is a decreasing function for all $x\geq3$.

By Lemma \ref{lemma:InteractiveSpannerVariety}, there is an interactive spanner with stretch 
\[2\cdot O(1)\cdot((1+\epsilon)k+O(h))=O(k)~,\]
query time 
\[q(k)+O(\log h)=O\left(\log\left\lceil\frac{k\cdot\log\log n}{\log n}\right\rceil\right)=O(\log\log k)~,\] 
and size
\begin{eqnarray*}
O(n\log k+h\cdot n^{1+\frac{1}{k}})
=O\left(n\log k+\left\lceil\frac{k\cdot\log\log n}{\log n}\right\rceil\cdot n^{1+\frac{1}{k}}\right)~.
\end{eqnarray*}

Lastly, note that $\frac{k\cdot\log\log n}{\log n}\cdot n^{1+\frac{1}{k}}=n\cdot\frac{\log\log n}{(\log n) / k}\cdot2^{\frac{\log n}{k}}\geq n\log\log n\geq n\log k$, since $2^x>x$ for every $x$ (in particular for $x=(\log n) / k$). Therefore the term $n\log k$ is negligible, and we conclude that the size of our interactive spanner is
\begin{eqnarray*}
O\left(n\log k+\left\lceil\frac{k\cdot\log\log n}{\log n}\right\rceil\cdot n^{1+\frac{1}{k}}\right)
=O\left(\left\lceil\frac{k\cdot\log\log n}{\log n}\right\rceil\cdot n^{1+\frac{1}{k}}\right)~.
\end{eqnarray*}

Note that as $k\leq\log n$, we have $\log k\geq\frac{k\log\log n}{\log n}$, and thus the size is $O(\log k\cdot n^{1+\frac{1}{k}})$.

\end{proof}

Similarly to Theorems \ref{thm:InteractiveSpanner1} and \ref{thm:InteractiveSpanner2}, more results can be derived by Lemma \ref{lemma:InteractiveSpannerVariety}, when using different combinations from Table \ref{table:EmulatorsAndPreservers}. We specify these results in Table \ref{table:InteractiveSpannersVariety} without proof (the proofs are relatively simple and similar to the proofs of Theorems \ref{thm:InteractiveSpanner1} and \ref{thm:InteractiveSpanner2}). Note that in this table, we assume that the parameter $\epsilon$ is constant. For a precise specification of the dependencies on $\epsilon$, see Table \ref{table:InteractiveSpannersVarietyDetails} in \nameref{sec:AppendixG}.

\def\arraystretch{1.5}

\begin{table*}[ht]
\begin{center}
\begin{tabular}{|c|c|c|c|c|}
\hline
Stretch  & Size  & Query  & Emulator & Distance \\ 
         &       & Time   &          & Preserver \\ \hline
$(4+\epsilon)k$  & $\left\lceil\frac{k\cdot\log\log n\cdot\log^{(3)}n}{\log n}\right\rceil\cdot n^{1+\frac{1}{k}}$  & $\log k$ & \cite{TZ01,WN13} & Theorem \ref{thm:DistancePreserver1}\\ \hline
$(8+\epsilon)k$  & $\left\lceil\frac{k\cdot\log\log n\cdot\log^{(3)}n}{\log n}\right\rceil\cdot n^{1+\frac{1}{k}}$ & $\log\log k$ & Theorem \ref{thm:APEmulator}, based on \cite{AP90a} & Theorem \ref{thm:DistancePreserver1}\\
\hline
$(12+\epsilon)k$  & $\left\lceil\frac{k\cdot\log\log n}{\log n}\right\rceil\cdot n^{1+\frac{1}{k}}$  & $\log k$ & \cite{TZ01,WN13} & Remark \ref{remark:PreserverWithWorseStretch}\\ \hline
$(24+\epsilon)k$  & $\left\lceil\frac{k\cdot\log\log n}{\log n}\right\rceil\cdot n^{1+\frac{1}{k}}$  & $\log\log k$ & Theorem \ref{thm:APEmulator} & Remark \ref{remark:PreserverWithWorseStretch}\\ \hline
$(16c_{MN}+\epsilon)k$  & $\left\lceil\frac{k\cdot\log\log k\cdot\log\log n}{\log n}\right\rceil\cdot n^{1+\frac{1}{k}}$  & $\log\log k$ & Theorem \ref{thm:InteractiveEmulator} & Theorem \ref{thm:DistancePreserver1}\\ \hline
$(48c_{MN}+\epsilon)k$  & $\left\lceil\frac{k\cdot\log\log n}{\log n}\right\rceil\cdot n^{1+\frac{1}{k}}$  & $\log\log k$ & Theorem \ref{thm:InteractiveEmulator} & Remark \ref{remark:PreserverWithWorseStretch}\\ \hline
\end{tabular}
\end{center}
\caption{A variety of results for interactive spanners by Lemma \ref{lemma:InteractiveSpannerVariety}, when using the interactive emulators and distance preservers from Table \ref{table:EmulatorsAndPreservers}. The $O$-notations are omitted from the size and query time columns. For simplicity, in all of the above, the parameter $\epsilon>0$ is considered to be constant. In the second and the fourth rows we assume $\Lambda=n^{O(1)}$. In rows $5$ and $6$ the query times are in fact at most $O\left(\log\left\lceil\frac{k\log\log k\cdot\log\log n}{\log n}\right\rceil\right)$ and $O\left(\log\left\lceil\frac{k\cdot\log\log n}{\log n}\right\rceil\right)$, respectively.} \label{table:InteractiveSpannersVariety}
\end{table*}

For a general aspect ratio $\Lambda$, the size bounds in the second and the fourth rows are \newline 
$O\left(\left\lceil\frac{k\cdot\log\log n}{\log n}\left(\log\log k+\frac{\log\log\Lambda}{\log k}\right)\right\rceil\cdot n^{1+\frac{1}{k}}\right)$ and $O\left(\left\lceil\frac{k\log\log\Lambda\cdot\log\log n}{\log n\cdot\log k}\right\rceil\cdot n^{1+\frac{1}{k}}\right)$, respectively. The query times are both at most $O(\log\log k+\log^{(3)}\Lambda)$.
These more general bounds imply the bounds that appear in the table (for the case where $\Lambda=n^{O(1)}$), since if $k\leq\sqrt{\log n}$, then for the second row we have
\begin{eqnarray*}
h&=&O\left(\left\lceil\frac{(\log\log k+\log_k\log\Lambda)\cdot k\cdot\log\log n}{\log n}\right\rceil\right)\\
&=&O\left(\left\lceil\frac{\log\log n\cdot k\cdot\log\log n}{\log n}\right\rceil\right)=O(1)~.
\end{eqnarray*}
Otherwise, if $k>\sqrt{\log n}$, then $\log_k\log\Lambda=O\left(\frac{\log\log n}{\log\log n}\right)=O(1)$, and therefore 
\begin{eqnarray*}
h&=&O\left(\left\lceil\frac{k\cdot\log\log k\cdot\log\log n}{\log n}\right\rceil\right)
=O\left(\left\lceil\frac{k\cdot\log\log n\cdot\log^{(3)}n}{\log n}\right\rceil\right)~.
\end{eqnarray*}
Similarly, for the fourth row, we have either $h=O(1)$ (for small $k$) or $h=O\left(\left\lceil\frac{k\cdot\log\log n}{\log n}\right\rceil\right)$ (for large $k)$.

Note that for constant $\epsilon$, all our interactive spanners have size of $\Omega\left(\left\lceil\frac{k\cdot\log\log n}{\log n}\right\rceil\cdot n^{1+\frac{1}{k}}\right)$. As we saw in the proof of Theorem \ref{thm:InteractiveSpanner2}, we have $n\log k\leq\left\lceil\frac{k\cdot\log\log n}{\log n}\right\rceil\cdot n^{1+\frac{1}{k}}$, therefore the term $n\log k$ that appears in Lemma \ref{lemma:InteractiveSpannerVariety}, in the size of the interactive spanners, is negligible. 
See also \nameref{sec:AppendixG} for explicit dependencies of these bounds on $\epsilon$.

\subsubsection{Linear-size Interactive Spanner} \label{sec:LinearInteractiveSpanner}

In order to construct a \textit{linear-size} oracle (as opposed to size $O(n\log\log n)$ that can be obtained from Theorem \ref{thm:InteractiveSpanner2}), we use a technique that is based on the results of Bezdrighin et al. \cite{BEGGHIV22}. In this paper, the authors presented the notion of \textit{stretch-friendly partitions}, defined as follows. In the following, given a graph $G$ and a subset $U$ of its vertices, $G[U]$ denotes the sub-graph of $G$ induced by the vertices of $U$.

\begin{definition} \label{def:StretchFriendly}
Let $G=(V,E)$ be an undirected weighted graph, and fix some $t>0$. A \textbf{stretch-friendly $t$-partition} of $G$ is a partition $V=\bigcup_{i=1}^qV_i$ ($\{V_i\}_{i=1}^q$ are called \textbf{clusters}), such that for every $i=1,...,q$,
\begin{enumerate}
    \item There is a spanning tree $T_i$ of $G[V_i]$, rooted at some $r_i\in V_i$, such that for every $v\in V_i$, the unique path in $T_i$ between $v$ and $r_i$ has at most $t$ edges.
    \item If $(x,y)\in E$ is such that $x\in V_i$ and $y\notin V_i$, then the weight of every edge on the unique path in $T_i$ between $x$ and $r_i$ is at most $w(x,y)$.
    \item If $(x,y)\in E$ is such that $x,y\in V_i$, then the weight of every edge on the unique path in $T_i$ between $x$ and $y$ is at most $w(x,y)$.
\end{enumerate}
\end{definition}

In \cite{BEGGHIV22}, it was proved that for every $t\geq1$, there is a polynomial-time algorithm that computes a stretch-friendly $O(t)$-partition with at most $\frac{n}{t}$ clusters. We use this construction to show that one can reduce the size of an interactive spanner, at the cost of increasing its stretch. A similar proof is presented in \cite{BEGGHIV22}, where the authors prove that one can reduce the size of a spanner by increasing its stretch. However, in the case of \textit{interactive} spanners, this reduction (that relies on stretch-friendly partitions) requires more care, since we also have to prove that the approximate shortest paths in the resulting spanner can be reported efficiently.

The proof of the following theorem is deferred to \nameref{sec:AppendixE}.

\begin{theorem} \label{thm:UseStretchFriendly}
Suppose that any $n$-vertex graph admits an interactive $\alpha$-spanner with query time $q$ and size $h(n)$. Then, given a positive number $t\geq1$, every $n$-vertex graph $G$ also admits an interactive $O(\alpha\cdot t)$-spanner with query time $O(q)$ and size $O(h(\frac{n}{t})+n)$.
\end{theorem}

Given Theorem \ref{thm:UseStretchFriendly}, we apply it to our interactive spanner from Theorem \ref{thm:InteractiveSpanner2}, that has stretch $O(k)$, query time $O\left(\log\left\lceil\frac{k\cdot\log\log n}{\log n}\right\rceil\right)$ and size $O\left(\left\lceil\frac{k\cdot\log\log n}{\log n}\right\rceil\cdot n^{1+\frac{1}{k}}\right)$. We use $t=\left\lceil\frac{k\cdot\log\log n}{\log n}\right\rceil$, and get the following result.

\begin{theorem} \label{thm:InteractiveSpanner3}
For every $n$-vertex graph and an integer parameter $k\in[3,\log_2n]$, there is an interactive spanner with stretch $O\left(k\cdot\left\lceil\frac{k\cdot\log\log n}{\log n}\right\rceil\right)=O(k\log k)$, query time $O\left(\log\left\lceil\frac{k\cdot\log\log n}{\log n}\right\rceil\right)=O(\log\log k)$ and size $O(n^{1+\frac{1}{k}})$.
\end{theorem}

For $k=O\left(\frac{\log n}{\log\log n}\right)$, the stretch of this interactive spanner is $O(k)$, and its query time is $O(1)$.

\subsubsection{Ultra-Sparse Interactive Spanner}

To produce ultra-sparse interactive spanners, we refine Theorem \ref{thm:UseStretchFriendly} in the following way (the proof appears in \nameref{sec:AppendixE}).

\begin{theorem} \label{thm:UltraSparse}
Suppose that any $n$-vertex graph admits an interactive $\alpha$-spanner with query time $q$ and size $h(n)$. Then, given a positive number $t\geq1$, every $n$-vertex graph $G$ also admits an interactive $O(\alpha\cdot t)$-spanner with query time $O(q+t)$ and size $n+O(h(\frac{n}{t}))$.
\end{theorem}

Along with Theorem \ref{thm:InteractiveSpanner3} (with $k=\Theta(\log n)$), this theorem implies the following corollary.

\begin{corollary} \label{cor:UltraSparsePRDO1}
For any parameter $t\geq1$, there exists an interactive $O(\log n\cdot\log\log n\cdot t)$-spanner with query time $O(\log\log n+t)$ and size $n+O(\frac{n}{t})$.
\end{corollary}

By setting $t=\omega(1)$, we obtain an ultra-sparse PRDO with stretch $\tilde{O}(\log n)$, query time $O(\log\log n)$ (as long as $t=O(\log\log n)$), and size $n+o(n)$.

\subsection*{Acknowledgements}
We wish to thank Arnold Filtser for reminding us that spanning metric Ramsey trees \cite{ACEFN20} can be viewed as a PRDO, and for helpful discussions about spanning clan embeddings.

\bibliography{hopset}

\appendix

\section{Appendix A} \label{sec:AppendixA}

In two separate papers, by Elkin and Neiman and by Huang and Pettie \cite{EN19,HP17}, the authors proved that the set that is denoted in Section \ref{sec:HierarchyOfSets} by $\bar{H}$ is a $(1+\epsilon,\beta_l)$-hopset where $\beta_l=O(\frac{l}{\epsilon})^{l-1}$.

In this section, we prove that $\bar{H}^{1/2}$ is also a $(1+\epsilon,\beta_l)$-hopset, albeit the constant hidden in the $O$-notation in $\beta_l=O(\frac{l}{\epsilon})^{l-1}$ becomes larger when we use half-bunches instead of bunches. We use essentially the same proof from \cite{EN19,HP17}, while rephrasing some of its concepts. Since $\bar{H}^{1/2}\subseteq\bar{H}$, this also proves (up to the constant mentioned above) the original result of \cite{EN19,HP17} about $\bar{H}$.


We also show that $\bar{H}^{1/2}$ is a $(3+\epsilon,k^{O(\log\frac{1}{\epsilon})})$-hopset. This proof generalizes a result from \cite{EGN22}, where it was shown that this is the case for $\bar{H}^1$, and that in unweighted graphs, $\bar{H}^{1/2}$ is a $(3+\epsilon,k^{O(\log\frac{1}{\epsilon})})$-spanner (more precisely, a version of $\bar{H}^{1/2}$, where shortest paths are added instead of virtual edges).

First, we restate and prove Claim \ref{claim:ENHopset}.

\begin{claim*}[Claim \ref{claim:ENHopset}]
Given an undirected weighted graph $G=(V,E)$ and a positive parameter $\epsilon\leq O(l)$, the set $\bar{H}^{1/2}$ is a $(1+\epsilon,O(\frac{l}{\epsilon})^{l-1})$-hopset.
\end{claim*}

\begin{remark}
Note that the definition of $\bar{H}^{1/2}$ is independent of $\epsilon$, i.e., $\bar{H}^{1/2}$ is a $(1+\epsilon,O(\frac{l}{\epsilon})^{l-1})$-hopset for all $\epsilon=O(l)$ \textbf{simultaneously}.
\end{remark}

\begin{proof}

Let $h>1$ be some integer. Fix $u,v\in V$ and let $P$ be the shortest path between $u,v$. In the following, given two vertices $x,y\in P$ such that $x$ is closer to $u$ than $y$ is, we denote by $[x,y]$ the sub-path of $P$ between $x$ and $y$.

We define a \textit{partition tree} for $P$ as follows. The partition tree has $l$ levels, and each node in the tree is a sub-path of $P$. The root of the tree ($(l-1)$-th level) is the path $P$ itself. The children of each node $Q$ in level $i>0$ are defined by partitioning $Q$ into at most $2h$ disjoint sub-paths, each of them is either a single edge, or has weight at most $\frac{w(Q)}{h}$. This can be done by setting $u_0$ to be the first vertex in $Q$, and then, for every $j\geq0$, setting $u_{j+1}$ to be the first vertex after $u_j$ such that $w([u_j,u_{j+1}])>\frac{w(Q)}{h}$. The sub-path $[u_j,u_{j+1}]$, without its last edge, has weight at most $\frac{w(Q)}{h}$, and there are at most $h$ such sub-paths. We add a child to $Q$ for the last edge of $[u_j,u_{j+1}]$, and another child for the rest of this sub-path (if it is not empty).

For every node in the partition tree, we will assign a label \textit{Good} or \textit{Bad}. For each level $i=0,1,...,l-1$ we will later define parameters $\alpha_i,\beta_i,r_i$, such that the following invariant is kept for all nodes in the partition tree (here $d(x,y)$ is a shortened notation for $d_G(x,y)$, for every $x,y\in V$).

\begin{invariant} \label{inv:Invariant1}
For every node $[x,y]$ in the $i$-th level of the partition tree,
\begin{enumerate} \label{enum:Invariant}
    \item If $[x,y]$ is \textit{Good}, then
    \[d_{G\cup\bar{H}^{1/2}}^{(\beta_i)}(x,y)\leq\alpha_i\cdot d(x,y)~.\]
    \item If $[x,y]$ is \textit{Bad}, then
    \[d(x,p_{i+1}(x)),d(y,p_{i+1}(y))\leq r_i\cdot d(x,y)~.\]
\end{enumerate}
\end{invariant}

We now define the labels \textit{Good} and \textit{Bad}, of which we will assign each node in the partition tree. Accordingly, we will later determine the values of the parameters $\alpha_i,\beta_i$ and $r_i$, so that all the nodes satisfy Invariant \ref{inv:Invariant1}.

\begin{definition}[\textit{Good} and \textit{Bad} leaves] \label{def:GoodBadLeaves}
Let $[x,y]$ be a leaf in the partition tree, i.e., a vertex in level $0$. If $x,y$ are connected with an edge or a hopset edge, we assign the label \textit{Good} to $[x,y]$, and otherwise \textit{Bad}.    
\end{definition}
Note that if the leaf $[x,y]$ is \textit{Good}, it means that
\[d_{G\cup\bar{H}^{1/2}}^{(1)}(x,y)=d(x,y)~.\]
Thus we can choose any $\alpha_0,\beta_0\geq1$ and item $1$ of Invariant \ref{inv:Invariant1} is satisfied.

If $[x,y]$ is \textit{Bad}, that means that $y\notin B^{1/2}_0(x)$ and $x\notin B^{1/2}_0(y)$, which implies $d(x,p_1(x)),d(y,p_1(y))\leq2d(x,y)$. Therefore we can choose any $r_0\geq2$ and item $2$ of Invariant \ref{inv:Invariant1} is satisfied.

\begin{definition}[\textit{Good} and \textit{Bad} internal nodes] \label{def:GoodBadInternal}
Let $[x,y]$ be a sub-path which is a node in the $i$-th level of the partition tree, for $i>0$. 
\begin{enumerate}
    \item If all the children of $[x,y]$ are labeled as \textit{Good}, we label $[x,y]$ as \textit{Good}.
    \item If there are children of $[x,y]$ that are labeled as \textit{Bad}, let $[a,b]$ and $[a',b']$ be the first and the last children of $[x,y]$ which are labeled as \textit{Bad}. If $p_i(a)$ and $p_i(b')$ are connected by a hop-edge, we label $[x,y]$ as \textit{Good}.
    \item If $p_i(a),p_i(b')$ are not connected by a hop-edge, we label $[x,y]$ as \textit{Bad}.
\end{enumerate}
\end{definition}

\begin{figure}[!ht]
\centerline{\includegraphics[width=13cm, height=2.6cm]{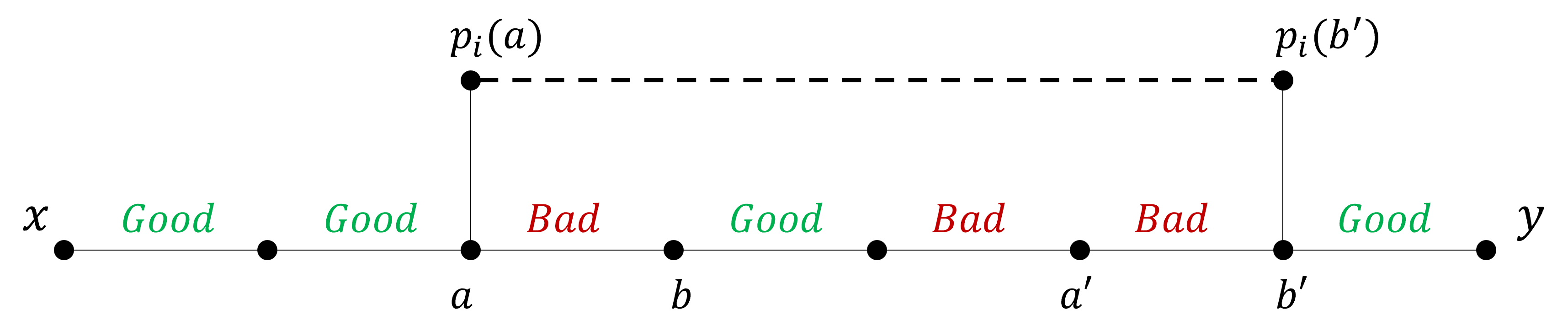}}
\caption{If the dashed line between $p_i(a),p_i(b')$ represents a hop-edge in $H$, then $[x,y]$ is labeled as \textit{Good}. Otherwise, it is labeled as \textit{Bad}.}
\label{fig:HuangPettieElkinNeiman}
\end{figure}

Before we move on to determine the values of $\alpha_i,\beta_i$ and $r_i$ such that Invariant \ref{inv:Invariant1} is satisfied, we prove the following simple fact.

\begin{fact} \label{fact:SingleEdge}
Let $[x,y]$ be a node in level $i>0$ of the partition tree, and let $[a,b]$ be a child of $[x,y]$ such that $d(a,b)>\frac{d(x,y)}{h}$. Then, $[a,b]$ is a single edge in $G$, and it is labeled as \textit{Good}.
\end{fact}

\begin{proof}

First, note that if $d(a,b)>\frac{d(x,y)}{h}$, then by the definition of the partition tree, $[a,b]$ must consist of a single edge in $G$ (recall that the children of any node $Q$ are either sub-paths of weight at most $\frac{w(Q)}{h}$, or single edges).

We prove that $Q'=[a,b]$ is \textit{Good} by induction over $i$. For $i=1$, the node $Q'=[a,b]$ is a leaf, and by Definition \ref{def:GoodBadLeaves}, it is labeled as \textit{Good} (since $a,b$ are connected by an edge of $G$). For $i>1$, note that when constructing the partition tree, $Q'=[a,b]$ is partitioned into sub-paths of weight at most $\frac{w(a,b)}{h}$ and single edges. In fact, since $[a,b]$ is a single edge, it must have only one child $Q''=[a,b]$. Since $h>1$, the weight of $Q''$ is $w(a,b)>\frac{w(a,b)}{h}$, and by the induction hypothesis we conclude that this child is labeled as \textit{Good}. By Definition \ref{def:GoodBadInternal}, the node $Q'$ is also labeled as \textit{Good}.
    
\end{proof}

For a node $[x,y]$ in the $i$-th level of the partition tree, consider the three cases in Definition \ref{def:GoodBadInternal}. In the first case, we know that every child $[a,b]$ of $[x,y]$ is \textit{Good}, and thus by Invariant \ref{inv:Invariant1}, $d_{G\cup\bar{H}^{1/2}}^{(\beta_{i-1})}(a,b)\leq\alpha_{i-1}\cdot d(a,b)$. Therefore ,
\[d_{G\cup\bar{H}^{1/2}}^{(2h\cdot\beta_{i-1})}(x,y)\leq \alpha_{i-1}\cdot d(x,y)~.\]

In the second case, we know that for every child $[s,t]$ before $[a,b]$ and after $[a',b']$ there is a path $P_{s,t}\subseteq G\cup\bar{H}^{1/2}$, between $s,t$, with length at most $\beta_{i-1}$ and weight at most $\alpha_{i-1}\cdot d(s,t)$. Also, the following is a path of $3$ edges from $a$ to $b'$ in $G\cup\bar{H}^{1/2}$: 
\[Q=(a,p_i(a))\circ(p_i(a),p_i(b'))\circ(p_i(b'),b')~.\]

Note that in both cases $2$ and $3$ in Definition \ref{def:GoodBadInternal}, since $[a,b],[a',b']$ are \textit{Bad}, we know that $d(a,b),d(a',b')\leq\frac{d(x,y)}{h}$, by Fact \ref{fact:SingleEdge}. Therefore, by Invariant \ref{inv:Invariant1},
\begin{equation}\label{eq:RDistance1}
d(a,p_i(a))\leq r_{i-1}\cdot d(a,b)\leq\frac{r_{i-1}}{h}d(x,y)\text{, and}
\end{equation}
\begin{equation}\label{eq:RDistance2}
d(b',p_i(b'))\leq r_{i-1}\cdot d(a',b')\leq\frac{r_{i-1}}{h}d(x,y)~.    
\end{equation}

Thus, in case $2$, the path $Q$ has weight
\begin{eqnarray*}
&&d(a,p_i(a))+d(p_i(a),p_i(b'))+d(b',p_i(b'))\\
&\leq&d(a,p_i(a))+d(p_i(a),a)+d(a,b')+d(b',p_i(b'))+d(b',p_i(b'))\\
&\leq&d(a,b')+4\cdot\frac{r_{i-1}}{h}d(x,y)~.
\end{eqnarray*}
We concatenate this path with the paths $P_{s,t}$, for children $[s,t]$ that appear either before $a$ or after $b'$. We obtain a path from $x$ to $y$ in $G\cup\bar{H}^{1/2}$, with length at most $(2h-1)\cdot\beta_{i-1}+3$ and weight at most
\begin{eqnarray*}
\sum_{[s,t]}\alpha_{i-1}\cdot d(s,t)+d(a,b')+4\cdot\frac{r_{i-1}}{h}d(x,y)
&\leq&\alpha_{i-1}\cdot d(x,y)+4\cdot\frac{r_{i-1}}{h}d(x,y)\\
&=&(\alpha_{i-1}+4\cdot\frac{r_{i-1}}{h})d(x,y)~.
\end{eqnarray*}

Summarizing cases $1$ and $2$, if $[x,y]$ was labeled as \textit{Good}, we have $d_{G\cup\bar{H}^{1/2}}^{(\beta_i)}(x,y)\leq\alpha_id(x,y)$, where
\begin{equation} \label{eq:AlphaConstraint}
\alpha_i=\max\{\alpha_{i-1},\alpha_{i-1}+4\cdot\frac{r_{i-1}}{h}\}=\alpha_{i-1}+4\cdot\frac{r_{i-1}}{h}\text{, and}
\end{equation}
\begin{equation}\label{eq:BetaConstraint}
\beta_i=\max\{2h\cdot\beta_{i-1},(2h-1)\cdot\beta_{i-1}+3\}=2h\cdot\beta_{i-1}~.
\end{equation}
The last equality is true when $\beta_{i-1}\geq3$. This does not happen when $i=1$; then $\beta_0=1$ and $\beta_1=\max\{2h\cdot1,(2h-1)\cdot1+3\}=2h+2$. But then, since $h>1$, for $i>1$ we do have $\beta_{i-1}\geq3$ (note that $\beta_i$ is monotonically increasing). Hence, indeed $\beta_i=2h\cdot\beta_{i-1}$ for $i>1$.

Finally, in the third case, $p_i(a),p_i(b')$ are not connected by an edge of $\bar{H}^{1/2}$, and then by inequalities (\ref{eq:RDistance1}), (\ref{eq:RDistance2}) we conclude that
\begin{eqnarray*}
d(p_i(a),p_{i+1}(p_i(a)))
&\leq&2d(p_i(a),p_i(b'))
\leq2(d(p_i(a),a)+d(a,b')+d(b',p_i(b')))\\
&\leq&2(d(a,b')+2\cdot\frac{r_{i-1}}{h}d(x,y))
\leq(2+4\cdot\frac{r_{i-1}}{h})d(x,y)~.
\end{eqnarray*}

The vertex $p_{i+1}(p_i(a))$ is in $A_{i+1}$. Hence, by the definition of $p_{i+1}(x)$ as the closest $A_{i+1}$-vertex to $x$, we have
\begin{eqnarray*}
d(x,p_{i+1}(x))
&\leq&d(x,p_{i+1}(p_i(a)))\\
&\leq&d(x,a)+d(a,p_i(a))+d(p_i(a),p_{i+1}(p_i(a)))\\
&\leq&d(x,y)+\frac{r_{i-1}}{h}d(x,y)+(2+4\cdot\frac{r_{i-1}}{h})d(x,y)\\
&=&(3+5\cdot\frac{r_{i-1}}{h})d(x,y)~.
\end{eqnarray*}
The proof of $d(y,p_{i+1}(y))\leq(3+5\cdot\frac{r_{i-1}}{h})d(x,y)$ is symmetric. Therefore, we can choose any sequence $\{r_i\}$ such that
\begin{equation} \label{eq:RConstraint}
 r_i\geq3+5\cdot\frac{r_{i-1}}{h}   
\end{equation}
holds, and item $2$ of Invariant \ref{inv:Invariant1} will be satisfied.

Assume that $h>5$ (this will be ensured in the sequel). We prove by induction that for the sequences choices $r_i=\frac{3h}{h-5}$, $\alpha_i=1+\frac{72i}{h}$ and $\beta_i=2(2h)^i$, inequalities (\ref{eq:AlphaConstraint}), (\ref{eq:BetaConstraint}), (\ref{eq:RConstraint}) hold (note that the sequence $\{r_i\}$ is independent of $i$). For $i=0$, we have $r_0=\frac{3h}{h-5}=3+\frac{15}{h-5}>2$, $\alpha_0=1$ and $\beta_0=2>1$. Thus, as we already saw, Invariant \ref{inv:Invariant1} holds for level $0$.

For $i>0$, 
\[3+5\cdot\frac{r_{i-1}}{h}=3+\frac{5}{h}\cdot\frac{3h}{h-5}=3+\frac{15}{h-5}=r_i~.\]
Hence, inequality (\ref{eq:RConstraint}) is satisfied. Also,
\begin{eqnarray*}
\alpha_{i-1}+4\cdot\frac{r_{i-1}}{h}&=&1+\frac{72(i-1)}{h}+\frac{12}{h-5}\\
&=&1+\frac{72}{h}(i-1+\frac{1}{6}\cdot\frac{h}{h-5})\\
&=&1+\frac{72}{h}(i-1+\frac{1}{6}(1+\frac{5}{h-5}))\\
&\leq&1+\frac{72}{h}(i-1+\frac{1}{6}(1+5))\\
&=&1+\frac{72i}{h}=\alpha_i~.
\end{eqnarray*}
This proves inequality (\ref{eq:AlphaConstraint}). Finally, note that $\beta_1=2h+2\leq2\cdot(2h)^1$, and for $i\geq2$, we have
\[2h\cdot\beta_{i-1}=2h\cdot2(2h)^{i-1}=2(2h)^i=\beta_i~.\]
Therefore, inequality (\ref{eq:BetaConstraint}) is satisfied. We conclude that Invariant \ref{inv:Invariant1} holds for every level $i$.

In particular, the invariant holds for the root $[u,v]$. Since $A_l=\emptyset$ and $p_l(u)$ is not defined, item $2$ from the invariant cannot be true for this node. Therefore, item $1$ holds, i.e.,
\[d_{G\cup\bar{H}^{1/2}}^{(\beta_{l-1})}(u,v)\leq\alpha_{l-1}\cdot d(u,v)~.\]

Choose $h=\lceil\frac{C\cdot l}{\epsilon}\rceil$, where $C$ is a constant such that $C\geq72$ and $\epsilon<\frac{C\cdot l}{5}$. There is such constant $C$ since we assumed that $\epsilon=O(l)$. Then, $h\geq\frac{C\cdot l}{\epsilon}>5$ and we get that our hopset has stretch 
\[\alpha_{l-1}=1+\frac{72(l-1)}{h}\leq1+\frac{72(l-1)\epsilon}{C\cdot l}\leq1+\frac{(l-1)\epsilon}{l}<1+\epsilon~,\]
and hopbound
\[\beta_{l-1}=2(2h)^{l-1}=O\left(\frac{l}{\epsilon}\right)^{l-1}~.\]

\end{proof}

The proof of the following claim is based on ideas from \cite{EGN22}. Using this claim, and a similar proof to the proof of Theorem \ref{thm:DistancePreserver1}, we show the existence of interactive $(3+\epsilon)$-distance preservers with query time $O(1)$ and size $O(|\mathcal{P}|k^{\log_{4/3}(12+\frac{40}{\epsilon})}+n\log k+n^{1+\frac{1}{k}})$ (see Remark \ref{remark:PreserverWithWorseStretch}).

\begin{claim}
Given an undirected weighted graph $G=(V,E)$ and a positive parameter $\epsilon\leq O(l)$, the set $\bar{H}^{1/2}$ is a $(3+\epsilon,(12+\frac{40}{\epsilon})^{l-1})$-hopset, for all $\epsilon=O(l)$ simultaneously.
\end{claim}

\begin{proof}

We use a similar proof technique as the one of Claim \ref{claim:ENHopset}. We fix an integer $h>5$. Consider two vertices $u,v$ in $G$. We again consider the \textit{partition tree} of the $u-v$ shortest path. We will assign a label \textit{Good} or \textit{Bad} for each node in the tree, and we will define the parameters $\{\alpha_i,\beta_i,r_i\}$ such that the same invariant as before is satisfied: if $[x,y]$ is a \textit{Good} node in the $i$-th level, then
\[d^{(\beta_i)}_{G\cup\bar{H}^{1/2}}(x,y)\leq\alpha_i\cdot d(x,y)~,\]
and if $[x,y]$ is a \textit{Bad} node in the $i$-th level, then
\[d(x,p_{i+1}(x)),d(y,p_{i+1}(y))\leq r_i\cdot d(x,y)~.\]
Recall that $d(\cdot,\cdot)$ denotes the distance in $G$.

The definition of \textit{Good} and \textit{Bad} nodes in the $0$-th level is the same as before: the leaf $[x,y]$ is \textit{Good} if and only if $(x,y)$ is either an original edge of $G$ or an edge of $\bar{H}^{1/2}$. Then we know that if $\alpha_0,\beta_0\geq1,r_0\geq2$, the invariant is satisfied for $i=0$.

For $[x,y]$ in the $i$-th level, we define the labels \textit{Good} and \textit{Bad} in a slightly different way. We again consider three possible cases:

\begin{enumerate}
    \item If all of the children of $[x,y]$ are labeled as \textit{Good}, $[x,y]$ is also labeled as \textit{Good}.
    \item If there are children of $[x,y]$ that are labeled as \textit{Bad}, and $p_i(x),p_i(y)$ are connected by a hop-edge, we still say that $[x,y]$ is \textit{Good}.
    \item If $p_i(x),p_i(y)$ are not connected by a hop-edge (and $[x,y]$ has \textit{Bad} children), we label $[x,y]$ as \textit{Bad}.
\end{enumerate}

Note that in the first case, the rule for labeling a node as \textit{Good} is the same as in the proof of Claim \ref{claim:ENHopset}. Hence, Fact \ref{fact:SingleEdge} still holds.

In the first case, similarly to the proof of Claim \ref{claim:ENHopset}, we have $d^{(2h\cdot\beta_{i-1})}_{G\cup\bar{H}^{1/2}}(x,y)\leq\alpha_{i-1}\cdot d(x,y)$. In the second case, note that the path
\[Q=(x,p_i(x))\circ(p_i(x),p_i(y))\circ(p_i(y),y)\]
is a three-edges path in $G\cup\bar{H}^{1/2}$ between $x,y$. For bounding its weight, consider some \textit{Bad} child $[a,b]$ of $[x,y]$. By the invariant, we know that $d(a,p_i(a))\leq r_{i-1}\cdot d(a,b)\leq\frac{r_{i-1}}{h}d(x,y)$. The last inequality is true since $[a,b]$ must be a sub-path of weight at most $\frac{d(x,y)}{h}$ (by Fact \ref{fact:SingleEdge}). As a result,
\begin{equation} \label{eq:RDistance3}
    d(x,p_i(x))\leq d(x,p_i(a))\leq d(x,a)+d(a,p_i(a))
    \leq d(x,a)+\frac{r_{i-1}}{h}d(x,y)~.
\end{equation}
Similarly, $d(y,p_i(y))\leq d(y,a)+\frac{r_{i-1}}{h}d(x,y)$.

Therefore, we have
\begin{equation} \label{eq:PivotsDist}
    \begin{split}
        d(p_i(x),p_i(y))&\leq d(p_i(x),x)+d(x,y)+d(y,p_i(y))\\
        &\stackrel{(\ref{eq:RDistance3})}{\leq}d(x,a)+\frac{r_{i-1}}{h}\cdot d(x,y)+d(x,y)+d(y,a)+\frac{r_{i-1}}{h}\cdot d(x,y)\\
        &=2d(x,y)+\frac{2r_{i-1}}{h}d(x,y)=(2+\frac{2r_{i-1}}{h})d(x,y)~.
    \end{split}
\end{equation}

We conclude that in the second case we have
\begin{eqnarray*}
d^{(3)}_{G\cup\bar{H}^{1/2}}(x,y)&\leq&w(Q)
=d(x,p_i(x))+d(p_i(x),p_i(y))+d(p_i(y),y)\\
&\leq&d(x,p_i(x))+(2+\frac{2r_{i-1}}{h})d(x,y)+d(p_i(y),y)\\
&\stackrel{(\ref{eq:RDistance3})}{\leq}&d(x,a)+d(a,y)+\frac{2r_{i-1}}{h}\cdot d(x,y)+(2+\frac{2r_{i-1}}{h})d(x,y)\\
&=&(3+\frac{4r_{i-1}}{h})d(x,y)~.
\end{eqnarray*}

Hence, by the analysis of the first two cases, we conclude that we can choose 
\[\alpha_i\geq\max\{\alpha_{i-1},3+\frac{4r_{i-1}}{h}\},\;\;\;\;\;\beta_i\geq\max\{2h\cdot\beta_{i-1},3\}~,\]
and the invariant will hold.

In the third case, $p_i(x),p_i(y)$ are not connected by an edge of $\bar{H}^{1/2}$. By the definition of half-bunches, it means that 
\begin{equation} \label{eq:RDistance4}
    d(p_i(x),p_{i+1}(p_i(x)))\leq2d(p_i(x),p_i(y))
    \stackrel{(\ref{eq:PivotsDist})}{\leq}(4+\frac{4r_{i-1}}{h})d(x,y)~,
\end{equation}
and similarly, $d(p_i(y),p_{i+1}(p_i(y)))\leq(4+\frac{4r_{i-1}}{h})d(x,y)$.

We conclude that in the third case,
\begin{eqnarray*}
d(x,p_{i+1}(x))
&\leq&d(x,p_{i+1}(p_i(x)))\\
&\leq&d(x,p_i(x))+d(p_i(x),p_{i+1}(p_i(x)))\\
&\stackrel{(\ref{eq:RDistance3}),(\ref{eq:RDistance4})}{\leq}&d(x,y)+\frac{r_{i-1}}{h}d(x,y)+(4+\frac{4r_{i-1}}{h})d(x,y)\\
&\leq&(5+\frac{5r_{i-1}}{h})d(x,y)~.
\end{eqnarray*}
Note that in the third case, there exists a child $[a,b]$ of $[x,y]$ labeled as \textit{Bad}, and thus, inequality (\ref{eq:RDistance3}) is applicable.

Hence, the invariant is satisfied if we choose $\{\alpha_i,\beta_i,r_i\}$ such that
\begin{equation} \label{eq:InvariantCondition}
    \alpha_i\geq\max\{\alpha_{i-1},3+\frac{4r_{i-1}}{h}\},\;\;\;\beta_i\geq\max\{2h\cdot\beta_{i-1},3\},\;\;\;
    r_i\geq5+\frac{5r_{i-1}}{h}~,
\end{equation}
and also $\alpha_0,\beta_0\geq1,r_0\geq2$.

Consider the choices $\alpha_i=3+\frac{20}{h-5},r_i=\frac{5h}{h-5}=5+\frac{25}{h-5}$ (for every $i$) and $\beta_i=(2h)^i$. Note that we have $\alpha_0>3>1,\beta_0=1$ and $r_0>5>2$ (recall that $h>5$). In addition, for all $i>0$,
\[5+\frac{5r_{i-1}}{h}=5+\frac{5\cdot\frac{5h}{h-5}}{h}=5+\frac{25}{h-5}=r_i~,\]
\begin{eqnarray*}
\max\{\alpha_{i-1},3+\frac{4r_{i-1}}{h}\}
=\max\{3+\frac{20}{h-5},3+\frac{20}{h-5}\}
=3+\frac{20}{h-5}=\alpha_i~,
\end{eqnarray*}
\[\max\{2h\cdot\beta_{i-1},3\}=\max\{(2h)^i,3\}=(2h)^i=\beta_i~.\]

This proves that Condition (\ref{eq:InvariantCondition}) is satisfied, and therefore the invariant is satisfied. Note that the root $[u,v]$ of the partition tree is in level $l-1$. As $p_l(u),p_l(v)$ are undefined, the condition for $[u,v]$ being \textit{Bad} never holds. Thus, $[u,v]$ is labeled \textit{Good}. Hence, we have
\[d^{(\beta_{l-1})}_{G\cup\bar{H}^{1/2}}(u,v)\leq\alpha_{l-1}\cdot d(u,v)~.\]

Choosing $h=5+\left\lceil\frac{20}{\epsilon}\right\rceil$, we get a stretch of 
\[\alpha_{l-1}=3+\frac{20}{h-5}\leq3+\epsilon\] 
and hopbound of $\beta_{l-1}=(2h)^{l-1}\leq(12+\frac{40}{\epsilon})^{l-1}$.

\end{proof}

\section{Appendix B} \label{sec:AppendixB}

In this appendix we prove Lemmas \ref{lemma:BranchingEvents} and \ref{lemma:BasicSizes} from Section \ref{sec:HierarchyOfSets}. These lemmas are then used to analyse the size and the support size of the hopset $\bar{H}^{1/2}$. Closely related lemmas and arguments can be found in \cite{CE05,Pet09,EP15,EN16b}.

\begin{lemma*}[Lemma \ref{lemma:BranchingEvents}]
For every $i=0,1,...,l-1$,
\[|Branch(H^{1/2}_i)|\leq4\sum_{u\in A_i}|B_i(u)|^3~.\]
\end{lemma*}

\begin{proof}
Let $(a,b,x)\in Branch(H^{1/2}_i)$ be a branching event, and let $a=(u,v),b=(y,z)$ such that $v\in B^{1/2}_i(u)$ and $z\in B^{1/2}_i(y)$. 

If $d(u,v)\leq d(y,z)$, we get
\begin{eqnarray*}
&&d(y,u)\leq d(y,x)+d(x,u)\leq d(y,z)+d(u,v)\\
&\leq&2d(y,z)<2\cdot\frac{1}{2}d(y,p_{i+1}(y))=d(y,p_{i+1}(y))~.
\end{eqnarray*}

See Figure \ref{fig:FewBranchingEvents} for illustration. Therefore $u\in B_i(y)$. Similarly, $v\in B_i(y)$ and also, of course, $z\in B_i^{1/2}(y)\subseteq B_i(y)$. If we instead have $d(y,z)\leq d(u,v)$, then symmetrically we get $y,z,v\in B_i(u)$.

We proved that $Branch(H^{1/2}_i)$ is a subset of the union
\begin{eqnarray*}
&&\bigg\{\big((u,v),(y,z),x\big)\in H^{1/2}_i\times H^{1/2}_i\times V\;\big|\;
v,y,z\in B_i(u)\text{ and }x\in P_{u,v}\cap P_{y,z}\bigg\}\\
&\bigcup&\bigg\{\big((u,v),(y,z),x\big)\in H^{1/2}_i\times H^{1/2}_i\times V\;\big|\;
u,v,z\in B_i(y)\text{ and }x\in P_{u,v}\cap P_{y,z}\bigg\}~.
\end{eqnarray*}

\begin{figure}[!ht]
\centerline{\includegraphics[width=6cm, height=4cm]{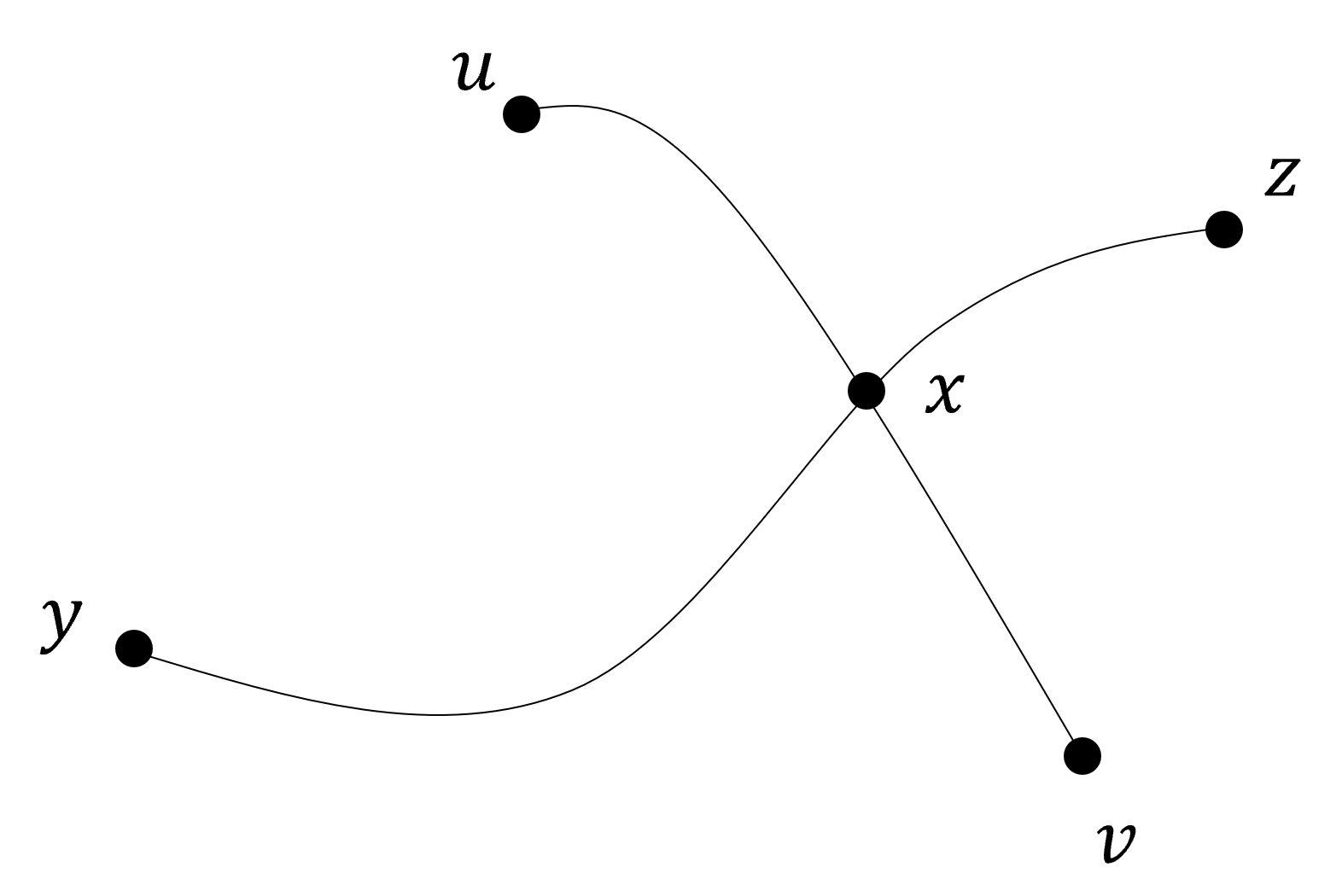}}
\caption{A branching event. When we assume that $d(u,v)\leq d(y,z)$, it follows that the weight of the sub-paths from each of $y,z,u,v$ to $x$ is at most $d(y,z)\leq\frac{1}{2}d(y,p_{i+1}(y))$.}
\label{fig:FewBranchingEvents}
\end{figure}

The size of the first set in the union above can be bounded by
\begin{eqnarray*}
&&\sum_{u\in A_i}\sum_{v,y,z\in B_i(u)}|\{x\in V\;|\;((u,v),(y,z),x)\in Branch(H^{1/2}_i))\}|\\
&\leq&\sum_{u\in A_i}\sum_{v,y,z\in B_i(u)}2=2\sum_{u\in A_i}|B_i(u)|^3~.
\end{eqnarray*}

Here the first inequality follows from the fact that the two shortest paths $P_{u,v},P_{y,z}$, that were chosen in a consistent manner, can have at most two branching events between them (see Lemma 7.5 in \cite{CE05}). We also used the fact that since $(u,v),(y,z)\in H^{1/2}_i$, then $u$ (as well as $v,y,z$) must be in $A_i$. 

Clearly, we can get the same bound for the size of the second set. Hence we conclude that
\[|Branch(H^{1/2}_i)|\leq4\sum_{u\in A_i}|B_i(u)|^3~.\]

\end{proof}

Next, we prove Lemma \ref{lemma:BasicSizes}.

\begin{lemma*}[Lemma \ref{lemma:BasicSizes}]
The following inequalities hold:
\begin{enumerate}
    \item For every $i<l-1$, $\mathbb{E}[|\bar{H}_i|]\leq\frac{n}{q_i}\prod_{j=0}^{i-1}q_j$.\newline
    \item $\mathbb{E}[|\bar{H}_{l-1}|]\leq2\cdot\max\{1,(n\prod_{j=0}^{l-2}q_j)^2\}$.\newline
    \item For every $i<l-1$, $\mathbb{E}[|Branch(H^{1/2}_i)|]\leq24\cdot\frac{n}{q^3_i}\prod_{j=0}^{i-1}q_j$.\newline
    \item $\mathbb{E}[|Branch(H^{1/2}_{l-1})|]\leq60\cdot\max\{1,(n\prod_{j=0}^{l-2}q_j)^4\}$.
\end{enumerate}
\end{lemma*}

\begin{proof}

For proving inequalities (1) and (3), fix $i<l-1$. We define a random variable $X(v)$ for every $v\in V$, that equals $|\bar{B}_i(v)|$ if $v\in A_i$, and otherwise equals $0$. By the definition of $\bar{H}_i$, we can see that 
\[|\bar{H}_i|\leq\sum_{v\in A_i}|\bar{B}_i(v)|=\sum_{v\in V}X(v)~.\]

Now fix some $v\in V$ and a subset $S\subseteq V$ such that $v\in S$, and suppose that $A_i=S$. We order the vertices of $A_i$ by an increasing order of their distances from $v$. Then $B_i(v)$ is a subset of the vertices that come before the first $A_{i+1}$-vertex in that list. Every vertex in this list is contained in $A_{i+1}$ with probability $q_i$. Therefore $|\bar{B}_i(v)|\leq|B_i(v)|+1$ is bounded by a geometric random variable with probability $q_i$, when it is conditioned on the event that $A_i=S$. By the law of total expectation,
\begin{eqnarray*}
\mathbb{E}[X(v)\;|\;v\in A_i]
&=&\sum_{\stackrel{S\subseteq V}{v\in S}}\mathbb{E}[X(v)\;|\;v\in A_i\text{ and }A_i=S]\Pr[A_i=S\;|\;v\in A_i]\\
&=&\sum_{\stackrel{S\subseteq V}{v\in S}}\mathbb{E}[|\bar{B}_i(v)|\;|\;A_i=S]\cdot\Pr[A_i=S\;|\;v\in A_i]\\
&\leq&\sum_{\stackrel{S\subseteq V}{v\in S}}\frac{1}{q_i}\cdot\Pr[A_i=S\;|\;v\in A_i]=\frac{1}{q_i}~.
\end{eqnarray*}

Using the law of total expectation again, we get
\begin{eqnarray*}
\mathbb{E}[X(v)]=\mathbb{E}[X(v)\;|\;v\in A_i]\cdot\Pr[v\in A_i]+0\cdot\Pr[v\notin A_i]
\leq\frac{1}{q_i}\cdot\Pr[v\in A_i]~.
\end{eqnarray*}
By the definition of $A_i$, the probability of $v$ to be in $A_i$ is the product $\prod_{j=0}^{i-1}q_j$. We finally get that
\[\mathbb{E}[|\bar{H}_i|]\leq\sum_{v\in V}\mathbb{E}[X(v)]\leq\sum_{v\in V}\frac{1}{q_i}\cdot\prod_{j=0}^{i-1}q_j=\frac{n}{q_i}\cdot\prod_{j=0}^{i-1}q_j~.\]
This proves inequality (1). For inequality (3), we use the fact that for a geometric random variable $X$ with probability $p$,
\[\mathbb{E}[X^3]=\frac{p^2-6p+6}{p^3}\leq\frac{6}{p^3}~.\]
Similarly to the computation above, we get
\begin{eqnarray*}
\mathbb{E}[X(v)^3\;|\;v\in A_i]
&=&\sum_{\stackrel{S\subseteq V}{v\in S}}\mathbb{E}[X(v)^3\;|\;v\in A_i\text{ and }A_i=S]\Pr[A_i=S\;|\;v\in A_i]\\
&=&\sum_{\stackrel{S\subseteq V}{v\in S}}\mathbb{E}[|\bar{B}_i(v)|^3\;|\;A_i=S]\cdot\Pr[A_i=S\;|\;v\in A_i]\\
&\leq&\sum_{\stackrel{S\subseteq V}{v\in S}}\frac{6}{q_i^3}\cdot\Pr[A_i=S\;|\;v\in A_i]=\frac{6}{q_i^3}~.
\end{eqnarray*}
Now, by Lemma \ref{lemma:BranchingEvents},
\begin{eqnarray*}
\mathbb{E}[|Branch(H^{1/2}_i)|]&\leq&4\mathbb{E}[\sum_{v\in A_i}|B_i(v)|^3]
=4\mathbb{E}[\sum_{v\in V}X(v)^3]\\
&=&4\sum_{v\in V}\mathbb{E}[X(v)^3]
=4\sum_{v\in V}\mathbb{E}[X(v)^3\;|\;v\in A_i]\Pr[v\in A_i]\\
&\leq&4\sum_{v\in V}\frac{6}{q_i^3}\cdot\Pr[v\in A_i]
=24\cdot\frac{n}{q_i^3}\cdot\prod_{j=0}^{i-1}q_j~,
\end{eqnarray*}
which proves inequality (3).

For proving inequalities (2) and (4), note that since $A_l=\emptyset$, then $\bar{B}_{l-1}(v)=A_{l-1}$, for any $v\in A_{l-1}$. Also, notice that $|A_{l-1}|$ is a binomial random variable with parameter $n$ and probability $q=\prod_{j=0}^{l-2}q_j$. For a binomial random variable $X\sim Bin(n,q)$, 
\[\mathbb{E}[X^2]=nq(1-q)+(nq)^2\leq nq+(nq)^2\leq2\max\{1,(nq)^2\}~,\]
since either $nq\geq1$, and then $nq+(nq)^2\leq(nq)^2+(nq)^2=2(nq)^2$, or $nq<1$, and then $nq+(nq)^2<1+1=2$.

By definition, $|\bar{H}_{l-1}|\leq\sum_{v\in A_{l-1}}|A_{l-1}|=|A_{l-1}|^2$, and therefore,
\[\mathbb{E}[|\bar{H}_{l-1}|]\leq\mathbb{E}[|A_{l-1}|^2]\leq2\max\{1,(n\prod_{j=0}^{l-2}q_j)^2\}~.\]
This proves inequality (2).

Another property of a variable $X\sim Bin(n,q)$ is the following.
\begin{eqnarray*}
\mathbb{E}[X^4]&=&nq+14\binom{n}{2}q^2+36\binom{n}{3}q^3+24\binom{n}{4}q^4\\
&\leq&nq+7(nq)^2+6(nq)^3+(nq)^4
\leq15\max\{1,(nq)^4\}~,
\end{eqnarray*}
where the reasoning for the last step is similar to the above. Then, by Lemma \ref{lemma:BranchingEvents},
\begin{eqnarray*}
|Branch(H_{l-1})|\leq4\sum_{v\in A_{l-1}}|B_{l-1}(v)|^3
=4\sum_{v\in A_{l-1}}|A_{l-1}|^3=4|A_{l-1}|^4~.
\end{eqnarray*}
Finally,
\begin{eqnarray*}
\mathbb{E}[|Branch(H_{l-1})|]\leq4\mathbb{E}[|A_{l-1}|^4]
\leq4\cdot15\max\{1,(nq)^4\}
=60\max\{1,(n\prod_{j=0}^{l-2}q_j)^4\}~.
\end{eqnarray*}

This proves inequality (4).

\end{proof}

\section{Appendix C} \label{sec:AppendixC}

Here we provide the proof of Theorem \ref{thm:DistancePreserver2}.

\begin{theorem*}[Theorem \ref{thm:DistancePreserver2}]
Given an $n$-vertex undirected weighted graph $G=(V,E)$, an integer $k\in[3,\log_2n]$, a positive parameter $\epsilon\leq1$, and a set of pairs $\mathcal{P}\subseteq V^2$, $G$ has an interactive $(1+\epsilon)$-distance preserver with query time $O(1)$ and size 
\[O\left(|\mathcal{P}|\cdot\Tilde{\gamma}_2(\epsilon,k,n)+n\cdot\left(\log k+\log\left(1+\frac{\log\frac{1}{\epsilon}}{\log^{(3)}n}\right)\right)+n^{1+\frac{1}{k}}\right)~,\]
where $\Tilde{\gamma}_2(\epsilon,k,n)=O\left(\frac{\log k+\rho}{\epsilon}\right)^{\left\lceil\log_2k\right\rceil+\rho}$, and
\begin{equation} \label{eq:RhoDef}
    \rho=\rho(\epsilon,k,n)=\log_{4/3}\left(\frac{k\log\log n\cdot\log^{(3)}n}{\log n}\right)+\log_{4/3}\left(1+\frac{\log\frac{1}{\epsilon}}{\log^{(3)}n}\right)+O(1)=O(\log\log k)+\log_{4/3}\log\frac{1}{\epsilon}~.
\end{equation}

In addition, there exists a $(1+\epsilon,\tilde{\gamma}_2(\epsilon,k,n))$-hopset with size $O(n^{1+\frac{1}{k}})$ and $(1+\epsilon)$-approximate support size 
\[O\left(n^{1+\frac{1}{k}}+n\left(\log k+\log\left(1+\frac{\log\frac{1}{\epsilon}}{\log^{(3)}n}\right)\right)\right)=O(n^{1+\frac{1}{k}}+n(\log k+\log\log\frac{1}{\epsilon}))~.\]
\end{theorem*}

\begin{proof}

For the sake of brevity, we denote $\gamma_{4/3}=\gamma_{4/3}(\epsilon,k)$. This time, we choose a different hierarchy of sets than in the proof of Theorem \ref{thm:DistancePreserver1}. Specifically, let
\begin{equation} \label{eq:AlternateProb}
q_i=\begin{cases}
\frac{1}{2}n^{-\frac{(4/3)^i}{3k}} & i<h\\
\frac{1}{2}(\frac{n}{\gamma_{4/3}})^{-\frac{2^{i-h}}{k}} & i\geq h\\
\end{cases}~,
\end{equation}
where $l,h$ will be determined later.

We use Claim \ref{claim:ENHopset}, but with an approximation parameter $\epsilon'=\frac{\epsilon}{3}$ instead of $\epsilon$, to conclude that $\bar{H}^{1/2}$ is a $(1+\frac{\epsilon}{3},\beta_l)$-hopset ($\beta_l$ is still $O(\frac{3l}{\epsilon})^{l-1}=O(\frac{l}{\epsilon})^{l-1}$). We divide this hopset into three disjoint sets.
\begin{equation} \label{eq:HopsetPartition}
\mathcal{Q}_1=\bigcup_{i=0}^{h-1}\{(v,p_i(v))\;|\;v\in V\}~,\;\;\;
\mathcal{Q}_2=\bigcup_{i=0}^{h-1}H^{1/2}_i~,\;\;\;
\mathcal{Q}_3=\bigcup_{i=h}^{l-1}\bar{H}^{1/2}_i~.
\end{equation}

By the size analysis below, we note that using full bunches instead of half-bunches in the definition of $\mathcal{Q}_3$ also achieves the same results, even with slightly better constant factors hidden by the $O$-notations. However, for convenience, and since asymptotically the results are similar, we stick with half-bunches.

For every $i=0,1,...,h-1$, let $(S^1_i,D^1_i)$ be the interactive distance preserver from Lemma \ref{lemma:PathsToPivots} on the set $\{(v,p_i(v))\;|\;v\in V\}$. Let $(S^2_i,D^2_i)$ be the interactive distance preserver from Theorem \ref{thm:DPPRO} on the set $H^{1/2}_i$. Lastly, let $(S^3,D^3)$ be the interactive distance preserver from Theorem \ref{thm:DistancePreserver1} on the set $\mathcal{Q}_3$, but with an approximation parameter of $\epsilon'=\frac{\epsilon}{3}$ instead of $\epsilon$.

We now construct an interactive distance preserver $(S,D)$, where the oracle $D$ stores the following information.
\begin{enumerate}
    \item For every $(x,y)\in\mathcal{P}$, the oracle $D$ stores a path $P_{x,y}\subseteq G\cup\bar{H}^{1/2}$ between $x,y$ with length at most $\beta_l$ and weight at most $(1+\frac{\epsilon}{3})d(x,y)$.
    \item For every $i=0,1,...,h-1$, the oracle $D$ stores the oracles $D^1_i$.
    \item For every $i=0,1,...,h-1$, the oracle $D$ stores the oracles $D^2_i$.
    \item The oracle $D$ stores the oracle $D^3$.
    \item For every $e\in\bar{H}^{1/2}$, the oracle $D$ stores a \textit{flag} variable $f_e$. If $e\in\{(v,p_i(v))\;|\;v\in V\}$ for some $i<h$, then $f_e=1$. If $e\in H^{1/2}_i$ for $i<h$, then $f_e=2$. If $e\in\bar{H}^{1/2}_i$ for $i\geq h$, $f_e=3$. In all cases, $D$ also stores the relevant index $i$.
\end{enumerate}

Next, we describe the algorithm for answering queries.

Given a query $(x,y)\in\mathcal{P}$, find the stored path $P_{x,y}$. Now replace each hop-edge $e=(a,b)\in\bar{H}^{1/2}$ that appears in $P_{x,y}$; if $f_e=1$, replace $e$ with the path returned from $D^1_i(a,b)$; if $f_e=2$, replace $e$ with the path returned from $D^2_i(a,b)$; if $f_e=3$, use $D^3(a,b)$. Return the resulting path as an output.

Since the query time of $D^1_i,D^2_i,D^3$ is linear in the number of edges of the returned path, for every $i$, we conclude that the total query time of $D$ is also linear in the size of the output. Thus, the query time of this oracle is $O(1)$.

Also, note that the paths that are returned from the oracles $D^1_i,D^2_i,D^3$ have a weight of at most $(1+\frac{\epsilon}{3})$ times the hop-edges they replace (By Lemma \ref{lemma:PathsToPivots} and Theorems \ref{thm:DPPRO} and \ref{thm:DistancePreserver1}). Therefore the resulting path has a weight of at most
\[(1+\frac{\epsilon}{3})^2d(x,y)\leq(1+\epsilon)d(x,y)~,\]
(since $\epsilon\leq1$). Hence, the stretch is indeed $1+\epsilon$.

We now analyse the size of $D$. The paths $P_{x,y}$ from item $1$ in the description of $D$ have a total size of $O(|\mathcal{P}|\beta_l)$. For item $2$, the size of the oracle $D^1_i$, for every $i<h$, is $O(n)$ (by Lemma \ref{lemma:PathsToPivots}), and therefore the total size of these oracles is $O(n\cdot h)$. The total size of the oracles $D^2_i$ from item $3$ can be bounded by $O(n\cdot h+n^{1+\frac{1}{k}})$, as in the proof of Theorem \ref{thm:DistancePreserver1}. The flags and the indices from item $5$ require $O(|\bar{H}^{1/2}|)$ space. To bound the size of the hopset $\bar{H}^{1/2}=\bigcup_{i=0}^{l-1}\bar{H}_i^{1/2}$, we bound $|\bar{H}_i|$ for every $i\in[0,l-1]$ (recall that $\bar{H}_i^{1/2}\subseteq\bar{H}_i$, and therefore the sum of these bounds is an upper bound for $|\bar{H}^{1/2}|$). By inequality (1) in Lemma \ref{lemma:BasicSizes}, for $i\in[0,h-1]$, in expectation,
\begin{equation} \label{eq:LowLevelSize}
|\bar{H}_i|\leq\frac{n}{q_i}\prod_{j=0}^{i-1}q_j
=n\cdot2n^{\frac{(4/3)^i}{3k}}\cdot\frac{1}{2^i}n^{-\frac{(4/3)^i-1}{k}}=\frac{1}{2^{i-1}}n^{1-\frac{2(4/3)^i}{3k}+\frac{1}{k}}\leq\frac{1}{2^{i-1}}n^{1+\frac{1}{3k}}~.
\end{equation}

By the same inequality, for $i\in[h,l-1)$, in expectation,
\begin{eqnarray*}
|\bar{H}_i|&\leq&\frac{n}{q_i}\prod_{j=0}^{i-1}q_j
=n\cdot2(\frac{n}{\gamma_{4/3}})^{\frac{2^{i-h}}{k}}\cdot\frac{1}{2^h}n^{-\frac{(4/3)^h-1}{k}}\cdot\frac{1}{2^{i-h}}(\frac{n}{\gamma_{4/3}})^{-\frac{2^{i-h}-1}{k}}\\
&=&\frac{1}{2^{i-1}}n\cdot n^{-\frac{(4/3)^h-1}{k}}\cdot(\frac{n}{\gamma_{4/3}})^{\frac{1}{k}}
\stackrel{*}{\leq}\frac{1}{2^{i-1}}n\cdot\frac{1}{\gamma_{4/3}}\cdot(\frac{n}{\gamma_{4/3}})^{\frac{1}{k}}
=\frac{1}{2^{i-1}}(\frac{n}{\gamma_{4/3}})^{1+\frac{1}{k}}~.
\end{eqnarray*}
For inequality (*) to hold, we need to choose $h$ so that 
\begin{equation} \label{eq:HReq}
    n^{\frac{(4/3)^h-1}{k}}\geq\gamma_{4/3}~,
\end{equation}
i.e., $h\geq\log_{4/3}(\frac{k\log\gamma_{4/3}}{\log n}+1)$. Indeed, we fix the value of $h$ to be $h=\left\lceil\log_{4/3}(\frac{k\log\gamma_{4/3}}{\log n}+1)\right\rceil$. Note that $h$ can be bounded as follows. 
\begin{eqnarray*}
    h&=&\left\lceil\log_{4/3}\left(\frac{k\log\gamma_{4/3}}{\log n}+1\right)\right\rceil\\
    &\stackrel{(\ref{eq:GammaDef})}{=}&\left\lceil\log_{4/3}\left(\frac{k\lceil\log_{4/3}k\rceil(\log\log k+\log\frac{1}{\epsilon})+O(1)}{\log n}\right)\right\rceil\\
    &=&\log_{4/3}\left(\frac{k\log k(\log\log k+\log\frac{1}{\epsilon})}{\log n}\right)+O(1)\\
    &\leq&\log_{4/3}\left(\frac{k\log\log n(\log^{(3)}n+\log\frac{1}{\epsilon})}{\log n}\right)+O(1)\\
    &=&\log_{4/3}\left(\frac{k\log\log n\cdot\log^{(3)}n}{\log n}\right)+\log_{4/3}\left(1+\frac{\log\frac{1}{\epsilon}}{\log^{(3)}n}\right)+O(1)\\
    &\stackrel{(\ref{eq:RhoDef})}{=}&\rho(\epsilon,k,n)~.
\end{eqnarray*}
In addition, since $\frac{x}{\log x\cdot\log\log x}$ is a non-decreasing function for every $x\geq8$, we have 
\[\log_{4/3}\left(\frac{k\log\log n\cdot\log^{(3)}n}{\log n}\right)=O(\log_{4/3}(\log k\cdot\log\log k))=O(\log\log k)~.\]
Hence, 
\begin{equation} \label{eq:RhoRoughBound}
    h\leq\rho(\epsilon,k,n)=\log_{4/3}\left(1+\frac{\log\frac{1}{\epsilon}}{\log^{(3)}n}\right)+O(\log\log k)
\end{equation}
We conclude that setting $h$ to this value yields that for every $i\in[h,l-1)$,
\begin{equation} \label{eq:HighLevelSize}
    |\bar{H}_i|\leq\frac{1}{2^{i-1}}(\frac{n}{\gamma_{4/3}})^{1+\frac{1}{k}}~.
\end{equation}

For $i=l-1$, denote $q=\prod_{j=0}^{l-2}q_j$. By inequality (2) in Lemma \ref{lemma:BasicSizes}, if $nq<1$, then the expected value of $|\bar{H}_{l-1}|$ is bounded by a constant. Otherwise,
\begin{eqnarray*}
|\bar{H}_{l-1}|&\leq&2(n\prod_{j=0}^{l-2}q_j)^2
\stackrel{(\ref{eq:AlternateProb})}{=}2n^2\cdot\left(\frac{1}{2^h}n^{-\frac{(4/3)^h-1}{k}}\cdot\frac{1}{2^{l-1-h}}(\frac{n}{\gamma_{4/3}})^{-\frac{2^{l-1-h}-1}{k}}\right)^2\\
&\stackrel{(\ref{eq:HReq})}{\leq}&n^2\cdot\frac{1}{2^{2l-3}}\left(\frac{1}{\gamma_{4/3}}\cdot(\frac{n}{\gamma_{4/3}})^{-\frac{2^{l-1-h}-1}{k}}\right)^2\\
&\stackrel{**}{\leq}&n^2\cdot\frac{1}{2^{2l-3}}\left(\frac{1}{\gamma_{4/3}}\cdot(\frac{n}{\gamma_{4/3}})^{-\frac{k-1}{k}}\right)^2\\
&=&\frac{1}{2^{2l-3}}(\frac{n}{\gamma_{4/3}})^{2-2\cdot\frac{k-1}{k}}\\
&=&\frac{1}{2^{2l-3}}(\frac{n}{\gamma_{4/3}})^{\frac{2}{k}}\leq\frac{1}{2^{2l-3}}(\frac{n}{\gamma_{4/3}})^{1+\frac{1}{k}}~.
\end{eqnarray*}
For inequality (**) to hold, we need to choose $l$ so that $2^{l-1-h}\geq k$. We define $l$ as the smallest value that satisfies this constraint, i.e., 
\begin{equation} \label{eq:LDef}
    l=\lceil\log_2k\rceil+1+h=\lceil\log_2k\rceil+1+\rho=\lceil\log_2k\rceil+\log_{4/3}\left(1+\frac{\log\frac{1}{\epsilon}}{\log^{(3)}n}\right)+O(\log\log k)
\end{equation}
(for the last bound, see (\ref{eq:RhoRoughBound})). Using the notation of $\Tilde{\gamma}_2(\epsilon,k,n)$ (see Theorem \ref{thm:DistancePreserver2} for its definition), we get
\[\beta_l=O\left(\frac{l}{\epsilon}\right)^{l-1}=O\left(\frac{\log k+\rho}{\epsilon}\right)^{\lceil\log_2k\rceil+\rho}=\Tilde{\gamma}_2(\epsilon,k,n)~.\]
We conclude that for the chosen value of $l$, we have
\begin{equation} \label{eq:LastLevelSize}
    |\bar{H}_{l-1}|\leq\frac{1}{2^{2l-3}}(\frac{n}{\gamma_{4/3}})^{1+\frac{1}{k}}~.
\end{equation}

Summing up, by inequalities (\ref{eq:LowLevelSize}), (\ref{eq:HighLevelSize}), and (\ref{eq:LastLevelSize}), the total size of the hopset $\bar{H}^{1/2}$ is bounded by
\begin{equation} \label{eq:HopsetSize}
    \sum_{i=0}^{l-1}|\bar{H}_i|\leq\sum_{i=0}^{h-1}\frac{1}{2^{i-1}}n^{1+\frac{1}{3k}}+\sum_{i=h}^{l-2}\frac{1}{2^{i-1}}(\frac{n}{\gamma_{4/3}})^{1+\frac{1}{k}}+\frac{1}{2^{2l-3}}(\frac{n}{\gamma_{4/3}})^{1+\frac{1}{k}}=O(n^{1+\frac{1}{k}})~.
\end{equation}
This also bounds the total size of the flags and the indices from item $5$ in the description of the oracle $D$.

Next, we estimate the expected size of the oracle $D^3$ (item $4$ in description of the oracle $D$). Note that the set of pairs $\mathcal{Q}_3$ (defined by (\ref{eq:HopsetPartition})) is partial to the union $\bigcup_{i=h}^{l-1}\bar{H}_i^{1/2}$. By inequalities (\ref{eq:HighLevelSize}) and (\ref{eq:LastLevelSize}), we have
\[|\mathcal{Q}_3|\leq\sum_{i=h}^{l-1}|H^{1/2}_i|\leq\sum_{i=h}^{l-1}\frac{1}{2^{i-1}}(\frac{n}{\gamma_{4/3}})^{1+\frac{1}{k}}=O((\frac{n}{\gamma_{4/3}})^{1+\frac{1}{k}})~.\]
Thus, the total size of $D^3$, which is the interactive distance preserver from Theorem \ref{thm:DistancePreserver1} for the set $\mathcal{Q}_3$, is
\[O(|\mathcal{Q}_3|\cdot\gamma_{4/3}+n\cdot\log k+n^{1+\frac{1}{k}})=O((\frac{n}{\gamma_{4/3}})^{1+\frac{1}{k}}\cdot\gamma_{4/3}+n\cdot\log k+n^{1+\frac{1}{k}})=O(n\cdot\log k+n^{1+\frac{1}{k}})~.\]

We finally conclude that the total size of the oracle $D$ is
\begin{eqnarray*}
O(|\mathcal{P}|\beta_l+n\cdot h+n\cdot\log k+n^{1+\frac{1}{k}})\stackrel{(\ref{eq:LDef})}{=}O\left(|\mathcal{P}|\cdot\Tilde{\gamma}_2(\epsilon,k,n)+n\cdot\left(\log k+\log\left(1+\frac{\log\frac{1}{\epsilon}}{\log^{(3)}n}\right)\right)+n^{1+\frac{1}{k}}\right)~.
\end{eqnarray*}

We also define the distance preserver $S$ to be the union of all the edges $P_{x,y}\cap G$, where $(x,y)\in\mathcal{P}$, together with $\bigcup_{i=0}^{l-1}S^1_i\cup\bigcup_{i=0}^{h-1}S^2_i\cup S^3$. Then, the output paths of $D$ are always contained in $S$, and $|S|=O(|D|)$. Hence, $(S,D)$ is an interactive $(1+\epsilon)$-distance preserver with query time $O(1)$ and size
\[O\left(|\mathcal{P}|\cdot\Tilde{\gamma}_2(\epsilon,k,n)+n\cdot\left(\log k+\log\left(1+\frac{\log\frac{1}{\epsilon}}{\log^{(3)}n}\right)\right)+n^{1+\frac{1}{k}}\right)~.\]

In addition, note that the set $S'=\bigcup_{i=0}^{l-1}S^1_i\cup\bigcup_{i=0}^{h-1}S^2_i\cup S^3$ is a $(1+\frac{\epsilon}{3})$-approximate supporting edge-set of the hopset $\bar{H}^{1/2}$. Thus, the set $\bar{H}^{1/2}$ is a $(1+\epsilon,\beta_l=\Tilde{\gamma}_2(\epsilon,k,n))$-hopset with size $O(n^{1+\frac{1}{k}})$ (see (\ref{eq:HopsetSize})) and $(1+\frac{\epsilon}{3})$-approximate support size \[O\left(n\cdot\left(\log k+\log_{4/3}\left(1+\frac{\log\frac{1}{\epsilon}}{\log^{(3)}n}\right)\right)+n^{1+\frac{1}{k}}\right)~.\]

\end{proof}

\section{Appendix D} \label{sec:AppendixD}

Lemma \ref{lemma:PartialPRDO} can be derived by a similar analysis as \cite{TZ01,WN13}. However, since we apply their construction only for $h<k$ of the levels, we bring here the full proof, for completeness.

\begin{lemma*}[Lemma \ref{lemma:PartialPRDO}]
Let $G=(V,E)$ be an undirected weighted graph with $n$ vertices, and let $1\leq h<k\leq\log_2n$ be two integer parameters. There is a set $S\subseteq V$ of size $O(n^{1-\frac{h}{k}})$ and an oracle $D$ with size $O(h\cdot n^{1+\frac{1}{k}})$, that acts as follows. Given a query $(u,v)\in V^2$, the oracle $D$ either returns a path between $u,v$ with weight at most $(2h+1)d_G(u,v)$, or returns two paths $P_{u,u'},P_{v,v'}$, from $u$ to some $u'\in S$ and from $v$ to some $v'\in S$, such that
\[w(P_{u,u'}),w(P_{v,v'})\leq h\cdot d(u,v)~.\]

The query time of the oracle $D$ is $O(\log h)$. In addition, there is a set $E'\subseteq E$ such that the output paths that $D$ returns are always contained in $E'$, and
\[|E'|=O(h\cdot n^{1+\frac{1}{k}})~.\]
\end{lemma*}

Before presenting the proof, we introduce some notions that will be useful in the sequel. For brevity, we will use the notation $d(x,y)$ instead of $d_G(x,y)$ for the distance between a pair of vertices $x,y$ in $G$. Given two vertices $u,v\in V$ and an index $i=0,1,...,h$, we define:
\begin{enumerate}
    \item $P(u,v,i)=$ a predicate that indicates whether $d(u,p_i(u))\leq i\cdot d(u,v)$. When considering the sequence $\sigma_u$ of distances $d(u,p_0(u)),d(u,p_1(u)),...,d(u,p_h(u))$, this predicate indicates which elements are considered \textit{small}.
    \item $Q(u,v,i)=$ a predicate that indicates whether 
    \[(p_i(u)\in B_i(v))\;\vee\;(p_{i+1}(v)\in B_{i+1}(u))\] 
    (this will be used only for $i<h$). Since our oracle will store the shortest paths between vertices and the members of their bunches, and the shortest paths between vertices and their pivots, $Q$ actually indicates whether we have enough information to output a path between $u,v$ that passes through either $p_i(u)$ or $p_{i+1}(v)$.
    \item $\Delta_u(i)=d(u,p_{i+2}(u))-d(u,p_i(u))$. This represents the difference between two adjacent elements of even indices in the distance sequence $\sigma_u$. This notation will be used only when $i+2\leq h$.
\end{enumerate}

The key observation regarding these notions is provided in the following lemma.
\begin{lemma} \label{lemma:FailedStep}
For $u,v\in V$ and $i\in[0,h-2]$, if $Q(u,v,i)=FALSE$, then $\Delta_u(i)\leq2d(u,v)$.

Therefore, if $P(u,v,i)=TRUE$ and also $Q(u,v,i)=FALSE$, then $P(u,v,i+2)=TRUE$.
\end{lemma}

\begin{proof}
By the definitions of bunches, assuming $Q(u,v,i)=FALSE$, we know that
\[d(v,p_{i+1}(v))\leq d(v,p_i(u))\leq d(v,u)+d(u,p_i(u))~,\]
and
\begin{eqnarray*}
d(u,p_{i+2}(u))\leq d(u,p_{i+1}(v))
\leq d(u,v)+d(v,p_{i+1}(v))~.
\end{eqnarray*}
Thus,
\begin{eqnarray*}
d(u,p_{i+2}(u))\leq d(u,v)+d(v,p_{i+1}(v))
\leq d(u,v)+d(v,u)+d(u,p_i(u))~,
\end{eqnarray*}
that is, $\Delta_u(i)=d(u,p_{i+2}(u))-d(u,p_i(u))\leq2d(u,v)$.

If we also have $P(u,v,i)=TRUE$, then we know that $d(u,p_i(u))\leq i\cdot d(u,v)$, and so we have
\begin{eqnarray*}
d(u,p_{i+2}(u))&\leq&2d(u,v)+d(u,p_i(u))
\leq2d(u,v)+i\cdot d(u,v)\\
&=&(i+2)\cdot d(u,v)~,
\end{eqnarray*}
i.e., $P(u,v,i+2)=TRUE$.
\end{proof}

In our oracle, we store the shortest paths from every $x\in V$ to every $y\in B_i(x)$, and from every $x\in V$ to each of its pivots $p_i(x)$, for $i\in[0,h]$. In the query algorithm, given a query $(u,v)\in V^2$, we would like to find an index $i\in[0,h)$ such that $Q(u,v,i)=P(u,v,i)=TRUE$. For such index $i$, we either return the shortest path from $u$ to $p_i(u)$, concatenated with the shortest path from $p_i(u)$ to $v$ (this is in the case where $p_i(u)\in B_i(v)$), or return the shortest path from $v$ to $p_{i+1}(v)$, concatenated with the shortest path from $p_{i+1}(v)$ to $u$. In the first case, the weight of the returned path is
\begin{equation} \label{eq:ConcatBound1}
\begin{split}
d(u,p_i(u))+d(p_i(u),v)&\leq2d(u,p_i(u))+d(u,v)
\leq(2i+1)d(u,v)\\
&\leq(2h+1)d(u,v)~.
\end{split}
\end{equation}
In the second case, since $p_i(u)\notin B_i(v)$, we have $d(v,p_{i+1}(v))\leq d(v,p_i(u))$, and thus, 
\begin{eqnarray*}
d(v,p_{i+1}(v))+d(p_{i+1}(v),u)
&\leq&2d(v,p_{i+1}(v))+d(u,v)\\
&\leq&2d(v,p_i(u))+d(u,v)\\
&\leq&2(d(v,u)+d(u,p_i(u)))+d(u,v)\\
&\leq&2d(u,p_i(u))+3d(u,v)\\
&\leq&(2i+3)d(u,v)\\
&\leq&(2(h-1)+3)d(u,v)=(2h+1)d(u,v)~.
\end{eqnarray*}
Hence, the weight of the returned path is
\begin{equation} \label{eq:ConcatBound2}
    d(v,p_{i+1}(v))+d(p_{i+1}(v),u)\leq(2h+1)d(u,v)~.
\end{equation}
In both cases we used the fact that $P(u,v,i)=TRUE$, and therefore $d(u,p_i(u))\leq i\cdot d(u,v)$.

For handling the case that there is no index $i$ such that $Q(u,v,i)=P(u,v,i)=TRUE$, we have the following lemma. From now on, we assume for simplicity that $h$ is even. The analysis of the case that $h$ is odd is similar.
\begin{lemma} \label{lemma:DesiredIndex}
Let $(u,v)\in V^2$, and let $[i_1,i_2]\subseteq[0,h-2]$ such that $i_1,i_2$ are even, and $P(u,v,i_1)=TRUE$. Then, either there is an even index $i\in[i_1,i_2]$ that satisfies $Q(u,v,i)=P(u,v,i)=TRUE$, or $P(u,v,i_2+2)=TRUE$.
\end{lemma}

\begin{proof}

First, if none of the even indices $i\in[i_1,i_2]$ satisfy $Q(u,v,i)=TRUE$, then by Lemma \ref{lemma:FailedStep} we know that $\Delta_u(i)\leq2d(u,v)$ for every even $i\in[i_1,i_2]$. Therefore, when summing $\Delta_u(i)$, for all even indices $i\in[i_1,i_2]$, we get
\begin{eqnarray*}
d(u,p_{i_2+2}(u))-d(u,p_{i_1}(u))
&=&\sum_{i}\Delta_u(i)
\leq\frac{i_2-i_1+2}{2}\cdot2d(u,v)\\
&=&(i_2-i_1+2)d(u,v)~.
\end{eqnarray*}
Since we know that $P(u,v,i_1)=TRUE$, we conclude that
\begin{eqnarray*}
d(u,p_{i_2+2}(u))&\leq&d(u,p_{i_1}(u))+(i_2-i_1+2)d(u,v)\\
&\leq&i_1\cdot d(u,v)+(i_2-i_1+2)d(u,v)\\
&=&(i_2+2)d(u,v)~.
\end{eqnarray*}
That is, $P(u,v,i_2+2)=TRUE$.

Suppose now that there is an even index $i\in[i_1,i_2]$ such that $Q(u,v,i)=TRUE$, and let $i$ be the smallest such index. We now know that in the interval $[i_1,i-2]$, we have $P(u,v,i_1)=TRUE$, and there are no indices $j\in[i_1,i-2]$ with $Q(u,v,j)=TRUE$. By the first case that we proved above, it follows that $P(u,v,i)=TRUE$. Recall that by the choice of $i$, we also know that $Q(u,v,i)=TRUE$. Therefore $i$ is the desired index.

\end{proof}

Suppose that there is no even index $i\in[0,h-2]$ such that $Q(u,v,i)=P(u,v,i)=TRUE$, and there is no even index $i\in[0,h-2]$ such that $Q(v,u,i)=P(v,u,i)=TRUE$ (with the order of $u,v$ switched). In that case, we can simply return the two shortest paths from $u$ to $p_h(u)$ and from $v$ to $p_h(v)$, because they both have a weight of at most $h\cdot d(u,v)$, as desired. This is true by Lemma \ref{lemma:DesiredIndex}, using the the fact that for the index $0$, $P(u,v,0)=P(v,u,0)=TRUE$ always holds.

To find an even index $i\in[0,h-2]$ such that $Q(u,v,i)=P(u,v,i)=TRUE$, or determine that there is no such index, we perform a sort of \textit{binary search}. The details of this binary search are described in the proof of the main lemma below. For the sake of this search, we use a designated data structure for every $u\in V$. We now describe this data structure.

\begin{definition} \label{def:MaxInRangeTree}
Given a vertex $u\in V$, we define the \textbf{max-in-range tree} of $u$ to be the following binary tree. Each node in the binary tree represents a sub-interval $[i_1,i_2]\subseteq[0,h]$, such that $i_1,i_2$ are even. The root represents the full interval $[0,h]$ and the leaves are intervals of the form $[i,i]$, for some $i\in[0,h]$. The left and right children of a node $[i_1,i_2]$ (where $i_1<i_2$) are $[i_1,j]$ and $[mid,i_2]$ respectively, where $mid$ is the even index among $\frac{i_1+i_2}{2}-1,\frac{i_1+i_2}{2}$, and $j$ is the even index in $[i_1,mid-2]$ that maximizes $\Delta_u(j)$.
\end{definition}

Note that $mid$ splits the interval $[i_1,i_2]$ into two almost equal parts. Thus the length of an interval in some node is at most half of the interval in the parent of this node. Hence, the depth of the max-in-range tree of every $u\in V$ is at most $\log h$, and it contains $O(h)$ nodes. Note also that this tree can be pre-computed, when constructing the oracle. The total size of these trees is $O(n\cdot h)$, and this term is dominated by the oracle's size $O(h\cdot n^{1+\frac{1}{k}})$.

\begin{proof}[Proof of Lemma \ref{lemma:PartialPRDO}]

Given $k$ and the $n$-vertex graph $G=(V,E)$, we start by creating a hierarchy of sets with $q_i=n^{-\frac{1}{k}}$ (see Section \ref{sec:HierarchyOfSets}). By the definition of the sets $A_i$, in expectation,
\[|A_h|=n\prod_{j=0}^{h-1}q_j=n\cdot n^{-\frac{h}{k}}=n^{1-\frac{h}{k}}~.\]
We choose $A_h$ to be the desired set $S\subseteq V$.

For every vertex $u\in V$ we define the \textit{cluster} of $u$ as
\begin{equation} \label{eq:ClusterDef}
    C(u)=\{v\in V\;|\;u\in\bigcup_{i=0}^{h}B_i(v)\}~.
\end{equation}
Given $v\in C(u)$, denote by $P_{v,u}$ the shortest path from $v$ to $u$ in $G$. Let $q_u(v)$ be a pointer from $v$ to the next vertex after $v$ in $P_{v,u}$.

For every $i=0,1,...,h$, let $(S^1_i,D^1_i)$ be the interactive $1$-distance preserver from Lemma \ref{lemma:PathsToPivots}, for the set $\{(v,p_i(v))\;|\;v\in V\}$. 

Our oracle $D$ stores the following information.
\begin{enumerate}
    \item For every $i=0,1,...,h$, the oracle $D$ stores the oracle $D^1_i$.
    \item For every $u\in V$ and $v\in C(u)$, the oracle $D$ stores the pointer $q_u(v)$.
    \item For every $i=0,1,...,h$ and $v\in V$, the oracle $D$ stores the pivot $p_i(v)$ and the bunch $B_i(v)$.
    \item For every $u\in V$, the oracle $D$ stores the max-in-range tree $T(u)$ of $u$ (see Definition \ref{def:MaxInRangeTree}).
\end{enumerate}

Given a query $(u,v)\in V^2$, we perform a binary search to find an index $i$ such that $P(u,v,i)=TRUE$ and either $Q(u,v,i)=TRUE$ or $i=h$ (by Lemma \ref{lemma:DesiredIndex}, applied to $i_1=0,i_2=h-2$, such an index exists). This binary search is performed as follows. Suppose that the current search range is $[i_1,i_2]$, and that it appears in the max-in-range tree $T(u)$ (at the beginning, $[i_1,i_2]=[0,h]$). Let $[i_1,j]$ and $[mid,i_2]$ be the left and right children of $[i_1,i_2]$, respectively (see Definition \ref{def:MaxInRangeTree}). If $Q(u,v,j)=TRUE$, continue to search in the range $[i_1,j]$. Otherwise, if $Q(u,v,j)=FALSE$, continue to $[mid,i_2]$. 


Finally, when $i_1=i_2=i$, we check whether $P(u,v,i)=Q(u,v,i)=TRUE$. If this condition does not hold, we say that this binary search failed, and we perform a symmetric binary search for the pair $(v,u)$ - i.e., where the roles of $u$ and $v$ are switched. If this binary search also fails, we return the two shortest paths from $u$ to $p_h(u)$ and from $v$ to $p_h(v)$, using the oracle $D^1_h$. Note that indeed $p_h(u),p_h(v)\in A_h=S$. 

Suppose, however, that one of the two binary searches succeeded. Without loss of generality, assume that $P(u,v,i)=Q(u,v,i)=TRUE$ for the last index $i$ that was considered in the binary search of $(u,v)$. Then in particular 
\[p_i(u)\in B_i(v)\;\vee\;p_{i+1}(v)\in B_{i+1}(u)~.\] 
We consider two cases.
\begin{itemize}
    \item If $p_i(u)\in B_i(v)$, we use $D^1_h$ to find the shortest path from $u$ to $p_i(u)$, and we use the pointers $q_{p_i(u)}(\cdot)$, starting at $v$, to find the shortest path $P_{p_i(u),v}$. The output path is the concatenation of these two paths.
    \item Otherwise, if $p_i(u)\notin B_i(v)$ and $p_{i+1}(v)\in B_{i+1}(u)$, we symmetrically output a path that consists of the concatenation between the shortest paths from $v$ to $p_{i+1}(v)$ and from $p_{i+1}(v)$ to $u$. The latter is constructed via the pointers $q_{p_{i+1}(v)}(\cdot)$, starting at $u$.
\end{itemize}
If, instead, the binary search of $(v,u)$ is the one that succeeds, we symmetrically return a path between $v$ and $u$, that passes through either $p_i(v)$ or $p_{i+1}(u)$.

For the query algorithm to be well-defined, first notice that if $p_i(u)\in B_i(v)$, then by definition $v\in C(p_i(u))$. Similarly, if $p_{i+1}(v)\in B_{i+1}(u)$, we have $u\in C(p_{i+1}(v))$. We claim that for every $a,b\in V$ such that $b\in C(a)$, we can indeed use only the pointers $q_a(\cdot)$, starting at $b$, to construct the shortest path from $a$ to $b$. For this purpose, it suffices to prove that $x\in C(a)$ for every $x$ on this shortest path, and thus $q_a(x)$ is the next vertex after $x$ on this shortest path from $b$ to $a$. We remark that this property is well-known for the standard definition of clusters (see \cite{TZ01}), i.e., for the case that $h=k-1$. We show that it also holds for clusters defined by (\ref{eq:ClusterDef}), for any $0\leq h\leq k-1$.

\begin{claim}
    Let $a,b\in V$ such that $b\in C(a)$. Then, for every vertex $x$ on the $a-b$ shortest path, we have $x\in C(a)$.
\end{claim}

\begin{proof}

Since $b\in C(a)$, we know that there is some $i\in[0,h]$ such that $a\in B_i(b)$, i.e., $a\in A_i$ and $d(b,a)<d(b,p_{i+1}(b))$. By the fact that $x$ is on the shortest path between $a,b$, we can write
\begin{eqnarray*}
d(x,a)&=&d(b,a)-d(b,x)\\
&<&d(b,p_{i+1}(b))-d(b,x)\\
&\leq&d(b,p_{i+1}(x))-d(b,x)\\
&\leq&d(b,x)+d(x,p_{i+1}(x))-d(b,x)\\
&=&d(x,p_{i+1}(x))~.
\end{eqnarray*}
Hence, by definition, $a\in B_i(x)$, for the same index $i\in[0,h]$. Therefore $x\in C(a)$, as desired.
    
\end{proof}

For the correctness of the binary search of $(u,v)$, note the following invariant: 
\begin{claim}
    If the current search range is $[i_1,i_2]$, then $P(u,v,i_1)=TRUE$, and either $Q(u,v,i_2)=TRUE$ or $i_2=h$.
\end{claim}

\begin{proof}

This invariant trivially holds for the first search range $[0,h]$, since
\[d(u,p_0(u))=d(u,u)=0\leq0\cdot d(u,v)~.\]
For smaller search ranges, consider two cases. If the current search range is $[i_1,j]$, which is the left child of $[i_1,i_2]$ (that satisfies the invariant), then we already have $P(u,v,i_1)=TRUE$. Recall also that the reason of searching in the left child $[i_1,j]$ is that $Q(u,v,j)=TRUE$. Hence, in particular, the invariant indeed holds.

If the current search range is $[mid,i_2]$, which is the right child of $[i_1,i_2]$ (that satisfies the invariant), then it means that for the even index $j\in[i_1,mid-2]$ that maximizes $\Delta_u(j)$, we have $Q(u,v,j)=FALSE$. But then, by Lemma \ref{lemma:FailedStep}, we know that $\Delta_u(j)\leq2d(u,v)$, and therefore $\Delta_u(i)\leq2d(u,v)$ for \textbf{all} even $i\in[i_1,mid-2]$. Hence, we get
\begin{eqnarray*}
d(u,p_{mid}(u))&=&d(u,p_{i_1}(u))+\sum_{i}\Delta_u(i)\\
&\leq&d(u,p_{i_1}(u))+\sum_{i}2d(u,v)\\
&\leq&d(u,p_{i_1}(u))+\frac{mid-i_1}{2}\cdot2d(u,v)\\
&\leq&d(u,p_{i_1}(u))+(mid-i_1)d(u,v)\\
&\leq&i_1\cdot d(u,v)+(mid-i_1)d(u,v)\\
&=&mid\cdot d(u,v)~.
\end{eqnarray*}
The sum goes over only even indices in $[i_1,mid-2]$, and we used the fact that $P(u,v,i_1)=TRUE$. We conclude that here also $P(u,v,mid)=TRUE$, and also $Q(u,v,i_2)=TRUE$ or $i_2=h$, since this is a part of the assumption on $[i_1,i_2]$.
    
\end{proof}

As a result, at the end of the algorithm, when we get to a search range of the form $[i,i]$, we have $P(u,v,i)=TRUE$, and either $Q(u,v,i)=TRUE$ or $i=h$. If $Q(u,v,i)=TRUE$, we return a $(2h+1)$-stretch path between $u,v$ (See Inequalities (\ref{eq:ConcatBound1}) and (\ref{eq:ConcatBound2})). Otherwise, we conclude that $P(u,v,h)=TRUE$, and thus the shortest path from $u$ to $p_h(u)$ is of weight at most $h\cdot d(u,v)$. Symmetrically, if the second binary search succeeds, then we return a $(2h+1)$-stretch path between $u,v$, and otherwise we have $d(v,p_h(v))\leq h\cdot d(u,v)$. If both binary searches fail, it implies that the two paths that we return are of weight at most $h\cdot d(u,v)$. In conclusion, our oracle either returns a $(2h+1)$-stretch path between $u,v$, or two paths of weight at most $h\cdot d(u,v)$: from $u$ to $u'=p_h(u)\in A_h=S$, and from $v$ to $v'=p_h(v)\in S$.

Next, we analyse the query time and the size of our oracle $D$. Its query time consists of performing the two binary searches, then applying $D^1_i$ once or twice for some $i\in[0,h]$, and also using the pointers $q_u(\cdot)$. The binary searches require $O(\log h)$ time (this is the depth of the max-in-range trees $T(u),T(v)$), while querying $D^1_i$ and using the pointers $q_u(\cdot)$ is done within a linear time in the length of the resulting paths. Hence, we get a total query time of $O(\log h)$.

Regarding the size of $D$, notice that each oracle $D^1_i$, for every $i=0,1,...,h$, is of size $O(n)$. Therefore, storing them requires $O(h\cdot n)$ space. The collection of pointers $q_u(v)$, for every $u,v\in V$ such that $v\in C(u)$, is of expected size
\begin{eqnarray*}
\sum_{u\in V}\sum_{v\in C(u)}1&=&\sum_{v\in V}\sum_{u|v\in C(u)}1
=\sum_{v\in V}\sum_{u\in\bigcup B_i(v)}1\\
&\leq&\sum_{v\in V}\sum_{i=0}^h|B_i(v)|
=\sum_{i=0}^h\sum_{v\in V}|B_i(v)|\\
&\leq&\sum_{i=0}^h\frac{n}{q_i}
=\sum_{i=0}^hn^{1+\frac{1}{k}}
=(h+1)n^{1+\frac{1}{k}}~.
\end{eqnarray*}
Here we used the definition of $C(u)$ and a bound on $\sum_{v\in V}|B_i(v)|$ which is achieved similarly to the proof of Lemma \ref{lemma:BasicSizes} in \nameref{sec:AppendixB}. Thus, the required storage for the pointers $q_u(v)$, for every $u,v\in V$ such that $v\in C(u)$, is
\begin{equation} \label{eq:PointersSize}
    O(hn^{1+\frac{1}{k}})~.
\end{equation}

Note that storing $p_i(v),B_i(v)$ for every $v\in V$ and $i=0,1,...,h$ is exactly like storing $\bar{H}_i$ for every $i=0,1,...,h$. By Lemma \ref{lemma:BasicSizes}, the expected total sizes of these sets is 
\[\sum_{i=0}^h\frac{n}{q_i}\prod_{j=0}^{i-1}q_j=\sum_{i=0}^hn^{1-\frac{i-1}{k}}\leq(h+1)n^{1+\frac{1}{k}}~.\]

Lastly, we note that for every $u\in V$, the max-in-range tree $T(u)$ has $O(h)$ nodes, each of which stores two indices. Hence, the space needed to store $T(u)$ for every $u\in V$ is $O(h\cdot n)$. We therefore conclude that the total size of $D$ is $O(h\cdot n^{1+\frac{1}{k}})$.


Finally, notice that the output paths by the oracle $D$ are contained in the union $Z_1\cup Z_2$, where
\[Z_1=\bigcup_{i=0}^hS^1_i\text{ and }Z_2=\bigcup_{u\in V}\bigcup_{v\in C(u)}P_{u,v}~,\]
and $P_{u,v}$ is the shortest path between $u$ and $v$. By Lemma \ref{lemma:PathsToPivots}, we know that each $S^1_i$ is of size $O(n)$. Thus, $|Z_1|=O(h\cdot n)$. Regarding $Z_2$, notice that each edge of this set is of the from $(b,q_a(b))$ for some $a,b\in V$ such that $b\in C(a)$. Thus, the total size of this set is bounded by the number of the pointers $q_u(\cdot)$. By (\ref{eq:PointersSize}), this size is $O(h\cdot n^{1+\frac{1}{k}})$. We conclude that there is a set $E'\subseteq E$ such that all of the returned paths by the oracle $D$ are in $E'$, and
\[|E'|=O(h\cdot n^{1+\frac{1}{k}})~.\]

\end{proof}

\section{Appendix E} \label{sec:AppendixE}

In this appendix we prove Theorem \ref{thm:UseStretchFriendly}.

\begin{theorem*}[Theorem \ref{thm:UseStretchFriendly}]
Suppose that any $n$-vertex graph admits an interactive $\alpha$-spanner with query time $q$ and size $h(n)$. Then, given a positive number $t\geq1$, every $n$-vertex graph $G$ also admits an interactive $O(\alpha\cdot t)$-spanner with query time $O(q)$ and size $O(h(\frac{n}{t})+n)$.
\end{theorem*}

\begin{proof}

Given $G=(V,E)$ and $t$, let $\mathcal{C}$ be a stretch-friendly $O(t)$-partition of $G$ with $|\mathcal{C}|\leq\frac{n}{t}$ (See Definition \ref{def:StretchFriendly}). Denote by $\mathcal{C}(v)$ the unique cluster in $\mathcal{C}$ that contains the vertex $v$. Recall that every $C\in\mathcal{C}$ has a spanning tree $T_C$, rooted at some vertex $r_C$, such that for every $v\in C$, the unique path between $v$ and $r_C$ in $T_C$ has $O(t)$ edges. For $x,y\in C$, we denote by $T_C[x,y]$ the unique path in $T_C$ between the vertices $x,y$. Generally, $T_C[x,y]$ is the concatenation of $T_C[x,lca(x,y)]$ and $T_C[lca(x,y),y]$, where $lca(x,y)$ is the lowest common ancestor of $x,y$ in $T_C$.

Now, the cluster graph $\mathcal{H}=(V_{\mathcal{H}},E_{\mathcal{H}})$ is defined as follows. The vertex set $V_{\mathcal{H}}$ is the set of clusters $\mathcal{C}$. For every $C,C'\in\mathcal{C}$, denote by $e(C,C')$ the edge in $E$ with the \textit{minimal} weight among the edges between $C$ and $C'$. If there is no such edge, denote $e(C,C')=\bot$. Then,
\[E_{\mathcal{H}}=\{(C,C')\;|\;C,C'\in V_{\mathcal{H}}\text{ and }e(C,C')\neq\bot\}~.\]
The weight of an edge $(C,C')\in E_{\mathcal{H}}$ is defined to be $w(e(C,C'))$.

By our assumption, the graph $\mathcal{H}$ admits an interactive $\alpha$-spanner $(S,D)$ with query time $q$ and size $h(\frac{n}{t})$. We define our oracle $D'$ for the new interactive spanner to contain the following.
\begin{enumerate}
    \item The oracle $D$.
    \item For every $v\in V$, its cluster $C=\mathcal{C}(v)$, a pointer $p(v)$ to its parent in the tree $T_C$ (for $r_C$, we set $p(r_C)=r_C$) and the number of edges in the path $T_C[v,r_C]$.
    \item $e(C,C')$, for every edge $(C,C')\in S$.
\end{enumerate}

The size of this oracle is then $h(\frac{n}{t})+|V|\cdot O(1)+|S|\cdot O(1)=O(h(\frac{n}{t})+n)$.

Given a query $(u,v)\in V^2$, our oracle uses $D$ to find a path $Q\subseteq S$ between $\mathcal{C}(u)$ and $\mathcal{C}(v)$. Let $Q=(\mathcal{C}(u)=C_1,C_2,...,C_l=\mathcal{C}(v))$, and for every $i\in[0,l-1]$, denote $(x_i,y_i)=e(C_i,C_{i+1})$. Also, let $y_0=u$ and $x_l=v$. Note that for every $i=1,...,l$, the vertices $x_i,y_{i-1}$ are in the same cluster $C_i$, and denote $P_i=T_{C_i}[y_{i-1},x_i]$. We later show how $P_i$ can be efficiently found, using the stored data in our oracle. The returned path for the query $u,v$ is then
\[P_1\circ(x_1,y_1)\circ P_2\circ(x_2,y_2)\circ\cdots\circ(x_{l-1},y_{l-1})\circ P_l~.\]
See Figure \ref{fig:StretchFriendly2} for an illustration.

\begin{figure}[!ht]
\centerline{\includegraphics[width=10.8cm, height=4.5cm]{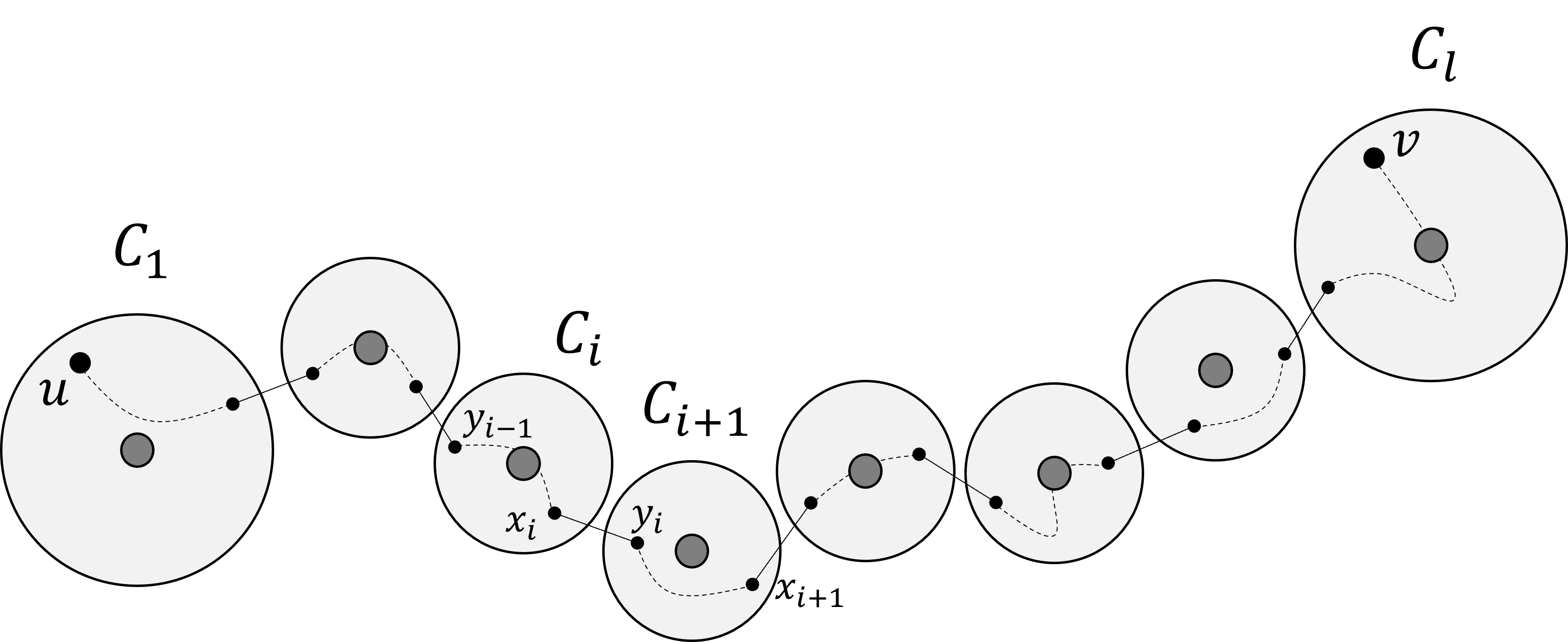}}
\caption{The returned path. The dashed lines are the paths $P_i$. Notice that they do not always pass through the roots of their respective clusters $C_i$.}
\label{fig:StretchFriendly2}
\end{figure}

Next, we describe in more detail how to find the paths $\{P_i\}$. Let $x,y$ be vertices in some cluster $C$. We want to find the unique path between $x,y$ in $T_C$, i.e., $T_C[x,y]$, in time that is proportional to length of this path. If $x=y$, then the desired path is empty. Otherwise, if $|T_C[x,r_C]|>|T_C[y,r_C]|$ (the number of edges in $T_C[x,r_C]$ is larger than the number of edges in $T_C[y,r_C]$), find $T_C[p(x),y]$ recursively and return $(x,p(x))\circ T_C[p(x),y]$. Symmetrically, if $|T_C[y,r_C]|\geq|T_C[x,r_C]|$, return $(y,p(y))\circ T_C[p(y),x]$. The correctness of this algorithm follows from the fact that if, for example, $|T_C[y,r_C]|\geq|T_C[x,r_C]|$, and $x\neq y$, then $x$ cannot be in the sub-tree of $y$ in $T_C$, hence the path $T_C[x,y]$ must pass through $p(y)$. In each recursive call we reduce the sum $|T_C[x,r_C]|+|T_C[y,r_C]|$ by $1$, and therefore the algorithm ends either when $x=y=lca(x,y)$, or when $|T_C[x,r_C]|=|T_C[y,r_C]|=0$ - which also means $x=y=r_C=lca(x,y)$. Therefore, the time needed to find $T_C[x,y]$ is $O(1)$ times the total length of $T_C[x,lca(x,y)]$ and $T_C[y,lca(x,y)]$, which is $O(|T_C[x,lca(x,y)]\circ T_C[lca(x,y),y]|)=O(|T_C[x,y]|)$.

The running time of the whole query algorithm of $D'$ is proportional to the number of edges in the returned path, plus the query time of $D$, which is $O(q)$.

We now analyse the stretch. For every $i$, the path $P_i$ consists of two paths, 
\[P^1_i=T_{C_i}[y_{i-1},lca(x_i,y_{i-1})]\text{ and }\]
\[P^2_i=T_{C_i}[lca(x_i,y_{i-1}),x_i]~.\]
Since these are sub-paths of the paths from $y_{i-1},x_i$ to the root, both paths have at most $O(t)$ edges. By the stretch-friendly property of $\mathcal{C}$, we know that the weight of each edge in $P^1_i$ is at most $w(x_{i-1},y_{i-1})$ (for $i>1$), and that the weight of each edge in $P^2_i$ is at most $w(x_i,y_i)$ (for $i<l$).

Therefore, the weight of the returned path is at most
\begin{eqnarray*}
\sum_{i=1}^lw(P_i)+\sum_{i=1}^{l-1}w(x_i,y_i)
&=&w(P^1_1)+\sum_{i=1}^{l-1}\big(w(P^2_i)+w(x_i,y_i)+w(P^1_{i+1})\big)+w(P^2_l)\\
&\leq&w(P^1_1)+\sum_{i=1}^{l-1}O(t)w(x_i,y_i)+w(P^2_l)\\
&=&w(P^1_1)+O(t)\cdot w(Q)+w(P^2_l)~,
\end{eqnarray*}
where $Q$ is the returned path from the oracle $D$ in the graph $\mathcal{H}$.

Now let $P_{u,v}$ be the shortest path in $G$ between $u,v$. Let $C'_1,C'_2,...,C'_f$ be the sequence of the clusters of $\mathcal{C}$ that $P_{u,v}$ passes through, in the same order that $P_{u,v}$ passes through them. In case it passes through some $C'_j$ more than once, $C'_j$ will appear multiple times. Then, $\hat{P}=(C'_1,...,C'_f)$ is a path in $\mathcal{H}$ between $\mathcal{C}(u)$ and $\mathcal{C}(v)$, and since the stretch of $D$ is $\alpha$, we get that 
\[w(Q)\leq\alpha\cdot d_{\mathcal{H}}(\mathcal{C}(u),\mathcal{C}(v))\leq\alpha\cdot w(\hat{P})~.\] 
But notice that for each $j=1,...,f-1$, $P_{u,v}$ contains an edge $e_j$ between the clusters $C'_j$ and $C'_{j+1}$, and since in $\mathcal{H}$ we defined $w(C'_j,C'_{j+1})$ to be the weight of the minimal edge between these clusters, we have
\begin{eqnarray*}
w(Q)&\leq&\alpha\cdot w(\hat{P})=\alpha\sum_{j=1}^{f-1}w(C'_j,C'_{j+1})\leq\alpha\sum_{j=1}^{f-1}w(e_j)
\leq\alpha\cdot w(P_{u,v})~.
\end{eqnarray*}

It is left to bound the terms $w(P^1_1),w(P^2_l)$. We show that $w(P^1_1)\leq O(t)\cdot w(P_{u,v})$, and the proof that $w(P^2_l)\leq O(t)\cdot w(P_{u,v})$ is symmetric. 

Recall that $Q=(C_1,...,C_l)$ is the path returned by $D$ for the query $(\mathcal{C}(u),\mathcal{C}(v))$. If the entire shortest path $P_{u,v}$ between $u,v$ in $G$ is contained in the cluster $\mathcal{C}(u)$, then in particular $\mathcal{C}(u)=\mathcal{C}(v)=C_1$. Thus, $Q$ is of length $0$, and the returned path from our oracle is just $T_{C_1}[u,v]$, which is the concatenation of two paths with $O(t)$ edges each. By the stretch-friendly property, for every edge $(x,y)\in P_{u,v}$, the edges of $T_{C_1}[x,y]$ have weight at most $w(x,y)\leq w(P_{u,v})$. Notice that the union $\bigcup_{(x,y)\in P_{u,v}}T_{C_1}[x,y]$ is a (not necessarily simple) path in $T_{C_1}$ between $u$ and $v$, and so it contains $T_{C_1}[u,v]$. Therefore each one of the $O(t)$ edges in $T_{C_1}[u,v]$ has weight of at most $w(P_{u,v})$. Thus, the returned path has weight of at most $O(t)\cdot w(P_{u,v})$.

Otherwise, let $v'$ be the first vertex in $P_{u,v}$ which is not in $C_1$, and let $u'$ be the vertex that precedes $v'$ in $P_{u,v}$. Similarly to the first case, we can prove that $w(T_{C_1}[u,u'])\leq O(t)\cdot w(P_{u,v})$. In addition, again by the stretch-friendly property, each of the $O(t)$ edges in $T_{C_1}[r_{C_1},u']$ has weight of at most $w(u',v')\leq w(P_{u,v})$. Since $P^1_1\subseteq T_{C_1}[u,r_{C_1}]\subseteq T_{C_1}[u,u']\cup T_{C_1}[r_{C_1},u']$, we get
\begin{eqnarray*}
w(P^1_1)&\leq&w(T_{C_1}[u,u'])+w(T_{C_1}[r_{C_1},u'])
\leq O(t)\cdot w(P_{u,v})+O(t)\cdot w(P_{u,v})
= O(t)\cdot w(P_{u,v})~.
\end{eqnarray*}

We proved that the weight of the returned path is at most
\begin{eqnarray*}
w(P^1_1)+w(P^2_l)+O(t)\cdot w(Q)
&\leq&O(t)\cdot w(P_{u,v})+O(t)\cdot w(P_{u,v})+O(t)\alpha\cdot w(P_{u,v})\\
&=&O(\alpha\cdot t)\cdot w(P_{u,v})~.
\end{eqnarray*}

Thus the stretch of our new interactive spanner is $O(\alpha\cdot t)$.

For completing the proof, notice that the output paths from our oracle $D'$ are always a concatenation of paths $T_C[x,y]$ for some $C\in\mathcal{C}$ and $x,y\in C$, and edges $e(C,C')$, where $(C,C')\in S$. These paths are always contained in the set $\bigcup_{C\in\mathcal{C}}T_C\cup\{e(C,C')\;|\;(C,C')\in S\}$. The set $\bigcup_{C\in\mathcal{C}}T_C$ is a forest, and thus has size at most $ n-1$. The set $\{e(C,C')\;|\;(C,C')\in S\}$ is of size $|S|\leq h(\frac{n}{t})$. Therefore, the size of this entire set is $h(\frac{n}{t})+n$.

\end{proof}

To prove Theorem \ref{thm:UltraSparse} (i.e., to build ultra-sparse interactive spanners), we slightly change the construction from the proof of Theorem \ref{thm:UseStretchFriendly}.

\begin{proof}[Proof of Theorem \ref{thm:UltraSparse}]

As in the proof of Theorem \ref{thm:UseStretchFriendly}, we construct a stretch-friendly $O(t)$-partition $\mathcal{C}$ of $G$ with $|\mathcal{C}|\leq\frac{n}{t}$. For every $C\in\mathcal{C}$, let $T_C$ and $r_C$ denote the corresponding spanning tree and root, respectively, and let $T_C[x,y]$ denote the unique path in $T_C$ between the vertices $x,y\in C$. The cluster graph $\mathcal{H}=(V_{\mathcal{H}},E_{\mathcal{H}})$ and the edges $e(C,C')$ are defined in the same way as in the proof of Theorem \ref{thm:UseStretchFriendly}.

By our assumption, the graph $\mathcal{H}$ admits an interactive $\alpha$-spanner $(S,D)$ with query time $q$ and size $h(\frac{n}{t})$. Our oracle for the new interactive spanner will still contain the oracle $D$ and the edges $\{e(C,C')\;|\;(C,C')\in S\}$. Also, for every $v\in V$, we will store the pointer $p(v)$ to its parent in the tree $T_C$, of the cluster $C$ such that $v\in C$. We will \textit{no longer} store $\mathcal{C}(v)$ (see definition in the proof of Theorem \ref{thm:UseStretchFriendly}) and the number of edges in the path $T_C[v,r_C]$.

Given a query $(u,v)\in V^2$, we will follow the pointers from $u$ and from $v$ until we get to the roots $r_u,r_v$ of the respective clusters that contains $u$ and $v$. As a result, in time $O(t)$ we discover the identities of the two respective clusters (i.e., we find $\mathcal{C}(u),\mathcal{C}(v)$), and the number of edges in the paths from $u,v$ to $r_u,r_v$, respectively, in the spanning trees of their respective clusters. From this point on, the query algorithm proceeds as described in the proof of Theorem \ref{thm:UseStretchFriendly}.

The stretch of the resulting oracle remains $O(\alpha\cdot t)$. The query time grows to $O(q+t)$. The size of the underlying spanner also stays the same as in Theorem \ref{thm:UseStretchFriendly}, because the same output paths are returned. Lastly, the storage required for the oracle $D$ and the edges $\{e(C,C')\;|\;(C,C')\in S\}$ is of size $O(h(\frac{n}{t}))$. Storing the pointers $p(\cdot)$ requires at most $n$ words. We conclude that the size of our new oracle is $n+O(h(\frac{n}{t}))$.

\end{proof}

\section{Appendix F} \label{sec:AppendixF}

In this appendix, we present an interactive emulator, based on the \textit{neighborhood covers} of Awerbuch and Peleg \cite{AP90a} and on the non-path-reporting distance oracle due to Mendel and Naor \cite{MN06}. We use the following notations regarding graph covers. Given an undirected weighted graph $G=(V,E)$, a \textbf{cover} of $G$ is a collection $\mathcal{S}=\{S_1,S_2,...,S_t\}$ of subsets of $V$ (called \textbf{clusters}), such that $V=\bigcup_{i=1}^tS_i$. Note that the clusters are not necessarily disjoint. Given a cluster $S\in\mathcal{S}$, let $G[S]$ be the induced sub-graph of $G$, on the vertices of $S$. We define the \textbf{radius} of $\mathcal{S}$, denoted by $Rad(\mathcal{S})$, to be the minimal number $r\geq0$ such that in every $S\in\mathcal{S}$, there is a vertex $v_S\in S$ that satisfies $d_{G[S]}(v_S,u)\leq r$, for all $u\in S$. If such $r$ does not exist, we denote $Rad(\mathcal{S})=\infty$. Lastly, denote by $|\mathcal{S}|$ the number of clusters in $\mathcal{S}$.

The following theorem is from \cite{AP90a}.

\begin{theorem}[Theorem 3.1 in \cite{AP90a}] \label{thm:CoarseningCover}
Let $k\geq1$ be an integer, let $G=(V,E)$ be an undirected weighted $n$-vertex graph, and let $\mathcal{S}$ be a cover of $G$. There exists another cover $\mathcal{T}$ of $G$ that satisfies the following properties.
\begin{enumerate}
    \item For every $S\in\mathcal{S}$ there is a cluster $T\in\mathcal{T}$ such that $S\subseteq T$.
    \item $Rad(\mathcal{T})\leq(2k-1)Rad(\mathcal{S})$.
    \item For every $v\in V$, the number of clusters $T\in\mathcal{T}$ such that $v\in T$, is at most $2k|\mathcal{S}|^{\frac{1}{k}}$.
\end{enumerate}
\end{theorem}

Using this theorem, we now construct the desired interactive emulator.

\begin{theorem} \label{thm:APEmulator}
Given an undirected weighted $n$-vertex graph $G=(V,E)$, an integer parameter $k\geq1$, and a parameter $\epsilon>0$, there is an interactive $(4+\epsilon)k$-emulator with size $O(k\cdot\frac{\log\Lambda}{\epsilon}\cdot n^{1+\frac{1}{k}})$ and query time $O(\log\log k+\log\frac{1}{\epsilon})$, where $\Lambda=\frac{\max_{u,v}d_G(u,v)}{\min_{u\neq v}d_G(u,v)}$ is the aspect ratio of $G$.
\end{theorem}

\begin{proof}

Throughout the proof we assume that the minimal weight of an edge in $G$ is $1$, as otherwise the weights can be re-scaled by dividing each weight by the minimal weight of an edge. This way, the aspect ratio is $\Lambda=\max_{u,v}d_G(u,v)$, i.e., the \textit{diameter} of $G$. Note that for every two distinct vertices $u,v\in V$, the distance $d_G(u,v)$ is between $1$ and $\Lambda$.

Fix some positive number $W>0$. Given the graph $G$, we define a cover $\mathcal{S}_W$ by taking the $W$-\textit{neighborhood} of each vertex in $G$: for every $v\in V$, denote
\[S(v,W)=\{u\in V\;|\;d_G(v,u)\leq W\}~.\]
The collection $\mathcal{S}_W$ is then defined as $\mathcal{S}_W=\{S(v,W)\;|\;v\in V\}$. This is a cover of $G$ with radius $Rad(\mathcal{S}_W)\leq W$. By Theorem \ref{thm:CoarseningCover}, there is a cover $\mathcal{T}_W$ such that
\begin{enumerate}
    \item For every $v\in V$, there is a cluster $T(v,W)\in\mathcal{T}_W$ such that $S(v,W)\subseteq T(v,W)$. We call $T(v,W)$ the \textbf{home-cluster} of $v$.
    \item $Rad(\mathcal{T}_W)\leq(2k-1)Rad(\mathcal{S}_W)\leq(2k-1)W$. That is, for every $T\in\mathcal{T}_W$ there is a vertex $v_T\in T$ such that every vertex in the induced graph $G[T]$ has distance at most $(2k-1)W$ from $v_T$. Hence, there is a tree on the vertices of $T$, rooted at $v_T$, such that every vertex of this tree has distance at most $(2k-1)W$ from the root $v_T$. Let $q_T(u)$, for every $u\in T$, be a pointer to the next vertex on the unique path in this tree from $u$ to $v_T$. Similarly, $h_T(u)$ denotes the number of edges in this unique path, i.e., the depth of $u$ in this tree.
    \item For every $u\in V$, the number of clusters $T\in\mathcal{T}_W$ such that $u\in T$ is at most
    \begin{equation} \label{eq:NumOfClusters}
        2k|\mathcal{S}_W|^{\frac{1}{k}}=2k\cdot n^{\frac{1}{k}}
    \end{equation}
\end{enumerate}

We call the cover $\mathcal{T}_W$ a \textbf{neighborhood cover} for the parameter $W$. For a given $\epsilon>0$, we define a \textit{hierarchy} of neighborhood covers: we construct the neighborhood cover $\mathcal{T}_W$, for every $W\in\{(1+\epsilon)^i\}_{i=0}^{\lambda}$, where $\lambda=\lceil\log_{1+\epsilon}\Lambda\rceil$. For every $v\in V$ and $i\in[0,\lambda]$, we define a pointer $HC^i(v)$ to the home-cluster of $v$ in the neighborhood cover $\mathcal{T}_{(1+\epsilon)^i}$. Recall also that every cluster $T\in\mathcal{T}_{(1+\epsilon)^i}$ has a spanning tree rooted at a vertex $v_T$. For every vertex $u\in T$, the pointer $q_T(u)$ points at the next vertex on the unique path from $u$ to the vertex $v_T$, and the variable $h_T(u)$ denotes the number of edges in this path.

We now describe the construction of our interactive emulator. We define an oracle $D$ that stores the following information.
\begin{enumerate}
    \item For every $v\in V$ and $i\in[0,\lambda]$, the oracle $D$ stores $HC^i(v)$, the home-cluster of $v$ in $\mathcal{T}_{(1+\epsilon)^i}$.
    \item For every $i\in[0,\lambda]$, every cluster $T\in\mathcal{T}_{(1+\epsilon)^i}$, and every vertex $u\in T$, the oracle $D$ stores the pointer $q_T(u)$ and the variable $h_T(u)$.
    \item The oracle $D$ stores the distance oracle of Mendel and Naor \cite{MN06} for the graph $G$.
\end{enumerate}

Next, we describe the algorithm for answering queries. Given a query $(u,v)\in V^2$, our oracle $D$ first queries the Mendel-Naor distance oracle on $(u,v)$. Recall that this distance oracle is not path-reporting. Thus, the result is only an estimate $\hat{d}=\hat{d}(u,v)$ such that $d_G(u,v)\leq\hat{d}\leq c_{MN}\cdot k\cdot d_G(u,v)$, or equivalently
\begin{equation} \label{eq:MNEstimate}
    \frac{\hat{d}}{c_{MN}\cdot k}\leq d_G(u,v)\leq\hat{d}~.
\end{equation}
Here $c_{MN}$ is the constant coefficient of the stretch in the Mendel-Naor distance oracle. If $\hat{d}>\Lambda$, we can safely replace $\hat{d}$ by $\Lambda$, and inequality (\ref{eq:MNEstimate}) will still be satisfied.

Then, our oracle $D$ finds the smallest $i$ such that $(1+\epsilon)^i\in[\frac{\hat{d}}{c_{MN}\cdot k},\hat{d}]$, and also $u$ is in the home-cluster of $v$, in the cover $\mathcal{T}_{(1+\epsilon)^i}$. It finds this $i$ by performing a binary search in the range $\mathcal{I}=\left[\left\lfloor\log_{1+\epsilon}\left(\frac{\hat{d}}{c_{MN}\cdot k}\right)\right\rfloor,\lceil\log_{1+\epsilon}\hat{d}\rceil\right]$. During this binary search, if the current search range is $[i_1,i_2]$, and the index $i$ is in the middle of this range, then the oracle $D$ refers to the home-cluster $T=HC^i(v)$, and checks whether $q_T(u)$ is defined (that is, if $u\in T$). If it is defined, then the binary search continues in the range $[i_1,i]$, and otherwise in the range $[i+1,i_2]$.

Lastly, the oracle $D$ uses the pointers $q_T(\cdot)$ and the variables $h_T(\cdot)$, to find a path in the spanning tree of $T$ between $u$ and $v$. This is done in the same way as in the proof of Theorem \ref{thm:UseStretchFriendly} in \nameref{sec:AppendixE}. The resulting path is the output path of $D$ for the query $(u,v)$.

For the query algorithm to be well defined, note that $\mathcal{I}\subseteq[0,\lambda]$, and thus we did construct the neighborhood cover $\mathcal{T}_{(1+\epsilon)^i}$ for every $i\in\mathcal{I}$. This is true since $\lceil\log_{1+\epsilon}\hat{d}\rceil\leq\lceil\log_{1+\epsilon}\Lambda\rceil=\lambda$. We also show that there must be an index $i$ in the range $\mathcal{I}$ such that $u$ is in the home-cluster of $v$ in $\mathcal{T}_{(1+\epsilon)^i}$. Indeed, since we know that $d_G(u,v)\leq\hat{d}$ (by inequality (\ref{eq:MNEstimate})), then for $i'=\lceil\log_{1+\epsilon}\hat{d}\rceil\in\mathcal{I}$, we have
\[(1+\epsilon)^{i'}\geq\hat{d}\geq d_G(u,v)~.\]
Since the home-cluster of $v$ in the cover $\mathcal{T}_{(1+\epsilon)^{i'}}$ contains all of the vertices within distance $(1+\epsilon)^{i'}$ from $v$, then in particular, it must contain $u$.

To analyse the stretch of the oracle $D$, we consider two cases regarding the index $i$, which is the smallest in $\mathcal{I}$ such that $u$ is in the home-cluster of $v$ in $\mathcal{T}_{(1+\epsilon)^i}$. In the first case, if $i-1$ is also in $\mathcal{I}$, then we know that $u$ is not in the home-cluster of $v$ in $\mathcal{T}_{(1+\epsilon)^{i-1}}$. Since this home-cluster contains all of the vertices within distance $(1+\epsilon)^{i-1}$ from $v$, it must be that 
\[(1+\epsilon)^{i-1}<d_G(u,v)~.\]
In the second case, where $i=\lfloor\log_{1+\epsilon}\left(\frac{\hat{d}}{c_{MN}\cdot k}\right)\rfloor$, then we also know by inequality (\ref{eq:MNEstimate}) that
\[(1+\epsilon)^{i-1}<(1+\epsilon)^{\log_{1+\epsilon}\left(\frac{\hat{d}}{c_{MN}\cdot k}\right)}=\frac{\hat{d}}{c_{MN}\cdot k}\leq d_G(u,v)~.\]
In both cases we saw that $(1+\epsilon)^{i-1}<d_G(u,v)$.

Now, note that the output path is the concatenation of two paths, from $u$ to $x$ and from $x$ to $v$, where $x$ is the lowest common ancestor of $u,v$ in the spanning tree of the home-cluster of $v$ in $\mathcal{T}_{(1+\epsilon)^i}$. Therefore, the weight of this path is at most twice the radius of this home-cluster, which by Theorem \ref{thm:CoarseningCover} is
\[Rad(\mathcal{T}_{(1+\epsilon)^i})\leq(2k-1)(1+\epsilon)^i~.\]
We conclude that the weight of the resulting output path is at most
\[2(2k-1)(1+\epsilon)^i<4k\cdot(1+\epsilon)^i<4(1+\epsilon)k\cdot d_G(u,v)~.\]
That is, our oracle $D$ has stretch $4(1+\epsilon)k$.

The query time of the oracle $D$ consists of applying the Mendel-Naor distance oracle, which requires $O(1)$ time, then performing a binary search in the range $\mathcal{I}=[\lfloor\log_{1+\epsilon}\left(\frac{\hat{d}}{c_{MN}\cdot k}\right)\rfloor,\lceil\log_{1+\epsilon}\hat{d}\rceil]$. Each step in this binary search is done in a constant time (checking whether $u$ is in the home-cluster of $v$). Thus, the running time of this binary search is
\begin{eqnarray*}
O(\log|\mathcal{I}|)&=&O\left(\log\left(\log_{1+\epsilon}\hat{d}-\log_{1+\epsilon}\left(\frac{\hat{d}}{c_{MN}\cdot k}\right)\right)\right)\\
&=&O(\log\log_{1+\epsilon}(c_{MN}\cdot k))\\
&=&O(\log\log k+\log\frac{1}{\epsilon})~.
\end{eqnarray*}
Lastly, finding the unique path between $u,v$ in the spanning tree of the home-cluster of $v$ takes linear time in the length of the output path. We conclude that the query time of $D$ is $O(\log\log k+\log\frac{1}{\epsilon})$.

Regarding the size of the oracle $D$, note that storing the Mendel-Naor distance oracle (item $3$ in the description of $D$) requires $O(n^{1+\frac{1}{k}})$ space. The variables that store the home-cluster for every vertex $v$ and cover $\mathcal{T}_{(1+\epsilon)^i}$ (item $1$) requires $O(n\cdot\lambda)=O(n\log_{1+\epsilon}\Lambda)$ space. Lastly, to store the pointers $q_T(\cdot)$ and the variables $h_T(\cdot)$ (item $2$), we need
\begin{eqnarray*}
O\left(\sum_{i=0}^{\lambda}\sum_{T\in\mathcal{T}_{(1+\epsilon)^i}}|T|\right)
&=&O\left(\sum_{i=0}^{\lambda}\sum_{v\in V}|\{T\in\mathcal{T}_{(1+\epsilon)^i}\;|\;v\in T\}|\right)\\
&\stackrel{(\ref{eq:NumOfClusters})}{=}&O\left(\sum_{i=0}^{\lambda}\sum_{v\in V}2k\cdot n^{\frac{1}{k}}\right)=O(k\cdot\log_{1+\epsilon}\Lambda\cdot n^{1+\frac{1}{k}})
\end{eqnarray*}
space. We get a total size of 
\[O(k\cdot\log_{1+\epsilon}\Lambda\cdot n^{1+\frac{1}{k}})=O(k\cdot\frac{\log\Lambda}{\epsilon}\cdot n^{1+\frac{1}{k}})\] 
for the oracle $D$.

Define $H$ as the sub-graph of $G$, which is the union of all the output paths from $D$. Note that $H$ is contained in the union of the spanning trees of every cluster of every neighborhood cover $\mathcal{T}_{(1+\epsilon)^i}$. The total size of $H$ is therefore
\[|H|\leq\sum_{i=0}^{\lambda}\sum_{T\in\mathcal{T}_{(1+\epsilon)^i}}|T|=O\left(k\cdot\frac{\log\Lambda}{\epsilon}\cdot n^{1+\frac{1}{k}}\right)~,\]
where the last step is justified by the same argument as in the size bound above.

We finally proved that $(H,D)$ is an interactive emulator with stretch $4(1+\epsilon)k$, size $O(k\cdot\frac{\log\Lambda}{\epsilon}\cdot n^{1+\frac{1}{k}})$, and query time $O(\log\log k+\log\frac{1}{\epsilon})$. By replacing $\epsilon$ by $\frac{\epsilon}{4}$, we get the assertion of the theorem.

\end{proof}

\pagebreak

\section{Appendix G} \label{sec:AppendixG}

In this appendix we specify the dependencies on the parameter $\epsilon$, that were omitted from Table \ref{table:InteractiveSpannersVariety} (where we assumed that $\epsilon$ was \textit{constant}).

\begin{table}[h]
\begin{center}
\begin{tabular}{|c|c|c|c|c|}
\hline
Stretch  & Size Coefficient  & Query Time  & Emulator & Distance \\ 
&&&& Preserver \\ \hline
$(4+\epsilon)k$  & $\frac{(\log\log k+\log\frac{1}{\epsilon})\cdot k\cdot\log\log n}{\epsilon\log n}$  & $\log k+\log\frac{1}{\epsilon}$ & \cite{TZ01,WN13} & Theorem \ref{thm:DistancePreserver1}\\ \hline
$(8+\epsilon)k$  & $\frac{k\cdot\log\log n}{\epsilon\log n}(\log\log k+\log\frac{1}{\epsilon}$ & $\log\log k+\log\frac{1}{\epsilon}$ & Theorem \ref{thm:APEmulator}, & Theorem \ref{thm:DistancePreserver1}\\
& $+\frac{\log\log\Lambda}{\log k})$ & $+\log^{(3)}\Lambda$ & based on \cite{AP90a} & \\ \hline
$(12+\epsilon)k$  & $\frac{\log\frac{1}{\epsilon}\cdot k\cdot\log\log n}{\epsilon\log n}$  & $\log k+\log\frac{1}{\epsilon}$ & \cite{TZ01,WN13} & Remark \ref{remark:PreserverWithWorseStretch}\\ \hline
$(24+\epsilon)k$  & $\frac{k\cdot\log\log n}{\epsilon\log n}(\log\frac{1}{\epsilon}+\frac{\log\log\Lambda}{\log k})$ & $\log\log k+\log\frac{1}{\epsilon}$ & Theorem \ref{thm:APEmulator} & Remark \ref{remark:PreserverWithWorseStretch}\\ 
& & $+\log^{(3)}\Lambda$ & & \\ \hline
$(16c_{MN}+\epsilon)k$  & $\frac{(\log\log k+\log\frac{1}{\epsilon})\cdot k\cdot\log\log n}{\epsilon\log n}$  & $\log\log k+\log\frac{1}{\epsilon}$ & Theorem \ref{thm:InteractiveEmulator} & Theorem \ref{thm:DistancePreserver1}\\ \hline
$(48c_{MN}+\epsilon)k$  & $\frac{(\log\frac{1}{\epsilon})\cdot k\cdot\log\log n}{\epsilon\log n}$  & $\log\log k+\log\frac{1}{\epsilon}$ & Theorem \ref{thm:InteractiveEmulator} & Remark \ref{remark:PreserverWithWorseStretch}\\ \hline
\end{tabular}
\end{center}
\begin{center}
\caption{A variety of results for interactive spanners by Lemma \ref{lemma:InteractiveSpannerVariety}, when using the interactive emulators and distance preservers from Table \ref{table:EmulatorsAndPreservers}. The $O$-notations are omitted. The column \textit{Size Coefficient} specifies the factor that multiplies $n^{1+\frac{1}{k}}$ in the size of each interactive spanner. More accurately, if the value of this column is $X$, then the size of the respective interactive spanner is $O(\lceil X\rceil\cdot n^{1+\frac{1}{k}}$).} \label{table:InteractiveSpannersVarietyDetails}
\end{center}
\end{table}

\def\arraystretch{1.5}

\end{document}